\documentclass[hidelinks,onefignum,onetabnum]{siamart250211}

\usepackage{graphicx}
\usepackage{amsmath}
\usepackage{amssymb}
\usepackage{tikz}
\usepackage{algorithm}
\usepackage{algpseudocode}
\usepackage{booktabs}
\usepackage[margin=1in]{geometry}
\usepackage[table]{xcolor} 
\usepackage{placeins} 
\usepackage{booktabs}      
\usepackage{tabularx}      
\usepackage{caption}       

\usepackage{lipsum}
\usepackage{amsfonts}
\usepackage{graphicx}
\usepackage{epstopdf}
\ifpdf
  \DeclareGraphicsExtensions{.eps,.pdf,.png,.jpg}
\else
  \DeclareGraphicsExtensions{.eps}
\fi

\newsiamremark{remark}{Remark}
\newsiamremark{hypothesis}{Hypothesis}
\crefname{hypothesis}{Hypothesis}{Hypotheses}
\newsiamthm{claim}{Claim}
\newsiamremark{fact}{Fact}
\crefname{fact}{Fact}{Facts}

\headers{Structured Analytic Mappings} {Wei Feng, Tengda Wei, and Haiyong Zheng}

\title{Structured Analytic Mappings for Point Set Registration\thanks{Submitted to the editors DATE.
}}

\author{%
Wei Feng\thanks{College of Electronic Engineering, Ocean University of China, Qingdao 266404, China 
  (\email{weifeng@stu.ouc.edu.cn}, \email{zhenghaiyong@ouc.edu.cn}).}
\and
Tengda Wei\thanks{School of Mathematics and Statistics, Shandong Normal University, Ji'nan 250014, China 
  (\email{tdwei123@sdnu.edu.cn}).}
\and
Haiyong Zheng\footnotemark[2]\hspace{.5em}\thanks{Corresponding author.}
}

\usepackage{amsopn}

\makeatletter
\DeclareFontShape{OT1}{cmr}{bx}{sc}{<-> ssub * cmr/m/sc}{}
\makeatother
\definecolor{myGreen}{HTML}{33FF00}
\definecolor{myRed}{HTML}{FF3030}
\definecolor{myGrey}{HTML}{AA5555}
\definecolor{myWhite}{HTML}{FFFFFF}
\definecolor{josef}{HTML}{FF3030}

\definecolor{maroon}{RGB}{128, 0, 0} 
\definecolor{petr}{RGB}{200, 30, 30} 

\newcommand{\figref}[1]{Fig.~\ref{#1}}
\newcommand{\tabref}[1]{Table~\ref{#1}}
\newcommand{\equref}[1]{Equ.~\ref{#1}}

\renewcommand{\algorithmicrequire}{\textbf{Input:}}
\renewcommand{\algorithmicensure}{\textbf{Output:}}
\ifpdf
\hypersetup{
  pdftitle={Structured Analytic Mappings for Point Set Registration},
  pdfauthor={Wei Feng, Tengda Wei, and Haiyong Zheng}
}

\begin{document}
\maketitle

\begin{abstract}
We present an analytic approximation model for non-rigid point set registration, grounded in the multivariate Taylor expansion of vector-valued functions. By exploiting the algebraic structure of Taylor expansions, we construct a structured function space spanned by truncated basis terms, allowing smooth deformations to be represented with low complexity and explicit form. To estimate mappings within this space, we develop a quasi-Newton optimization algorithm that progressively lifts the identity map into higher-order analytic forms. This structured framework unifies rigid, affine, and nonlinear deformations under a single closed-form formulation, without relying on kernel functions or high-dimensional parameterizations.
The proposed model is embedded into a standard ICP loop---using (by default) nearest-neighbor correspondences---resulting in Analytic-ICP, an efficient registration algorithm with quasi-linear time complexity. Experiments on 2D and 3D datasets demonstrate that Analytic-ICP achieves higher accuracy and faster convergence than classical methods such as CPD and TPS-RPM, particularly for small and smooth deformations.
\end{abstract}


\section{Introduction}

Point set registration is a fundamental problem in computer vision and geometric data analysis, with wide-ranging applications in 3D reconstruction~\cite{2022NIMBLE}, sensor alignment~\cite{2022Survey}, motion tracking~\cite{IOSIFESCU199713}, and robotic perception~\cite{pomerleau2015review}. While rigid and affine methods have been extensively studied~\cite{ICP1992,RICP,yang2015go}, registering point sets under smooth non-rigid deformations remains a significant challenge due to the need for expressive mappings in high-dimensional spaces. Such techniques are crucial in modern systems ranging from visual metrology~\cite{2022NIMBLE} to augmented reality~\cite{1634322} and autonomous navigation~\cite{geiger2012we}.

Most existing non-rigid registration methods adopt non-parametric models, such as radial basis functions~\cite{TPS–RPM}, Gaussian mixtures~\cite{CPD}, or statistical estimation frameworks~\cite{jian2010robust,BCPD}. While flexible, these models typically require large numbers of basis functions or parameters to express complex deformations, resulting in high computational complexity and increased risk of overfitting~\cite{ma2017non,ogundare2018understanding,tam2012registration}. To ensure numerical stability, additional regularization is often introduced, which may obscure geometric structure and hinder interpretability.

From a mathematical perspective, multivariate Taylor expansions offer a classical framework for approximating smooth functions via structured polynomial bases. In particular, Vetter's work~\cite{vetter1973matrix} provides a general formulation for matrix- and vector-valued Taylor expansions, enabling analytic mappings to be expressed with controllable complexity and explicit algebraic structure. Despite its theoretical appeal, this formulation remains underexplored in computer vision and registration tasks.

In this paper, we present the first application of multivariate Taylor expansion of vector-valued functions to non-rigid point set registration. By constructing a low-dimensional admissible function space spanned by truncated Taylor basis terms, we provide a compact and expressive framework capable of representing smooth deformations without relying on kernel-based models or dense parameter sets.

To illustrate the expressive capacity of low-order analytic expansions, we show in Fig.~\ref{fig:fistdistort} a series of 2D deformations generated from a second-order Taylor model with only three parameters. Despite its simplicity, the model can produce a rich variety of nonlinear and smooth transformations.

\begin{figure}[htbp]
  \centering

  \begin{minipage}{0.95\linewidth}
    \centering
    {\scriptsize
    Analytic mapping: 
    $
    \begin{bmatrix}
    x_{1} \\ x_{2}
    \end{bmatrix} =
    \begin{bmatrix}
    0 \\ 0
    \end{bmatrix} +
    \begin{bmatrix}
    1 & a_{1} \\ 0 & 1
    \end{bmatrix}
    \begin{bmatrix}
    y_{1} \\ y_{2}
    \end{bmatrix} +
    \frac{1}{2!}
    \begin{bmatrix}
    a_{2} & 0 & 0 \\
    0 & a_{3} & 0
    \end{bmatrix}
    \begin{bmatrix}
    y_{1}^{2} \\ 2y_{1}y_{2} \\ y_{2}^{2}
    \end{bmatrix}
    $
    }
    \\
    \vspace{1mm}
    {\scriptsize Effect of parameter $(a_{1}, a_{2}, a_{3})$ perturbation in 2D second-order Taylor expansion.}
  \end{minipage}

  \vspace{2mm}

  \begin{minipage}{0.23\linewidth}
    \centering
    \includegraphics[width=0.85\linewidth]{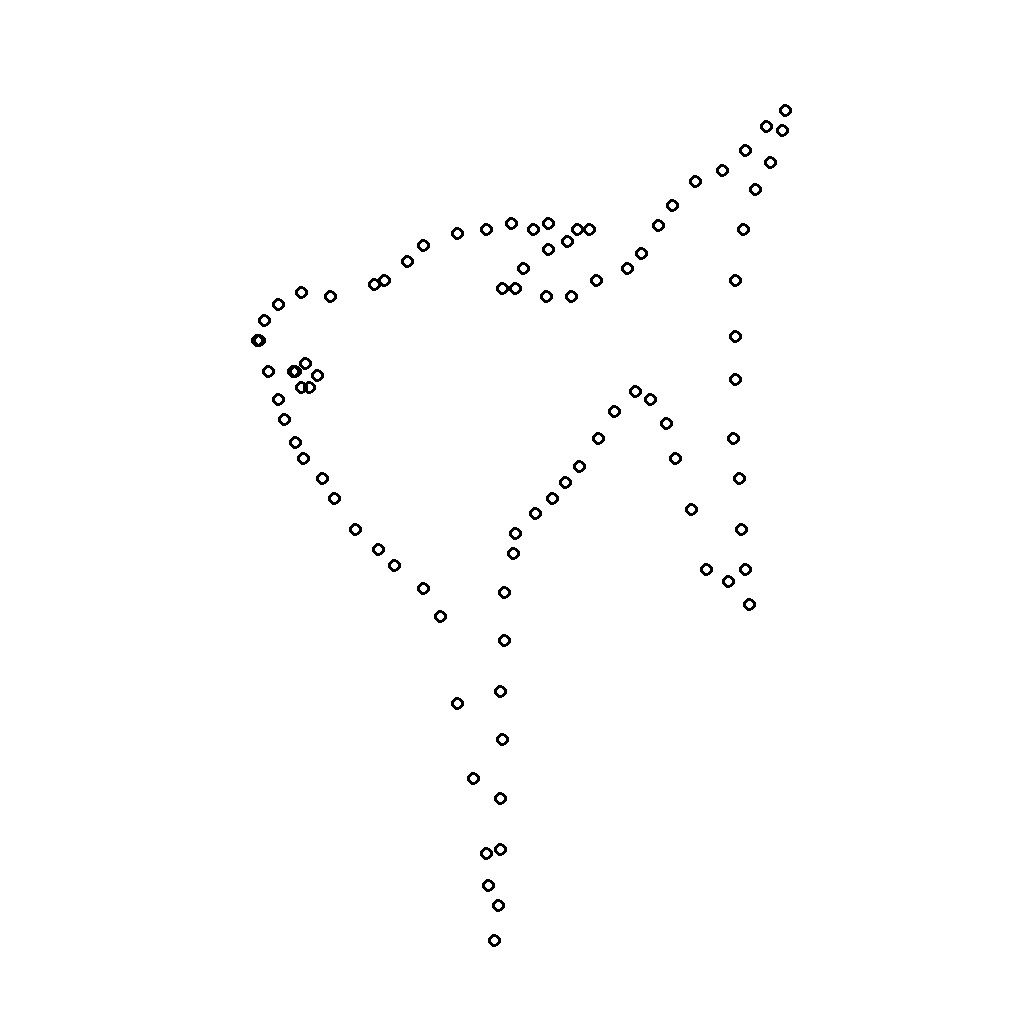}\\
    {\scriptsize (0, 0, 0)}
  \end{minipage}
  \begin{minipage}{0.23\linewidth}
    \centering
    \includegraphics[width=0.85\linewidth]{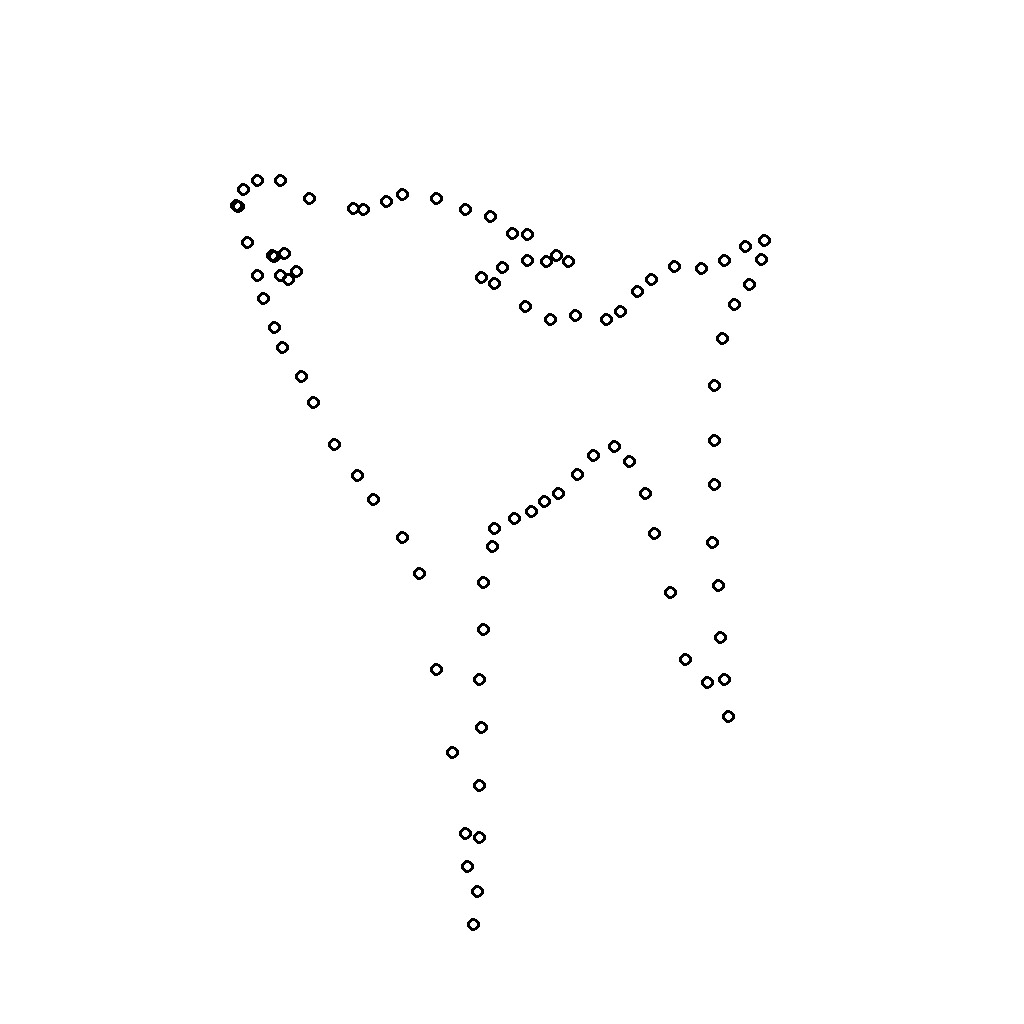}\\
    {\scriptsize (0.5, 0, 0)}
  \end{minipage}
  \begin{minipage}{0.23\linewidth}
    \centering
    \includegraphics[width=0.85\linewidth]{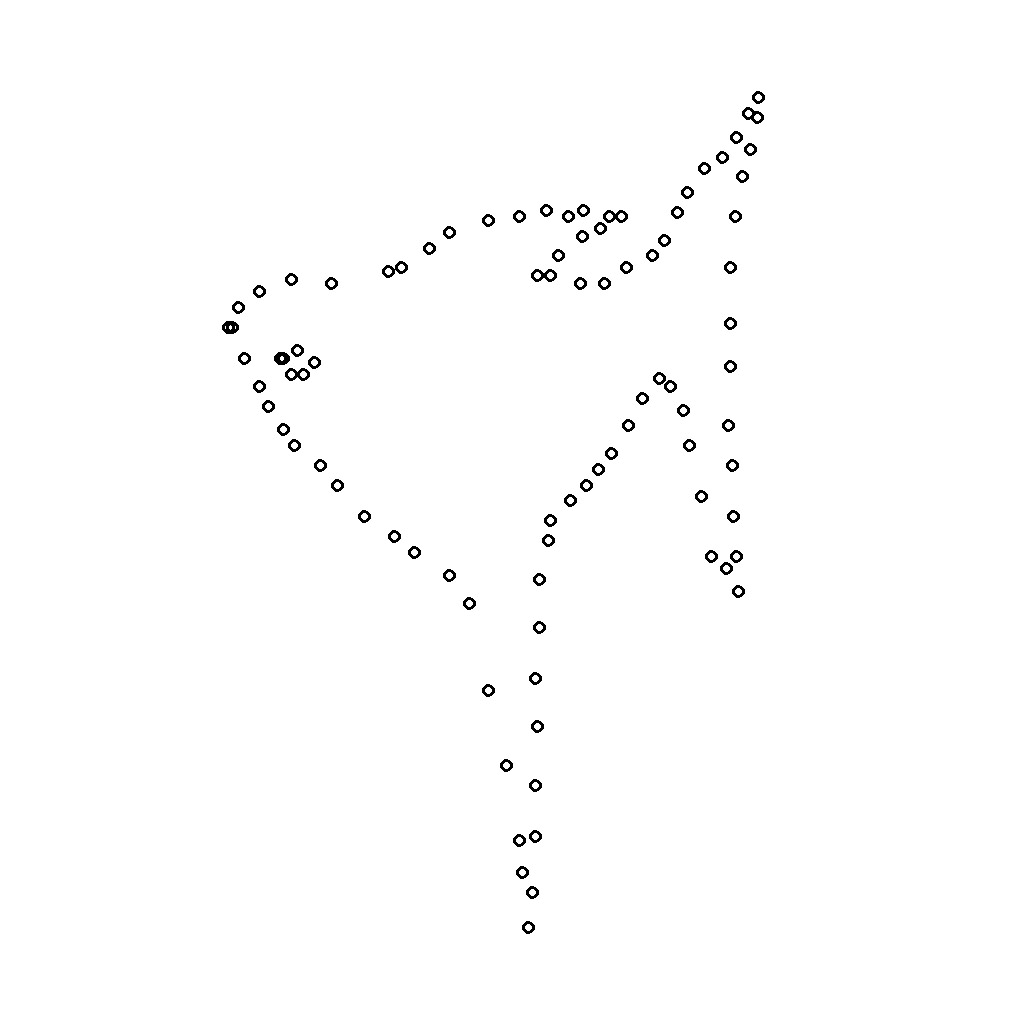}\\
    {\scriptsize (0, 0.9, 0)}
  \end{minipage}
  \begin{minipage}{0.23\linewidth}
    \centering
    \includegraphics[width=0.85\linewidth]{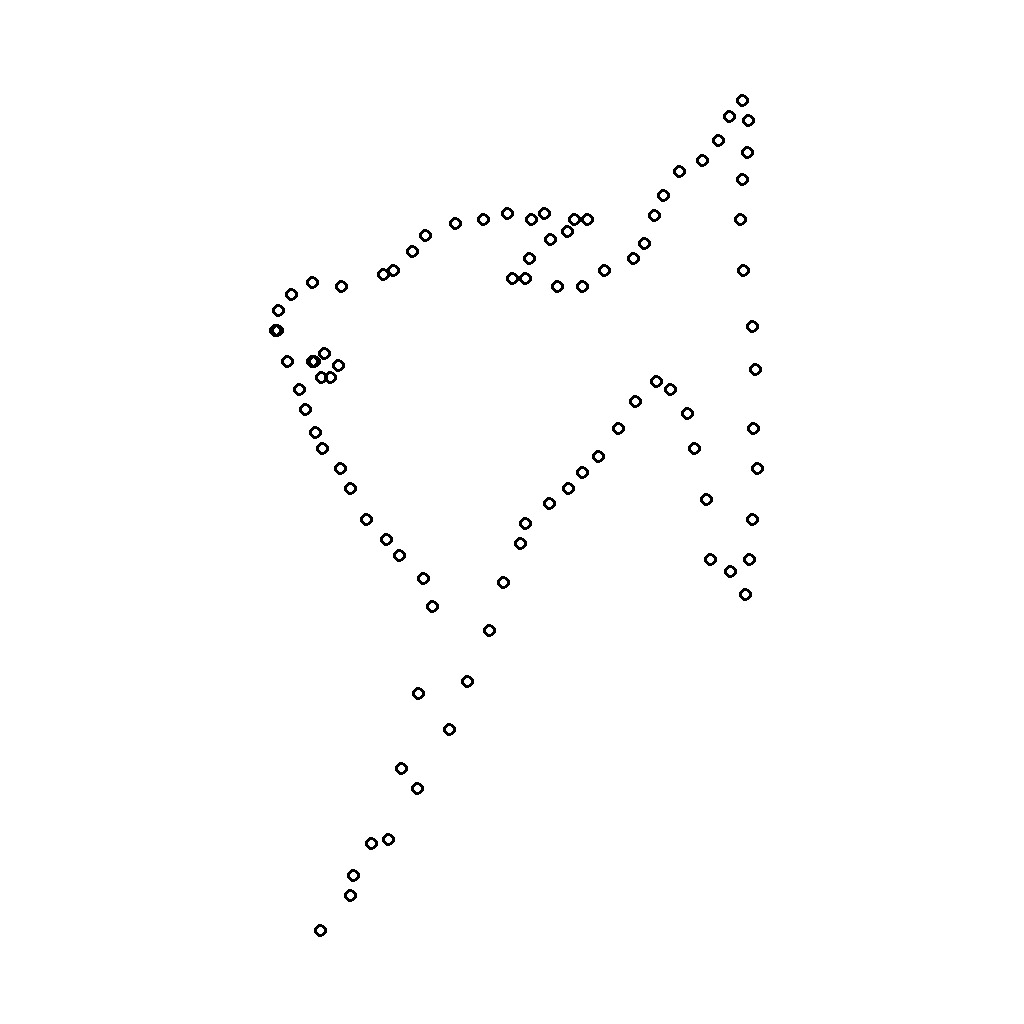}\\
    {\scriptsize (0, 0, 0.7)}
  \end{minipage}

  \vspace{1mm}

  \begin{minipage}{0.23\linewidth}
    \centering
    \includegraphics[width=0.85\linewidth]{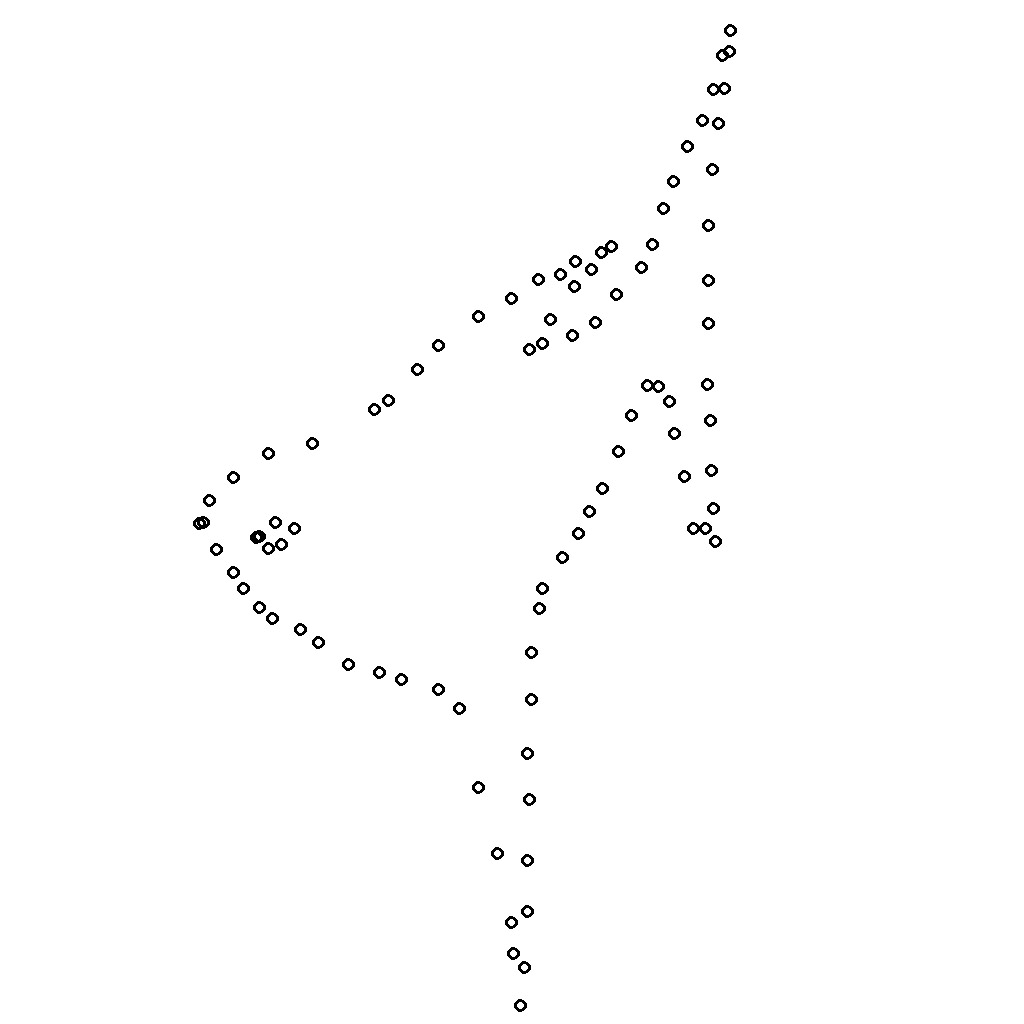}\\
    {\scriptsize (-0.5, -1.2, 0)}
  \end{minipage}
  \begin{minipage}{0.23\linewidth}
    \centering
    \includegraphics[width=0.85\linewidth]{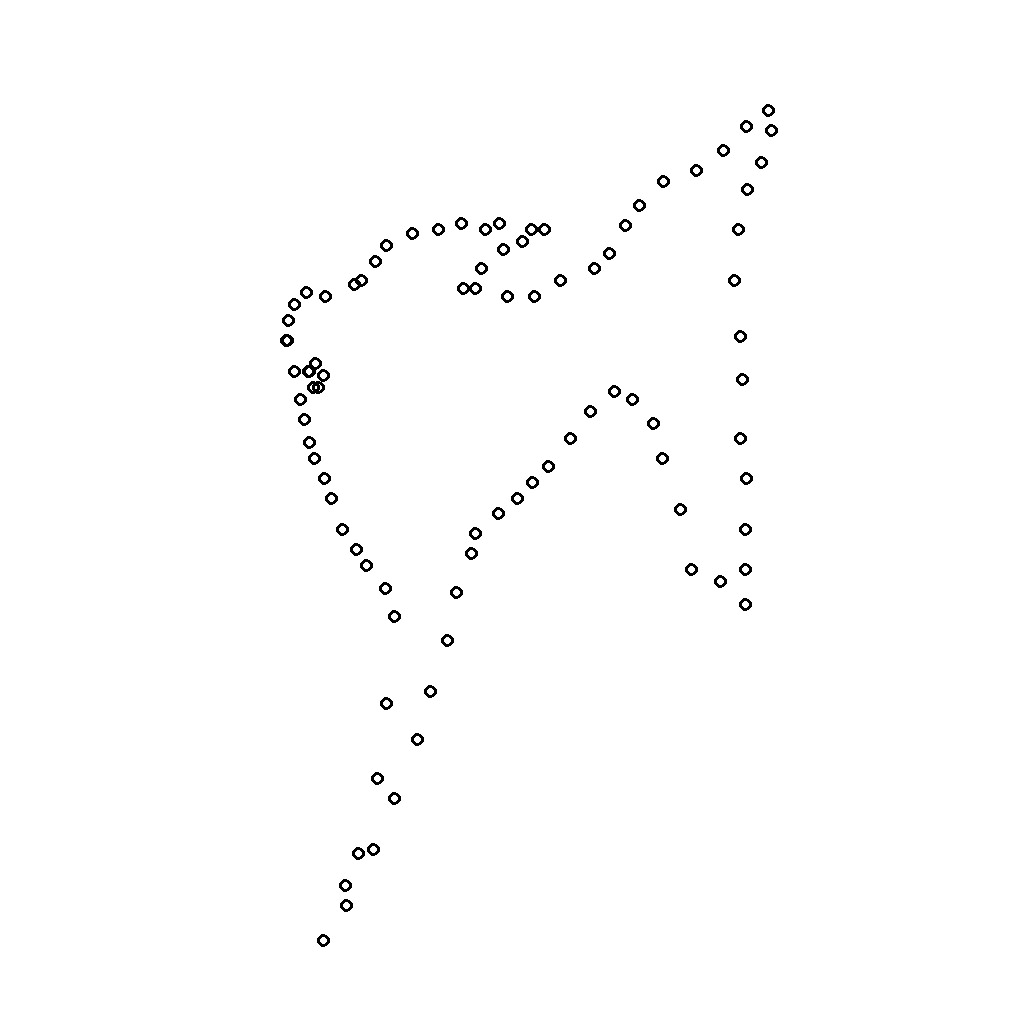}\\
    {\scriptsize (0, 0.9, -0.5)}
  \end{minipage}
  \begin{minipage}{0.23\linewidth}
    \centering
    \includegraphics[width=0.85\linewidth]{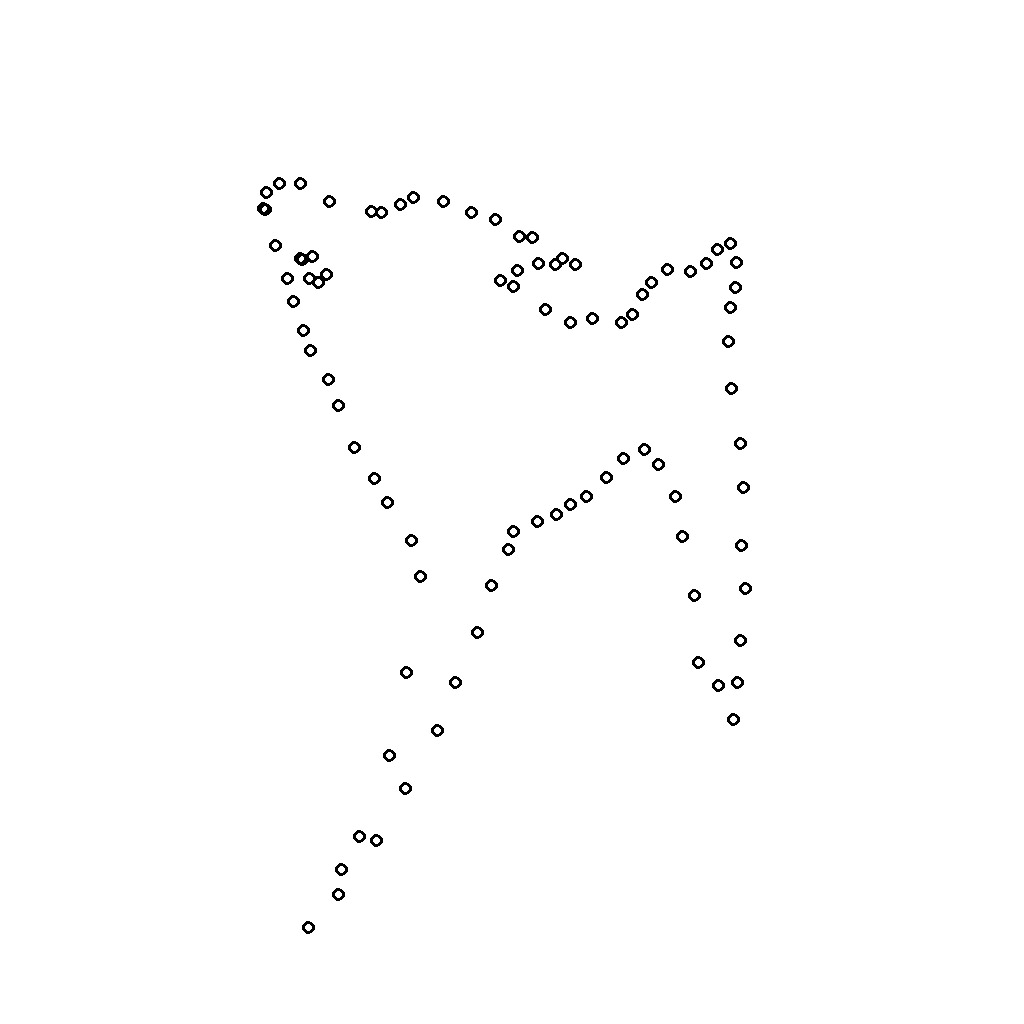}\\
    {\scriptsize (0.5, 0, -0.7)}
  \end{minipage}
  \begin{minipage}{0.23\linewidth}
    \centering
    \includegraphics[width=0.85\linewidth]{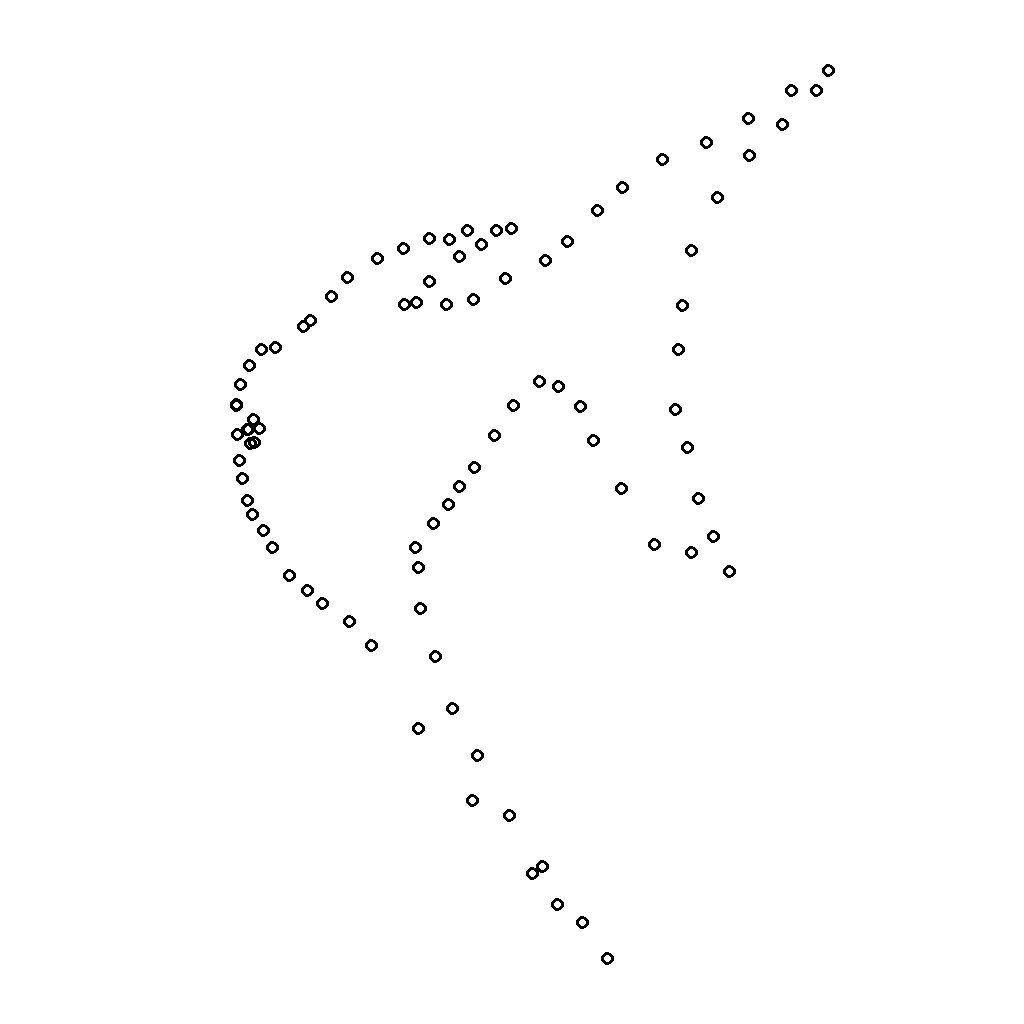}\\
    {\scriptsize (-0.2, 1.2, 0.8)}
  \end{minipage}

  \caption{\textbf{Illustrative effect of second-order analytic deformation.} Even with only three tunable coefficients $(a_1, a_2, a_3)$, the second-order Taylor expansion model produces a diverse range of smooth, nonlinear transformations.}
  \label{fig:fistdistort}
\end{figure}

Based on this representation, we design a progressive fitting algorithm that lifts the identity map to higher-order analytic forms via structured quasi-Newton optimization. This method is integrated into a standard ICP loop---with (by default) nearest-neighbor correspondences---to construct \emph{Analytic-ICP}, a fast and robust registration algorithm that supports hierarchical refinement of rigid, affine, and nonlinear mappings.

\vspace{1mm}
\noindent
\textbf{Our main contributions are summarized as follows:}
\begin{itemize}
    \item We present the first structured approximation framework based on multivariate Taylor expansions of vector-valued functions in computer vision. This formulation enables interpretable and hierarchical modeling of smooth deformations, and provides a principled foundation for registration and beyond.

    \item We formalize a matrix–vector representation of analytic mappings using \textit{generalized derivative matrices} and \textit{generalized monomial vectors}, and prove a \textbf{Structured Approximation Theorem} establishing the approximation capability of truncated Taylor models over smooth mappings.

    \item We construct a low-dimensional function space spanned by truncated Taylor basis terms, unifying rigid, affine, and nonlinear transformations within a single closed-form expression.

    \item We design a progressive quasi-Newton fitting algorithm that incrementally lifts the identity map into higher-order approximants, offering controllable model complexity and efficient convergence.

    \item We integrate the proposed framework into an ICP-style optimization routine. The resulting \textit{Analytic-ICP} algorithm achieves quasi-linear complexity and outperforms classical non-rigid methods (e.g., CPD, TPS-RPM) in both accuracy and efficiency on 2D and 3D registration tasks.
\end{itemize}

To facilitate reproducibility, source code and data are available; see Section~\ref{sec:code} for details.

\section{Related Work}
We propose a novel analytic mapping model that can be integrated into many unsupervised registration frameworks to improve their performance. In this section, we briefly review the development of rigid, affine, and non-rigid registration techniques that are relevant to our method.

\subsection{Rigid and Affine Registration}

The Iterative Closest Point (ICP) algorithm~\cite{TPS–RPM} is one of the most fundamental and widely used geometric registration methods. It alternates between estimating correspondences and fitting orthogonal transformations. Due to its simplicity and efficiency, ICP has become the de facto standard for rigid registration in industrial applications~\cite{maiseli2017recent,vizzo2023kiss}, and has inspired numerous variants~\cite{RICP,zhang2021fast,segal2009generalized,low2004linear,lv2023kss,yang2015go,trimmedICP}.

Affine registration~\cite{affine-icp,ho2007new,jenkinson2001global,du2008affine} extends the rigid model by allowing linear deformations. It preserves the simplicity of parameterization and does not require additional constraints, making it highly practical. The number of parameters involved is only the square of that in rigid registration.

\subsection{Non-Rigid Registration and Statistical Frameworks}

The Non-Rigid ICP (N-ICP)~\cite{NICP2007} was among the first to extend the ICP framework to non-rigid settings by modeling affine transformations over smooth manifolds. However, its computational complexity is significantly higher than rigid ICP, which limits its practical use.

Differential geometry-based methods~\cite{rusinkiewicz2019symmetric,anderson2022delaunay,fahandezh2020robust,zampogiannis2019topology,yao2023fast,2010Non} have made considerable progress by leveraging intrinsic surface properties and geometric feature spaces. While these methods achieve high robustness and accuracy, they typically rely on 3D reconstruction and are challenging to implement in real-world applications.

In practice, registration algorithms based on one-to-one correspondences, such as ICP, are highly efficient. However, due to the non-convexity of the mapping support and the lack of closedness in the mapping space~\cite{villani2009optimal}, such methods struggle to achieve optimal solutions in non-rigid scenarios.

To address these issues, a large class of softassign-based methods has been developed. Notably, the RPM~\cite{TPS–RPM} and EM~\cite{1998Graph} frameworks relax the hard correspondence constraint and introduce probabilistic models to improve convergence. By using Gaussian Mixture Models (GMM), these approaches reframe point set registration as a statistical estimation problem~\cite{chui2000feature}.

The TPS-RPM algorithm~\cite{TPS–RPM} combines the robustness of the TPS model~\cite{bookstein1989principal} with the annealing strategy of RPM. It is particularly effective for handling noise and outliers in 2D and 3D data, but does not generalize well to higher-dimensional or more structured deformation problems.

The Coherent Point Drift (CPD) algorithm~\cite{CPD2006,CPD} builds upon Motion Coherence Theory and models the moving point set as a GMM. It uses EM to estimate soft correspondences but lacks a unified or interpretable form for the transport mapping, which limits its analytical transparency and generalization ability.

Further GMM-based approaches~\cite{jian2005robust,jian2010robust} attempt to unify rigid and non-rigid registration by minimizing the $L^2$ distance between two GMMs. Bayesian extensions such as BCPD~\cite{BCPD}, based on Variational Bayesian inference~\cite{gelman1995bayesian}, improve convergence and interpretability but handle rigid and non-rigid cases redundantly, reducing overall efficiency. Accelerated variants like BCPD++~\cite{hirose2020acceleration} use downsampling for speed, but this compromises geometric fidelity and algorithmic robustness.

Other correspondence-free methods~\cite{tsin2004correlation,zhou2016fast} have also been proposed to relax constraints even further, but often at the cost of increasing algorithmic complexity.

\subsection{Analytic Approximation Models}

In non-rigid registration, statistical frameworks dominate due to their modeling flexibility. However, their reliance on dense basis functions, high parameter dimensionality, and computational overhead limits their practicality in time-sensitive or resource-constrained settings. Despite the efficiency advantages of geometry-based methods, few successful non-rigid frameworks follow this path.

In this work, we propose a closed-form, low-dimensional analytic mapping model that approximates smooth deformations via truncated Taylor expansions. Unlike existing basis function methods, our model requires no reproducing kernels or large parameter spaces, and achieves strong approximation capacity with minimal computational cost. Our approach combines the strengths of rigid efficiency and non-rigid flexibility, offering a new perspective for geometric registration.

\section{Overview of Structured Analytic Mappings}
\label{sec:overview}

We first give a high-level view of our structured analytic mapping framework; all mathematical details are deferred to Sec.~\ref{sec:main} and subsequent subsections. Each outer iteration alternates between (i) a \emph{correspondence update} $C^{(t)}$ (nearest neighbor as in ICP by default; optionally soft assignments in the style of RPM/CPD) and (ii) \emph{staged fitting} of the map. After the rigid and affine stages, the pipeline takes a \textbf{mutually exclusive branch}: either a restricted projective refinement (when imaging/planarity effects dominate), or a structured Taylor lifting for smooth non-rigid deformation. Coefficients are estimated via least-squares / quasi-Newton updates (described later). The loop stops when a standard criterion is met (e.g., $\Delta\mathrm{RMSE}<\varepsilon$, parameter-update norm $\|\Delta\theta\|<\eta$, or a maximal order/iteration cap); otherwise, on the Taylor branch we increase the order $m_{t+1}\!\gets m_t{+}1$ and repeat. The overall workflow is summarized in Fig.~\ref{fig:analytic_mapping_flow}.

\begin{figure*}[t]
  \centering
  \includegraphics[width=\textwidth]{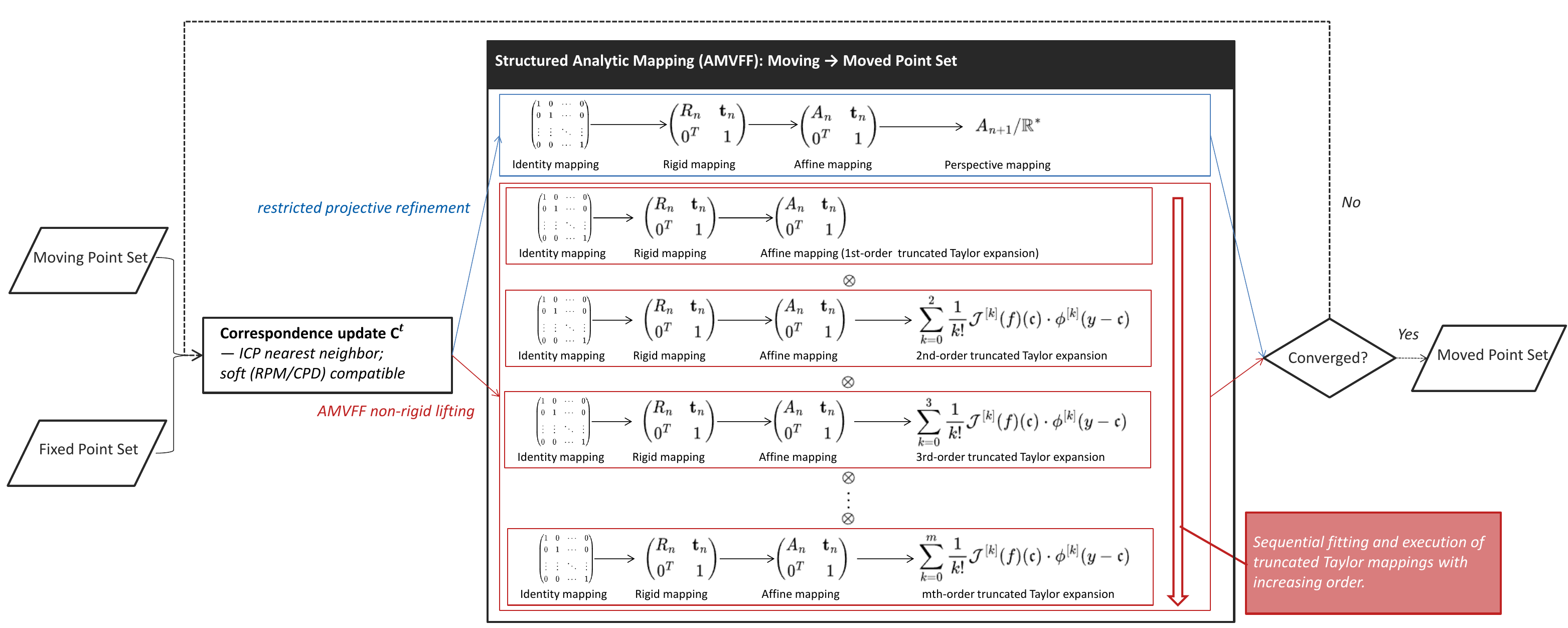}
  \vspace{-2mm}
  \caption{\textbf{Pipeline of the structured analytic mapping approach.}
At iteration $t$, we form correspondences $C^{(t)}$ and fit in stages:
rigid $\mathrm{SE}(n,\mathbb{R}) \rightarrow$ affine (residual) $\mathrm{AGL}(n,\mathbb{R})/\mathrm{SE}(n,\mathbb{R})$,
then take a \textbf{mutually exclusive branch}:
restricted projective $\mathrm{PGL}(n{+}1,\mathbb{R})/\mathrm{AGL}(n,\mathbb{R})$ for planar/homography–dominant cases,
\emph{or} structured Taylor lifting of order $m_t$ for smooth non-rigid deformation
(if the stop test fails, increase $m_{t+1}\!\leftarrow m_t{+}1$).
The diamond marks the stopping rule ($\Delta\mathrm{RMSE}$ / $\|\Delta\theta\|$ / order / iteration).
A negative test returns to correspondence (outer ICP-style loop); a positive test outputs the moved point set and the final mapping~$\tau$.
When the correspondence step uses nearest-neighbor updates (ICP-style), we refer to this instantiation as \emph{Analytic-ICP};
the coefficient fitting routine is detailed later (\emph{AMVFF}).}

  \label{fig:analytic_mapping_flow}
  \vspace{-2mm}
\end{figure*}

\section{Analytic Approximation of Smooth Mappings via Taylor Expansion}
\label{sec:main}
\subsection{Motivation: Analytic Approximation of Smooth Deformations}
Many naturally occurring deformation processes—such as those found in fluid motion, elastic distortion, or optical systems—are governed by partial differential equations with analytic solutions under mild regularity conditions~\cite{arnold1989mathematical,john1971partial}. This observation motivates the use of analytic functions as structural priors for modeling smooth transformations in vision and geometry-related tasks.

From a theoretical standpoint, the real analytic functions form a strict yet expressive subset of smooth functions. Classical approximation theory asserts that analytic functions are dense in the space of smooth functions on compact domains, either in the topology of uniform convergence of derivatives or in suitable Sobolev norms~\cite{krantz2002primer,hormander1983analysis}. That is, any smooth mapping can be approximated arbitrarily well by an analytic mapping, provided the domain is bounded.

These considerations establish a solid foundation for analytic approximation in registration tasks. Instead of relying on data-driven basis expansions (e.g., radial basis functions or kernels), we propose to directly construct analytic mappings via truncated Taylor expansions of multivariate vector-valued functions. This approach yields a low-dimensional yet expressive function space with a clear geometric structure and tractable optimization properties.

\subsection{Analytic Mapping Model via Taylor Approximation}\label{sec:deduction}

To formulate our analytic deformation model, we construct a low-dimensional function space based on multivariate Taylor expansions of vector-valued functions. This section introduces the notation used throughout the formulation and outlines the algebraic structure of the proposed mapping.

\begin{itemize}
  \item  $\left \| \cdot  \right \|$ --- the Euclidean norm.
  \item  $x=\left (x_{1}, \dots, x_{n}  \right )$, $y=\left ( y_{1}, \dots, y_{n} \right )$, $\tau y$ --- the $n$-dimensional fixed point, the $n$-dimensional moving point, the $n$-dimensional moved point.
  \item  $X$, $Y$ --- the fixed point set, the moving point set.
  \item $\tau_{\mathrm{rig}},\tau_{\mathrm{aff}},\tau_{\mathrm{proj}}$ --- rigid, affine, and projective transforms.
  \item $\widehat{Y}_{\ast}=\tau_{\ast}(Y)$ --- transformed point set under stage $\ast\in\{\mathrm{rig},\mathrm{aff},\mathrm{proj}\}$.
\item $k$, $\tau^{(k)}$, $\widehat{Y}^{(k)}$ --- truncation order, $k$-th order model, $\widehat{Y}^{(k)}=\tau^{(k)}(Y)$.
  \item $SO\left ( n,\mathbb{R} \right )$, $SE\left ( n,\mathbb{R} \right )$, $AGL\left ( n,\mathbb{R} \right )$, $PGL\left ( n,\mathbb{R} \right )$ --- the special orthogonal group, the special Euclidean group, the affine general linear group, the projective linear group, where $n$ denotes dimension, $\mathbb{R}$ denotes real number field.
  \item  $\arg\min_{\tau} \left \| \cdot \right \| ^{2} $ --- a functional minimization problem where $\tau$ represents the mapping is to find the value of $\tau$ that minimizes the square of the Euclidean norm of an expression.
  \item  $x_{m,i}$, $y_{m,i}^{n} $ --- the $i$-th component of the $m$-th point in the fixed point set, the $n$-th power of the $i$-th component of the $m$-th point in the moving point set.
  \item $R_m$ --- the Taylor remainder after truncation at order $m$.
  \item  $\Re$, $\mathbf t$, $R$, $\Lambda$ --- Rotation matrix, translation vector, R component of QR decomposition, orthogonal complement of affine transformation in perspective transformation.
  \item  $a_{i}^{o}$ --- the initial value of the $i$-th parameter.
  \item  $L$, $\hat{v}$, $\bar{\nu}$, $\delta a$ --- the observed vector, the real-time residual vector, the threshold of residual vector, adjustment value of parameter $a$.  
  \item  $\inf_{x}\left \{ M \left ( x \right )   \right \}$ --- The value of the variable $x$ when taking the infimum of $M$.
  \item  $N_{r}$, $p$ --- the number of corresponding point pairs, the number of parameters to be fitted.
  \item  $\binom{a}{b}$  ---  the binomial coefficient (i.e., the number of $b$-combinations from a set of $a$ elements), representing the count of basis monomials at order $b$.
\end{itemize}

We begin by generalizing the multivariate Taylor expansion to describe smooth vector-valued function spaces, which serve as the foundation for modeling deformations in point set registration. Our goal is to construct low-complexity analytic mappings that approximate smooth distortions between point sets.

Specifically, we consider deformations represented by diffeomorphisms of the form:
\begin{equation}\label{eq:deform}
    X_e = F(X_b),
\end{equation}
where \( X_b \) denotes the reference (or template) point set and \( X_e \) is the deformed point set after applying a smooth transformation \( F \).

In this paper, we focus exclusively on smooth (i.e., differentiable and invertible) deformations. Examples include those induced by optical distortions, elastic transformations, or incompressible flows. Discontinuous or abrupt transformations, such as occlusions or topological mutations, are considered out of scope and treated as structural noise or defects.

Our proposed mapping model is grounded in the multivariate Taylor expansion of vector-valued functions, and can be regarded as an instance of the classical Taylor theory of matrix-valued functions~\cite{vetter1973matrix}. To the best of our knowledge, this is the first application of such a formulation in the fields of computer vision and machine learning.

Within the domain of registration, the model provides a unified framework that captures rigid, affine, perspective, and non-rigid deformations using a single elementary expression. Unlike prior models that rely on kernel-based representations or heuristic constructions, our approach derives directly from foundational mathematics—without introducing structural redundancy.

As a strict representation of analytic functions, the Taylor expansion offers rich mathematical properties, including global approximation guarantees and natural extensibility to the complex domain. Compared with polynomial kernel methods~\cite{saitoh2016theory}, which primarily serve as similarity measures in reproducing kernel Hilbert spaces, our formulation defines an explicit global mapping with interpretable geometric structure. This distinguishes our model from existing techniques in both theory and practice.

In point set registration, both the source and target point sets reside in a Euclidean coordinate space, which forms a free module over \( \mathbb{R} \). This enables us to naturally model the registration mapping as a multivariate vector-valued function \( \tau : \mathbb{R}^d \to \mathbb{R}^d \), regardless of whether the basis is orthogonal.

For implementation convenience, we assume the coordinate axes to be orthogonal and origin-centered, but this is not a theoretical requirement—only a practical one. Our approximation model remains valid for any coordinate frame defined by a linearly independent basis.

We consider the registration mapping as a multivariate vector-valued function defined from the moving space \( W \subset \mathbb{R}^2 \) into a target coordinate space structured as \( H \times V \), where \( H, V \subset \mathbb{R} \) represent the horizontal and vertical coordinate spaces of the fixed point set, respectively. In this formulation, \( H \times V \subset \mathbb{R}^2 \) serves as a representation of the 2D point space, and the mapping \( \tau : W \to H \times V \) captures the geometric transformation between the two point spaces.

This formulation yields a bijective decomposition of the vector-valued registration mapping into two scalar-valued functions, each responsible for one spatial dimension:
\begin{eqnarray}\label{eq:W2XY}
\begin{aligned}
W \to H \times V \Rightarrow \left\{
\begin{matrix}
 W \to H \\
 W \to V
\end{matrix}
\right.
\end{aligned}.
\end{eqnarray}

To formalize the registration objective, we consider the problem of estimating the transformation \( \tau \) that minimizes the squared distance between the fixed and transformed moving points. For a pair of corresponding 2D points \( (x_1, x_2) \in H \times V \) and \( (y_1, y_2) \in W \), this can be written as:
\begin{eqnarray}\label{eq:xy2uv}
\begin{aligned}
\arg\min_{\tau}\left \| \left (x_{1} ,x_{2}  \right )- \tau \left ( y_{1} ,y_{2} \right )  \right \| ^{2}.
\end{aligned}
\end{eqnarray}

By combining the bijective decomposition in \eqref{eq:W2XY} with the objective above, we express the vector-valued mapping \( \tau \) in terms of its scalar components: \( \tau = (f_1, f_2) \). This reformulates the optimization problem as two independent scalar regression tasks:
\begin{eqnarray}\label{eq:x2uy2v}
\begin{aligned}
\left\{
\begin{matrix}
\arg\min_{f_{1}}\left |x_{1}- f_{1} \left ( y_{1} ,y_{2}  \right )  \right |^{2},     \\
\arg\min_{f_{2}} \left |x_{2}- f_{2} \left ( y_{1} ,y_{2}  \right )  \right | ^{2}.  
\end{matrix}
\right.
\end{aligned}
\end{eqnarray}

In the following section, we introduce a family of smooth approximators for the functions \( f_1 \) and \( f_2 \), based on multivariate Taylor expansions, which enable compact and flexible modeling of analytic mappings.

\subsection{Analytic Function Space Based on Multivariate Taylor Expansion}

According to Taylor’s theorem for functions of two independent variables, we expand \( f_1 \) and \( f_2 \) individually around a reference point \( \mathfrak{c} \). These scalar-valued expansions are then reorganized into a unified matrix-vector form to describe the vector-valued deformation mapping \( \tau(y) \). The resulting analytic mapping takes the following form:

\begin{eqnarray}\label{eq:analytic_mapping_2d}
\begin{aligned}
\tau(y) &= 
\begin{bmatrix}
 f_{1}(\mathfrak{c}) \\ f_{2}(\mathfrak{c})
\end{bmatrix}
+
\begin{bmatrix}
\frac{\partial f_{1} }{\partial y_{1}}(\mathfrak{c}) & \frac{\partial f_{1} }{\partial y_{2}}(\mathfrak{c})\\ 
\frac{\partial f_{2} }{\partial y_{1}}(\mathfrak{c}) & \frac{\partial f_{2} }{\partial y_{2}}(\mathfrak{c})
\end{bmatrix}
\begin{bmatrix}
y_{1} - \mathfrak{c}_{1} \\
y_{2} - \mathfrak{c}_{2}
\end{bmatrix} \\
&+
\frac{1}{2}
\begin{bmatrix}
\frac{\partial^{2} f_{1} }{\partial y_{1}^{2}}(\mathfrak{c}) & \frac{\partial^{2} f_{1} }{\partial y_{1} \partial y_{2}}(\mathfrak{c}) & \frac{\partial^{2} f_{1} }{\partial y_{2}^{2}}(\mathfrak{c}) \\
\frac{\partial^{2} f_{2} }{\partial y_{1}^{2}}(\mathfrak{c}) & \frac{\partial^{2} f_{2} }{\partial y_{1} \partial y_{2}}(\mathfrak{c}) & \frac{\partial^{2} f_{2} }{\partial y_{2}^{2}}(\mathfrak{c})
\end{bmatrix}
\begin{bmatrix}
(y_{1} - \mathfrak{c}_{1})^{2} \\
2(y_{1} - \mathfrak{c}_{1})(y_{2} - \mathfrak{c}_{2}) \\
(y_{2} - \mathfrak{c}_{2})^{2}
\end{bmatrix} \\
&+ \dots + \frac{1}{n!}
\begin{bmatrix}
  \frac{\partial^{n} f_{1} }{\partial y_{1}^{n}}\left ( \mathfrak{c}\right ) &\dots   & \frac{\partial^{n} f_{1} }{\partial y_{1}^{n-k}\partial y_{2}^{k}}\left ( \mathfrak{c}\right ) &\dots   & \frac{\partial^{n} f_{1} }{\partial y_{2}^{n}}\left ( \mathfrak{c}\right )\\
  \frac{\partial^{n} f_{2} }{\partial y_{1}^{n}}\left ( \mathfrak{c}\right ) &\dots   & \frac{\partial^{n} f_{2} }{\partial y_{1}^{n-k}\partial y_{2}^{k}}\left ( \mathfrak{c}\right ) &\dots   & \frac{\partial^{n} f_{2} }{\partial y_{2}^{n}}\left ( \mathfrak{c}\right )
\end{bmatrix}
\begin{bmatrix}
(y_{1} - \mathfrak{c}_{1})^{n} \\
\vdots \\
\binom{n}{k}(y_{1} - \mathfrak{c}_{1})^{n-k}(y_{2} - \mathfrak{c}_{2})^{k} \\
\vdots \\
(y_{2} - \mathfrak{c}_{2})^{n}
\end{bmatrix}
+ R_n,
\end{aligned}
\end{eqnarray}

where \( R_n \) denotes the Taylor remainder term. This representation allows us to express analytic mappings as compact, parameterized structures that are both mathematically rigorous and computationally efficient for use in registration tasks.

We now express the analytic mapping $\tau(y)$ as an observation equation in the form of a residual:

\begin{eqnarray}\label{eq:Observation_equation_2d}
\begin{aligned}
\begin{bmatrix} L_{1}\\L_{2}\end{bmatrix} = \begin{bmatrix} x_{1}\\x_{2}\end{bmatrix} - \tau(y),
\end{aligned}
\end{eqnarray}

where \( L_1, L_2 \) denote the residuals between observed target points and the predicted values produced by the truncated Taylor model.

To estimate the unknown parameters in the Taylor expansion (i.e., the coefficients of partial derivatives), we formulate the following least-squares optimization problem:

\begin{eqnarray}\label{eq:Matrix_form}
\begin{aligned}
\arg\min_{f_1, f_2} \left\| \begin{bmatrix} L_{1}\\L_{2} \end{bmatrix} \right\|^2,
\end{aligned}
\end{eqnarray}

where \( f_1, f_2 \) denote the truncated Taylor approximations of the coordinate functions. The notation \( f(\mathfrak{c}) \) abbreviates \( f(\mathfrak{c}_1, \mathfrak{c}_2) \), and \( R_n \) denotes the remainder of the Taylor expansion.

Equation~\eqref{eq:analytic_mapping_2d}, which expresses a truncated Taylor expansion for a 2D vector-valued mapping, can be naturally generalized to arbitrary dimensions. For a smooth function \( \mathrm{f} : \mathbb{R}^n \to \mathbb{R}^n \), we expand each output coordinate around a reference point \( \mathfrak{c} \in \mathbb{R}^n \) using the multivariate Taylor theorem. The complete vector-valued expansion is organized into matrix form as follows:

\begin{eqnarray}\label{eq:analytic_mapping_anydegree}
\begin{aligned}
 & \mathrm {f}\left ( y \right ) =
\begin{bmatrix}
 f_{1}(\mathfrak{c}) \\\vdots\\ f_{n}(\mathfrak{c})\end{bmatrix}+\begin{bmatrix}\frac{\partial f_{1} }{\partial y_{1}}\left ( \mathfrak{c}\right ) &\cdots& \frac{\partial f_{1} }{\partial y_{n}}\left ( \mathfrak{c}\right )
\\\vdots& \ddots &\vdots
\\ \frac{\partial f_{n} }{\partial y_{1}}\left ( \mathfrak{c}\right )& \cdots &\frac{\partial f_{n} }{\partial y_{n}}\left ( \mathfrak{c}\right )
\end{bmatrix}\begin{bmatrix}y_{1}-\mathfrak{c}_{1}\\\vdots \\y_{n}-\mathfrak{c}_{n}\end{bmatrix}
\\&+\frac{1}{2} \begin{bmatrix}
 \frac{\partial^{2} f_{1} }{\partial y_{1}^{2}}\left ( \mathfrak{c}\right )  &  \frac{\partial^{2} f_{1} }{\partial y_{1}\partial y_{2}}\left ( \mathfrak{c}\right ) & \frac{\partial^{2} f_{1} }{\partial y_{1}\partial y_{3}}\left ( \mathfrak{c}\right )& \cdots &\frac{\partial^{2} f_{1} }{\partial y_{n}^{2}}\left ( \mathfrak{c}\right )
\\\vdots &  & \ddots &  &\vdots
\\ \frac{\partial^{2} f_{n} }{\partial y_{1}^{2}}\left ( \mathfrak{c}\right )  &  \frac{\partial^{2} f_{n} }{\partial y_{1}\partial y_{2}}\left ( \mathfrak{c}\right ) & \frac{\partial^{2} f_{n} }{\partial y_{1}\partial y_{3}}\left ( \mathfrak{c}\right )& \cdots &\frac{\partial^{2} f_{n} }{\partial y_{n}^{2}}\left ( \mathfrak{c}\right )\end{bmatrix}
\begin{bmatrix}
 \left ( y_{1}-\mathfrak{c}_{1} \right )^{2} \\2\left ( y_{1}-\mathfrak{c}_{1} \right )\left ( y_{2}-\mathfrak{c}_{2} \right )\\2\left ( y_{1}-\mathfrak{c}_{1} \right )\left ( y_{3}-\mathfrak{c}_{3} \right )\\\vdots\\\left ( y_{n}-\mathfrak{c}_{n} \right )^{2}\end{bmatrix}
 +\dots +\\&\frac{1}{m!}
 \begin{bmatrix}
  \frac{\partial^{m} f_{1} }{\partial y_{1}^{m}}\left ( \mathfrak{c}\right ) & \dots& \frac{\partial^{m} f_{1} }{\partial y_{1}^{m_{1}}\dots \partial y_{n}^{m_{n}}}\left ( \mathfrak{c}\right )&\dots   & \frac{\partial^{m} f_{1} }{\partial y_{n}^{m}}\left ( \mathfrak{c}\right )\\
  \vdots& &\ddots & & \vdots\\
    \frac{\partial^{m} f_{n} }{\partial y_{1}^{m}}\left ( \mathfrak{c}\right ) & \dots& \frac{\partial^{m} f_{n} }{\partial y_{1}^{m_{1}}\dots \partial y_{n}^{m_{n}}}\left ( \mathfrak{c}\right )&\dots   & \frac{\partial^{m} f_{n} }{\partial y_{n}^{m}}\left ( \mathfrak{c}\right )
\end{bmatrix}
\begin{bmatrix}
 \left ( y_{1}-\mathfrak{c}_{1} \right ) ^{m} \\\vdots\\\frac{m!}{m_{1}!\cdots m_{n}!} \left ( y_{1}-\mathfrak{c}_{1} \right )^{m_{1}}\cdots \left ( y_{n}-\mathfrak{c}_{n} \right )^{m_{n}}\\\vdots \\ \left ( y_{n}-\mathfrak{c}_{n} \right )^{m}\end{bmatrix}
\\&+R_{m}.
\end{aligned}
\end{eqnarray}

Here, the index relation $m_{1} + m_{2} + \cdots + m_{n} = m$ follows the multinomial theorem. The number of monomials (i.e., the number of basis terms) at the $m$-th order in $n$-dimensional space is given by the binomial coefficient $\binom{m+n-1}{n-1}$, where $m \ge 0$ and $n \ge 1$.

This formulation generalizes the vector-valued Taylor expansion to arbitrary input and output dimensions and any approximation order. Each output coordinate function is analytically approximated as a polynomial in \( y_1, \dots, y_n \), with the full mapping constructed by stacking all output expansions.

The structure forms a compact, low-dimensional function space for approximating smooth deformations, while retaining interpretability and mathematical rigor. In practice, even low-order truncations of this form are sufficient to express a wide range of smooth deformations in point set registration. The expansion also naturally supports efficient optimization strategies, as discussed in the following sections.

The observation equation and optimization formulation for an $m$-th order Taylor expansion of a point mapping in $n$-dimensional space are as follows:

\begin{eqnarray}\label{eq:Observation_equation_anydegree}
\begin{aligned}
 \begin{bmatrix} L_{1}\\\vdots \\L_{n}\end{bmatrix}=
\begin{bmatrix} x_{1}\\\vdots \\x_{n}\end{bmatrix}- \mathrm {f}\left ( y \right )
\end{aligned},
\end{eqnarray}

\begin{eqnarray}\label{eq:Matrix_form_anydegree}
\begin{aligned}
 \arg\min_{f_{1},\dots, f_{n} } \left\| \begin{bmatrix} L_{1}\\\vdots \\L_{n}\end{bmatrix} \right\| ^{2}
\end{aligned}.
\end{eqnarray}

\paragraph{Objective with estimated correspondence}
Since the ground-truth correspondence is unknown, we alternate between forming
\(C^{(t)}\) (e.g., nearest-neighbor in ICP) and minimizing the objective in~\eqref{eq:Matrix_form_anydegree}
given \(C^{(t)}\). This outer loop is summarized in
Fig.~\ref{fig:analytic_mapping_flow} and instantiated in Section~\ref{sec:regist}.

In Equation~\eqref{eq:analytic_mapping_anydegree}, $f(\mathfrak{c})$ represents the evaluation of the multivariate function at the expansion center, i.e., $f(\mathfrak{c}_{1},\dots,\mathfrak{c}_{n})$, and $\mathrm{f}(y)$ represents the vector-valued function:

\[
\mathrm{f}(y) = \begin{bmatrix} f_{1}(y_{1},\dots, y_{n}) & \cdots & f_{n}(y_{1},\dots, y_{n}) \end{bmatrix}^T.
\]

\equref{eq:analytic_mapping_anydegree} is a natural elementary expression combining rigid, affine, and nonrigid registration. Where the translation vector is equivalent to

\begin{eqnarray}\label{eq:tran_vec}
\begin{aligned}
\begin{bmatrix}
 f_{1}\left ( \mathfrak{c} \right ) \\
 \vdots \\
 f_{n}\left ( \mathfrak{c} \right ) 
\end{bmatrix} -\begin{bmatrix}
\frac{\partial f_{1}}{\partial y_{1}}\left ( \mathfrak{c} \right )    & \cdots  & \frac{\partial f_{1}}{\partial y_{n}}\left ( \mathfrak{c} \right ) \\
 \vdots  & \ddots  & \vdots \\
 \frac{\partial f_{n}}{\partial y_{1}}\left ( \mathfrak{c} \right ) & \cdots &\frac{\partial f_{n}}{\partial y_{n}}\left ( \mathfrak{c} \right )
\end{bmatrix}\begin{bmatrix}
 \mathfrak{c}_{1}\\
 \vdots \\
\mathfrak{c}_{n}
\end{bmatrix}
\end{aligned},
\end{eqnarray}

rotation matrix (rigid registration) and nonsingular matrix (affine registration) are included in

\begin{eqnarray}\label{eq:nonsingular}
\begin{aligned}
\begin{bmatrix}
\frac{\partial f_{1}}{\partial y_{1}}\left ( \mathfrak{c} \right )    & \cdots  & \frac{\partial f_{1}}{\partial y_{n}}\left ( \mathfrak{c} \right ) \\
 \vdots  & \ddots  & \vdots \\
 \frac{\partial f_{n}}{\partial y_{1}}\left ( \mathfrak{c} \right ) & \cdots &\frac{\partial f_{n}}{\partial y_{n}}\left ( \mathfrak{c} \right )
\end{bmatrix}\begin{bmatrix}
 y_{1}\\
 \vdots \\
y_{n}
\end{bmatrix}
\end{aligned},
\end{eqnarray}

Ultimately, the terms of the second-order and above in the series control the non-affine deformation between the point sets.

As shown in \equref{eq:analytic_mapping_anydegree}, if $\left \| y - \mathfrak{c} \right \|$ is sufficiently small and the Taylor expansion order is high enough, the remainder term becomes negligible as an $n$-th order infinitesimal. In such cases, the truncated Taylor series can be approximated by a ring structure.

This property enables the optimization problem in \cref{eq:Matrix_form_anydegree} to incrementally expand the admissible function space throughout the fitting process, thereby reducing the risk of overfitting caused by an overly expressive space. Moreover, since the terms in \cref{eq:Matrix_form_anydegree} are disjoint and can be regarded as forming a graded ring, it becomes feasible to design algorithms that avoid redundant computations across fitting iterations—leading to improvements in both computational efficiency and robustness. This structure not only ensures theoretical soundness but also inspires an efficient implementation strategy, as we will demonstrate in the subsequent sections.

Our Taylor expansion model offers a global approximation framework for modeling smooth distortions in registration tasks. It significantly reduces the complexity of deformation mapping while maintaining expressive power. Although the number of parameters grows rapidly with the expansion order—specifically, the $m$-th order term involves $n \cdot \binom{m + n - 1}{n - 1}$ coefficients, as shown in \equref{eq:analytic_mapping_anydegree}—even low-order truncated expansions exhibit strong representational capacity, as demonstrated in \figref{fig:fistdistort}. This property is especially advantageous in iterative registration, where low-order approximations often suffice to meet practical accuracy requirements.

Moreover, our algorithm exhibits a favorable convergence behavior: with each iteration, the fitted mapping evolves as the product of $n$ Taylor series, leading to an exponential increase in expressive order. Consequently, high-complexity deformations can often be captured accurately within just a few low-order iterative refinements.

\subsection{Structured Taylor Expansion Model}

While the Taylor expansion in its full form involves high-order partial derivatives and monomial combinations, the following definitions provide a structured matrix-vector representation of each term. This allows compact formulation and efficient implementation of the approximation process.

To approximate smooth vector-valued functions in a numerically constructive manner, we define a matrix-based representation of higher-order derivatives and monomial terms. This formulation provides the foundation of our analytic deformation model, enabling efficient approximation of smooth mappings with controllable complexity.

\begin{definition}[Generalized Derivative Matrix]
Let $f: \mathbb{R}^n \to \mathbb{R}^n$ be a smooth vector-valued function, and let $m \in \mathbb{N}$ denote the expansion order. Denote $\mathfrak{c} \in \mathbb{R}^n$ as the expansion center.

The $m$-th order \emph{generalized derivative matrix} $\mathcal{J}^{[m]}(f)(\mathfrak{c}) \in \mathbb{R}^{n \times N_m}$ is defined as:
\[
\mathcal{J}^{[m]}(f)(\mathfrak{c}) := 
\begin{bmatrix}
\displaystyle \frac{\partial^{m} f_1 }{\partial y_1^{m}} & \cdots & \displaystyle \frac{\partial^{m} f_1 }{\partial y_1^{m_1}\cdots \partial y_n^{m_n}} & \cdots & \displaystyle \frac{\partial^{m} f_1 }{\partial y_n^{m}} \\
\vdots & & \ddots & & \vdots \\
\displaystyle \frac{\partial^{m} f_n }{\partial y_1^{m}} & \cdots & \displaystyle \frac{\partial^{m} f_n }{\partial y_1^{m_1}\cdots \partial y_n^{m_n}} & \cdots & \displaystyle \frac{\partial^{m} f_n }{\partial y_n^{m}}
\end{bmatrix}_{\big| y = \mathfrak{c}},
\]
where each column corresponds to one multi-index $(m_1, \dots, m_n)$ satisfying $m_1 + \cdots + m_n = m$. The total number of such terms is $N_m = \binom{n + m - 1}{n-1}$.
\end{definition}

\vspace{2mm}

\begin{definition}[Generalized Monomial Vector]
Given $y, \mathfrak{c} \in \mathbb{R}^n$, the \emph{generalized monomial vector} of order $m$ is defined as
\[
\phi^{[m]}(y - \mathfrak{c}) :=
\left[
(y_1 - \mathfrak{c}_1)^m,\;
\cdots,\;
\frac{m!}{m_1! \cdots m_n!} (y_1 - \mathfrak{c}_1)^{m_1} \cdots (y_n - \mathfrak{c}_n)^{m_n},\;
\cdots,\;
(y_n - \mathfrak{c}_n)^m
\right]^T \in \mathbb{R}^{N_m},
\]
where each term corresponds to a unique multi-index $(m_1, \dots, m_n)$ such that $m_1 + \cdots + m_n = m$.
\end{definition}

\vspace{2mm}

\begin{definition}[Structured Taylor Expansion]
Using the above notation, the $m$-th order structured Taylor approximation of $f$ at point $\mathfrak{c}$ is expressed as:
\[
f(y) \approx \sum_{k = 0}^{m} \frac{1}{k!} \mathcal{J}^{[k]}(f)(\mathfrak{c}) \cdot \phi^{[k]}(y - \mathfrak{c}),
\]
where $\phi^{[0]}(y - \mathfrak{c}) = 1$, and the dot product indicates matrix-vector multiplication. The residual term $R_m(y)$ is omitted here for brevity but can be quantified using standard Taylor remainder bounds.
\end{definition}

\vspace{2mm}
\begin{theorem}[Structured Taylor Approximation with Remainder]\label{thm:structured_taylor}
Let $f: \mathbb{R}^n \to \mathbb{R}^n$ be a $C^{m+1}$ smooth vector-valued function, and let $\mathfrak{c} \in \mathbb{R}^n$ be the expansion center. Then for any $y \in \mathbb{R}^n$, the Taylor expansion of $f$ up to order $m$ admits the following structured form:
\[
f(y) = \sum_{k = 0}^{m} \frac{1}{k!} \mathcal{J}^{[k]}(f)(\mathfrak{c}) \cdot \phi^{[k]}(y - \mathfrak{c}) + R_{m+1}(y),
\]
where the remainder term $R_{m+1}(y)$ is expressed as
\[
R_{m+1}(y) = \frac{1}{(m+1)!} \, \mathcal{J}^{[m+1]}(f)(\theta y + (1-\theta)\mathfrak{c}) \cdot \phi^{[m+1]}(y - \mathfrak{c}),
\quad \text{for some } \theta \in (0,1).
\]
\end{theorem}

\vspace{1mm}
\begin{proof}[Proof Sketch]
This result follows from the classical multivariate Taylor theorem with Lagrange-form remainder. For each component $f_i$, the $(m+1)$-th order expansion around $\mathfrak{c}$ can be written as
\[
f_i(y) = \sum_{|\alpha| \le m} \frac{1}{\alpha!} \partial^\alpha f_i(\mathfrak{c}) (y - \mathfrak{c})^\alpha 
+ \sum_{|\alpha| = m+1} \frac{1}{\alpha!} \partial^\alpha f_i(\xi) (y - \mathfrak{c})^\alpha,
\]
where $\xi$ lies on the segment between $\mathfrak{c}$ and $y$.

By collecting all partial derivatives of order $k$ across all output dimensions, the expressions can be reorganized into the form
\[
f(y) = \sum_{k = 0}^{m} \frac{1}{k!} \mathcal{J}^{[k]}(f)(\mathfrak{c}) \cdot \phi^{[k]}(y - \mathfrak{c}) + 
\frac{1}{(m+1)!} \mathcal{J}^{[m+1]}(f)(\xi) \cdot \phi^{[m+1]}(y - \mathfrak{c}),
\]
with $\xi = \theta y + (1 - \theta)\mathfrak{c}$ for some $\theta \in (0,1)$. This completes the structured remainder representation.
\end{proof}

\vspace{3mm}
\begin{lemma}[Nontriviality of the Structured Remainder]\label{lemma:remainder_nonzero}
Let $f: \mathbb{R}^n \to \mathbb{R}^n$ be a $C^{m+1}$ vector-valued function that is not a polynomial of total degree $\le m$ in any neighborhood of $\mathfrak{c} \in \mathbb{R}^n$. Then the structured Taylor remainder term
\[
R_{m+1}(y) = \frac{1}{(m+1)!} \cdot \mathcal{J}^{[m+1]}(f)(\xi) \cdot \phi^{[m+1]}(y - \mathfrak{c})
\]
is nonzero for some $y$ sufficiently close to $\mathfrak{c}$ and some $\xi$ on the segment $[\mathfrak{c}, y]$.
\end{lemma}

\begin{proof}
Let us consider the classical multivariate Taylor expansion of each output component $f_i$ at point $\mathfrak{c}$. If $f$ is not a degree-$m$ polynomial in a neighborhood of $\mathfrak{c}$, then by definition, there exists at least one multi-index $\alpha$ with $|\alpha| = m+1$ such that the partial derivative $\partial^\alpha f_i(\mathfrak{c}) \neq 0$ for some $i \in \{1, \ldots, n\}$.

By continuity of partial derivatives, there exists a point $\xi$ arbitrarily close to $\mathfrak{c}$ where this nonzero derivative persists. Then, evaluating the structured remainder expression at any $y$ sufficiently close to $\mathfrak{c}$ ensures that the monomial vector $\phi^{[m+1]}(y - \mathfrak{c})$ contains nonzero components in the corresponding direction, and thus:
\[
R_{m+1}(y) = \frac{1}{(m+1)!} \cdot \mathcal{J}^{[m+1]}(f)(\xi) \cdot \phi^{[m+1]}(y - \mathfrak{c}) \ne 0.
\]
This proves that the remainder is nontrivial unless $f$ is a polynomial of degree at most $m$.
\end{proof}

\paragraph{Componentwise interpretation}
A vector-valued map $\tau: W\subset\mathbb{R}^n\to\mathbb{R}^n$ can be written as
$\tau(y)=(f_1(y),\dots,f_n(y))$, with scalar $f_i:W\to\mathbb{R}$.
Existence/uniqueness and approximation in the structured Taylor basis hold
componentwise for each $f_i$ and then follow for $\tau$ by stacking the expansions.

\begin{theorem}[Structured Approximation of Smooth Registration Mappings]\label{thm:main}
Let $\tau: W \subset \mathbb{R}^n \to \mathbb{R}^n$ be a smooth mapping, decomposed into components $\tau(y) = (f_1(y), \dots, f_n(y))$.

\begin{enumerate}
    \item[(i)]If $\tau$ is analytic, then each $f_i$ admits a unique expansion in the structured multivariate basis:
    \[
    f_i(y) = \sum_{k=0}^\infty \frac{1}{k!} \mathcal{J}^{[k]}(f_i)(\mathfrak{c}) \cdot \phi^{[k]}(y - \mathfrak{c}),
    \]
    where $\mathcal{J}^{[k]}$ and $\phi^{[k]}$ denote the generalized derivative matrix and monomial vector. Consequently, the mapping $\tau$ admits a unique representation in the structured Taylor basis.

    \item[(ii)]If $\tau$ is merely smooth, then for any compact $K \subset W$ and $\epsilon > 0$, there exists an analytic mapping $\tilde{\tau}$ whose truncated expansion in the structured Taylor basis approximates $\tau$ to within $\epsilon$:
    \[
    \|\tau - \tilde{\tau}\|_{C^0(K)} < \epsilon.
    \]
\end{enumerate}
Thus, the structured Taylor model introduced in this paper can approximate any smooth registration mapping arbitrarily closely on compact subsets of $W$.
\end{theorem}

\begin{proof}[Proof Sketch]
See Appendix~\ref{appendix:proof_structured_taylor} for a complete proof.
\end{proof}

\noindent\emph{Remark.} We use a matrix–vector rewriting of the multivariate Taylor basis that aligns with standard least–squares algebra, yielding a compact parameterization with good numerical conditioning.

\vspace{3mm}
\begin{lemma}[Density of Structured Taylor Approximants in $C^\infty$]\label{lemma:structured_density_strong}
Let $W\subset\mathbb{R}^n$ be open, $K\Subset W$ a compact set, and let $f: W \to \mathbb{R}^n$ be a smooth function. Fix a center $\mathfrak{c}\in W$. Then for any $\epsilon > 0$, there exist an analytic map $g: W \to \mathbb{R}^n$ (a polynomial suffices) and $m\in\mathbb{N}$, and a structured Taylor approximant of the form
\[
\tilde{f}(y) = \sum_{k=0}^{m} \frac{1}{k!} \, \mathcal{J}^{[k]}(g)(\mathfrak{c}) \cdot \phi^{[k]}(y - \mathfrak{c}),
\]
such that
\[
\|f - \tilde{f}\|_{C^0(K)} < \epsilon.
\]
That is, the set of structured multivariate Taylor approximants (with coefficients taken from an analytic surrogate) is dense in $C^\infty(K; \mathbb{R}^n)$ under the uniform topology.
\end{lemma}

\begin{proof}[Sketch of Proof]
It is well known that the space of vector-valued multivariate polynomials is dense in $C^0(K; \mathbb{R}^n)$ (\emph{cf.}~\cite{hormander1983analysis, rudin1991functional}). The structured Taylor forms introduced in this paper are precisely a reparameterization of such polynomials, organized by total degree and multi-index order in a matrix–vector form.

Specifically, for any smooth function \( f \) and any \( \epsilon>0 \), there exists a polynomial (hence analytic) mapping \( g:W\to\mathbb{R}^n \) with \( \|f - g\|_{C^0(K)} < \epsilon/2 \). Since \( g \) is analytic on \( W \), it equals its Taylor series at \( \mathfrak c \); therefore, for a sufficiently large truncation order \( m \),
\[
\big\|g - \sum_{k=0}^{m} \tfrac{1}{k!}\,\mathcal{J}^{[k]}(g)(\mathfrak c)\cdot \phi^{[k]}(\,\cdot-\mathfrak c)\big\|_{C^0(K)} < \epsilon/2.
\]
Letting \( \tilde f \) be this truncated structured Taylor form gives \( \|f-\tilde f\|_{C^0(K)}<\epsilon \) by the triangle inequality.
\end{proof}

Building on this structured formulation, we now turn to the practical aspects of model fitting. Since our expansion model encapsulates both linear and nonlinear components of the transformation, we propose to decouple and estimate them in a hierarchical manner. In particular, the linear portion of the mapping—corresponding to rotation, scaling, shear, and projective effects—admits a natural decomposition aligned with classical transformation groups such as $\mathrm{SE}(n,\mathbb{R})$.

By isolating and fitting these components in sequence, we not only preserve the interpretability of each transformation stage but also improve numerical stability and convergence. This motivates the following section, where we detail the construction and estimation of the linear substructure prior to introducing higher-order terms.

\subsection{Fitting of Perspective Mapping}\label{sec:linear}

\begin{figure}[t]
  \centering
  \begin{tikzpicture}[>=stealth]
    \node [node font=\large] (Y0) at (-3, 0) {$Y$};
    \node [node font=\large] (Y1) at (-1, 0) {$\widehat{Y}_{\mathrm{rig}}$};
    \node [node font=\large] (Y2) at ( 1, 0) {$\widehat{Y}_{\mathrm{aff}}$};
    \node [node font=\large] (Y3) at ( 3, 0) {$\widehat{Y}_{\mathrm{proj}}$};

    \draw[->] (Y0.east) -- (Y1.west);
    \draw[->] (Y1.east) -- (Y2.west);
    \draw[->] (Y2.east) -- (Y3.west);

    \node [node font=\small] at (-2, 0.30) {$\mathrm{SE}(n,\mathbb{R})$};
    \node [node font=\small] at ( 0, 0.30) {$\frac{\mathrm{AGL}(n,\mathbb{R})}{\mathrm{SE}(n,\mathbb{R})}$};
    \node [node font=\small] at ( 2, 0.30) {$\frac{\mathrm{PGL}(n{+}1,\mathbb{R})}{\mathrm{AGL}(n,\mathbb{R})}$};
  \end{tikzpicture}

  \vspace{1mm}
  \caption{\textbf{Hierarchical decomposition of linear components.}
  The registration pipeline follows the subgroup chain
  $\mathrm{SE}(n,\mathbb{R})\subset \mathrm{AGL}(n,\mathbb{R})\subset \mathrm{PGL}(n{+}1,\mathbb{R})$.
  Each stage fits residual degrees of freedom along the corresponding homogeneous direction:
  first rigid motions $\mathrm{SE}(n,\mathbb{R})$, then affine modes modulo rigid
  $\mathrm{AGL}(n,\mathbb{R})/\mathrm{SE}(n,\mathbb{R})$ (scales, shears), and finally projective modes modulo affine
  $\mathrm{PGL}(n{+}1,\mathbb{R})/\mathrm{AGL}(n,\mathbb{R})$.}
  \label{fig:group_hierarchy}
  \vspace{-0.2mm}
\end{figure}

\paragraph{Hierarchical linear–projective decomposition}
We decompose the linear part of the analytic mapping into a sequence of sub-mappings—\emph{rotation}, \emph{affine stretch/shear}, and \emph{perspective distortion}—which mirrors the structural hierarchy
\[
\mathrm{SE}(n,\mathbb{R})\ \subset\ \mathrm{AGL}(n,\mathbb{R})\ \subset\ \mathrm{PGL}(n{+}1,\mathbb{R}),
\]
and enables stepwise fitting with separate optimization of each component, thereby reducing complexity and improving stability.%
\footnote{Throughout, the notation “$G/H$” denotes a \emph{homogeneous space} (quotient set) rather than a quotient group.}

\paragraph{Rigid $\to$ affine}
Following an initial orthogonal fit~\cite{ICP1992}, we lift to an affine model. Given the fitted linear part \(A\in \mathrm{GL}(n,\mathbb{R})\), we compute its QR (or polar) factorization
\[
A = Q R \quad\text{with}\quad Q\in \mathrm{SO}(n,\mathbb{R}),\ \ R\ \text{upper-triangular (positive diagonal),}
\]
or equivalently \(A=Q S\) with \(S\in \mathrm{Sym}_n^{+}\). This separates the rotational factor \(Q\) from the non-rotational linear factor (\(R\) or \(S\)), where the latter represents a point in the homogeneous space \(\mathrm{GL}(n,\mathbb{R})/\mathrm{SO}(n,\mathbb{R})\).
To pass from the orthogonal stage to the \emph{affine} stage, we keep the previously fitted \(Q\) fixed and estimate the pair \((R,\mathbf t)\) with \(\mathbf t\in\mathbb{R}^{n}\):
\[
x \ \mapsto\ Q\,R\,x + \mathbf t .
\]
\emph{In addition, we estimate a global-frame translation increment \(\delta \mathbf t_A\) and \textbf{accumulate} it to the current translation} (left-multiplicative update):
\[
\mathbf t \leftarrow \mathbf t + \delta \mathbf t_A, 
\qquad
(I,\delta \mathbf t_A)\circ(QR,\mathbf t)=(QR,\,\mathbf t+\delta \mathbf t_A),
\]
which leaves the linear factor \(QR\) unchanged and improves numerical stability; details are discussed later.
In either parameterization, augmenting the rigid component by the extra linear degrees of freedom yields the canonical identification
\[
\mathrm{AGL}(n,\mathbb{R})/\mathrm{SE}(n,\mathbb{R})\ \cong\ \mathrm{GL}(n,\mathbb{R})/\mathrm{SO}(n,\mathbb{R}),
\]
i.e., “affine modulo rigid” corresponds to “linear modulo orthogonal”. This staged design—first \((Q,\mathbf t)\), then \((R,\delta \mathbf t_A)\) with the accumulation \(\mathbf t \leftarrow \mathbf t+\delta \mathbf t_A\)—enhances controllability, numerical stability, and interpretability during optimization.

\paragraph{Affine $\to$ perspective (restricted projective refinement with translation compensation)}
After completing the affine mapping, we \emph{optionally} perform a restricted projective refinement \emph{together with a global translation compensation}. 
Concretely, in 2D we estimate the two projective ``tilt'' parameters $(a_6,a_7)$ \emph{and} a translation increment $(d\mathbf t_x,d\mathbf t_y)$ accumulated in the global frame. 
The corresponding (restricted) homography uses a last row $[\,a_6\ a_7\ 1\,]$, while the top-right entries absorb the translation compensation:
\[
\begin{bmatrix}
x' \\[2pt] y' \\[2pt] w'
\end{bmatrix}
=
\begin{bmatrix}
a_0 & a_1 & a_2 + d\mathbf t_x \\
a_3 & a_4 & a_5 + d\mathbf t_y \\
a_6 & a_7 & 1
\end{bmatrix}
\begin{bmatrix}
x_0 \\[2pt] y_0 \\[2pt] 1
\end{bmatrix},
\qquad
(x,y)=\Big(\tfrac{x'}{w'},\ \tfrac{y'}{w'}\Big).
\]
Linearizing with respect to $(a_6,a_7,d\mathbf t_x,d\mathbf t_y)$ around the current estimate 
$(a_6^o,a_7^o,d\mathbf t_x^o,d\mathbf t_y^o)$ yields the stacked error equations
\begin{equation}\label{eq:perspect_t_re}
\begin{aligned}
\begin{bmatrix}
v_1 \\[2pt] v_2
\end{bmatrix}
\;=\;
\underbrace{\begin{bmatrix}
x_0\,x' & y_0\,x' & -1 & \ \ 0 \\
x_0\,y' & y_0\,y' & \ \ 0 & -1 
\end{bmatrix}}_{J(a_6,a_7,d\mathbf t_x,d\mathbf t_y)}
\begin{bmatrix}
\delta a_6 \\[2pt] \delta a_7 \\[2pt] \delta d\mathbf t_x \\[2pt] \delta d\mathbf t_y
\end{bmatrix}
\;-\;
\begin{bmatrix}
a_0^o x_0 + a_1^o y_0 + a_2^o + d\mathbf t_x^o - x' - x_0 x' a_6^o - y_0 x' a_7^o \\
a_3^o x_0 + a_4^o y_0 + a_5^o + d\mathbf t_y^o - y' - x_0 y' a_6^o - y_0 y' a_7^o
\end{bmatrix},
\end{aligned}
\end{equation}
from which we form normal equations and solve by least squares until convergence (cf.\ indirect adjustment~\cite{ogundare2018understanding}). 
The parameter update is 
\[
(a_6,a_7,d\mathbf t_x,d\mathbf t_y)
\leftarrow
(a_6^o,a_7^o,d\mathbf t_x^o,d\mathbf t_y^o) + (\delta a_6,\delta a_7,\delta d\mathbf t_x,\delta d\mathbf t_y),
\]
and we \emph{accumulate} the global translation as
\[
\mathbf t \ \leftarrow\ \mathbf t + \delta \mathbf t_P,\qquad 
\delta \mathbf t_P \equiv (\,\delta d\mathbf t_x,\,\delta d\mathbf t_y\,)^\top,
\]
leaving the previously estimated linear part $(Q,R)$ unchanged. 
Equivalently, left-multiplying by the translation matrix 
$T(\delta \mathbf t_P)=\bigl[\begin{smallmatrix}1&0&\delta d\mathbf t_x\\[1pt]0&1&\delta d\mathbf t_y\\[1pt]0&0&1\end{smallmatrix}\bigr]$ 
updates only the top two rows of the homography and leaves the bottom row $[\,a_6\ a_7\ 1\,]$ intact, hence decoupling the projective tilt from the affine part and improving stability.

From a group-theoretic viewpoint, the affine model lives in $\mathrm{AGL}(n,\mathbb{R})$, and the projective refinement adds the residual modes along the homogeneous space 
$\mathrm{PGL}(n{+}1,\mathbb{R})/\mathrm{AGL}(n,\mathbb{R})$ (``projective modulo affine''). 
In 2D this corresponds to a plane homography in $\mathrm{PGL}(3,\mathbb{R})$ realized by $(a_6,a_7)$, while $(d\mathbf t_x,d\mathbf t_y)$ implements a global-frame translation accumulation that does not modify the linear factors already fitted.

\paragraph{Orthogonality of stages}
At the zero-tilt linearization point $(a_6,a_7)=(0,0)$ and after zero-mean/unit-scale normalization, the Jacobian columns for the projective-tilt parameters $(a_6,a_7)$ (acting via the denominator $w'$) are nearly orthogonal to those of the affine block (acting on the numerator). Thus the tilt increment $(\delta a_6,\delta a_7)$ can be estimated without disturbing the optimal affine part; the translation compensation $(\delta d\mathbf t_x,\delta d\mathbf t_y)$ is applied by a left translation that alters only the top-right entries, leaving the linear block and bottom row $[\,a_6\ a_7\ 1\,]$ intact.

\textit{Benefits:} improves numerical stability, simplifies solving in weakly coupled blocks, and keeps each update interpretable (affine vs.\ tilt vs.\ translation).

\paragraph{Summary}
We summarize the fitting sequence as
\begin{equation}\label{eq:orthogonal2projective}
Y \ \xrightarrow{\ \tau_{\mathrm{rig}}\ }\ \widehat{Y}_{\mathrm{rig}}
\ \xrightarrow{\ \tau_{\mathrm{aff}}\ }\ \widehat{Y}_{\mathrm{aff}}
\ \xrightarrow{\ \tau_{\mathrm{proj}}\ }\ \widehat{Y}_{\mathrm{proj}},
\end{equation}
where each arrow “$\rightarrow$” denotes a sequential fitting step—rigid, then affine, then projective with translation compensation. 
The staged process and the associated homogeneous directions are visualized in Fig.~\ref{fig:group_hierarchy}.

\subsection{Algorithm for Fitting Smooth Vector-Valued Functions}\label{sec:analyticFit}

\equref{eq:Matrix_form_anydegree} defines the optimization formulation for the $m$-th order Taylor expansion of point mappings in an $n$-dimensional space. We refer to the corresponding optimization algorithm as \textbf{Analytic Multivariate Vector-Valued Function Fitting (AMVFF)}.

We use the 2D Taylor expansion structure as an example to introduce the AMVFF, which can be easily extended to higher dimension. Assuming the transport mapping in 2D point set registration is smooth, we can express it as a Taylor series. In order to apply the 2D Taylor expansion structure to registration, it is necessary to transform the 2D version of \equref{eq:Observation_equation_anydegree} into an observation equation as follow: 
\begin{eqnarray}\label{eq:adjustment_form}
\begin{aligned}
 &\begin{bmatrix} L_{1}\\L_{2}\end{bmatrix}  = \begin{bmatrix} x_{1}\\x_{2}\end{bmatrix}-\begin{bmatrix}
 a_{0} \\ a_{1}\end{bmatrix}-\begin{bmatrix}a_{2}& a_{3}\\ a_{4}&  a_{5}
\end{bmatrix}\begin{bmatrix}y_{1}-\mathfrak{c}_{1} \\y_{2}-\mathfrak{c}_{2}\end{bmatrix}-\frac{1}{2} \begin{bmatrix}
 a_{6} &  a_{7} &  a_{8}\\a_{9} &  a_{10} &  a_{11}\end{bmatrix}\begin{bmatrix}
 \left ( y_{1}-\mathfrak{c}_{1} \right ) ^{2} \\2\left ( y_{1}-\mathfrak{c}_{1} \right )\left ( y_{2}-\mathfrak{c}_{2} \right )\\\left ( y_{2}-\mathfrak{c}_{2} \right )^{2}\end{bmatrix}
 \\&-\dots -\frac{1}{n!}
 \begin{bmatrix}
  a_{n^2+n}&\dots& a_{n^2+n+k} &\dots& a_{n^2+2n}\\
  a_{n^2+2n+1}&\dots& a_{n^2+2n+1+k} &\dots& a_{n^2+3n+1}
\end{bmatrix}\begin{bmatrix}
 \binom{n}{0}\left ( y_{1}-\mathfrak{c}_{1} \right )^{n} \\\vdots\\\binom{n}{k}  \left ( y_{1}-\mathfrak{c}_{1} \right )^{n-k}\left ( y_{2}-\mathfrak{c}_{2} \right )^{k}\\\vdots \\ \binom{n}{n}\left ( y_{2}-\mathfrak{c}_{2} \right )^{n}\end{bmatrix}.
\end{aligned}
\end{eqnarray}

Here, we adopt an alternating update strategy between the mapping coefficients 
$(a_{0}, \cdots, a_{n^{2} + 3n + 1})$ and the center position $(\mathfrak{c}_{1}, \mathfrak{c}_{2})$.
Specifically, the center parameters are updated via linear least-squares adjustment 
based on the formulation in \eqref{eq:error_form}–\eqref{eq:dcy2}, while the mapping 
coefficients are optimized using a quasi-Newton strategy as described below. 
This decoupled iteration avoids global search and leverages the locality and smoothness 
of the Taylor basis, typically requiring only a few steps to converge.

\subsubsection*{Adjustment Theory for Analytic Mapping}
This section introduces a parameter estimation method based on classical least-squares adjustment theory, which has its origins in geodesy and photogrammetry~\cite{ogundare2018understanding}. While this framework is less common in the computer vision and machine learning communities, it provides a principled and computationally efficient way to estimate the large number of parameters in high-order analytic mappings. In our context, the method is adapted to support the iterative fitting of multivariate vector-valued Taylor models by treating mapping coefficients and center positions as unknowns. We now describe the adjustment procedure for the 2D case, which can be extended to arbitrary dimensions.

We convert the $n^2 + 3n + 2$ mapping coefficients and the 2 location parameters in \equref{eq:adjustment_form} into unknowns, and reformulate them into the standard form of an error equation~\cite{ogundare2018understanding} for $m$ corresponding point pairs:
\begin{eqnarray}\label{eq:error_form}
\begin{aligned}
   V = B \hat{x} - l,
\end{aligned}
\end{eqnarray}
where $B$ is defined for the location parameters $(\mathfrak{c}_1, \mathfrak{c}_2)$ as
\begin{eqnarray}\label{eq:B_form4c}
\begin{aligned}
B = \begin{bmatrix}
 d\mathfrak{c}_{1x}^{1} & d\mathfrak{c}_{2x}^{1} \\
 d\mathfrak{c}_{1y}^{1} & d\mathfrak{c}_{2y}^{1} \\
 \vdots & \vdots \\
 d\mathfrak{c}_{1x}^{m} & d\mathfrak{c}_{2x}^{m} \\
 d\mathfrak{c}_{1y}^{m} & d\mathfrak{c}_{2y}^{m}
\end{bmatrix},
\end{aligned}
\end{eqnarray}
with each element defined as follows:

\begin{eqnarray}\label{eq:dcx1}
\begin{aligned}
d\mathfrak{c}_{1x}^{m} &= - \sum_{i=1}^{n} \sum_{j=0}^{i-1} \frac{1}{i!} \binom{i}{j} (i-j) \, a_{i^2 + i + j} (y_{m,1} - \mathfrak{c}_1)^{i-j-1} (y_{m,2} - \mathfrak{c}_2)^j
\end{aligned}
\end{eqnarray}

\begin{eqnarray}\label{eq:dcx2}
\begin{aligned}
d\mathfrak{c}_{2x}^{m} &= - \sum_{i=1}^{n} \sum_{j=1}^{i} \frac{1}{i!} \binom{i}{j} j \, a_{i^2 + i + j} (y_{m,1} - \mathfrak{c}_1)^{i-j} (y_{m,2} - \mathfrak{c}_2)^{j-1}
\end{aligned}
\end{eqnarray}

\begin{eqnarray}\label{eq:dcy1}
\begin{aligned}
d\mathfrak{c}_{1y}^{m} &= - \sum_{i=1}^{n}\sum_{j=0}^{i-1} \frac{1}{i!} \binom{i}{j} (i-j) \, a_{i^2 + 2i + 1 + j} (y_{m,1} - \mathfrak{c}_1)^{i-j-1} (y_{m,2} - \mathfrak{c}_2)^j
\end{aligned}
\end{eqnarray}

\begin{eqnarray}\label{eq:dcy2}
\begin{aligned}
d\mathfrak{c}_{2y}^{m} &= - \sum_{i=1}^{n}\sum_{j=1}^{i-1} \frac{1}{i!} \binom{i}{j} j \, a_{i^2 + 2i + 1 + j} (y_{m,1} - \mathfrak{c}_1)^{i-j} (y_{m,2} - \mathfrak{c}_2)^{j-1}
\end{aligned}
\end{eqnarray}

where $y_{m,1}$ and $y_{m,2}$ denote the $x$ and $y$ coordinates of the $m$-th moving point, respectively.
$B$ is defined for the mapping parameters \((a_{0}, \cdots, a_{n^2+3n+1})\) as
\begin{equation}\label{eq:B_form}
\begin{aligned}
B = \begin{bmatrix}
 S_{1}^{0}  & \cdots  & S_{1}^{n}\\
 \vdots     & \ddots  & \vdots \\
 S_{m}^{0}  & \cdots  & S_{m}^{n}
\end{bmatrix},
\end{aligned}
\end{equation}
where each block \( S_{m}^{k} \) denotes the basis vector constructed from the \( k \)th-order Taylor monomials evaluated at the \( m \)th sample point.

The monomial vector \( Y_{m}^{n} \), used to build the Taylor basis, is defined as
\begin{equation}\label{eq:Ymn_form}
\begin{aligned}
Y_{m}^{n} = \frac{1}{n!}
\begin{bmatrix}
 \binom{n}{0} \left( y_{m,1} - \mathfrak{c}_1 \right)^n \\
 \vdots \\
 \binom{n}{t} \left( y_{m,1} - \mathfrak{c}_1 \right)^{n - t} \left( y_{m,2} - \mathfrak{c}_2 \right)^t \\
 \vdots \\
 \binom{n}{n} \left( y_{m,2} - \mathfrak{c}_2 \right)^n
\end{bmatrix}^{\top}.
\end{aligned}
\end{equation}

We now deduce \( S_{m}^{n} \) in \equref{eq:B_form}, resulting in
\begin{equation}\label{eq:Smn_form}
\begin{aligned}
S_{m}^{n} = \begin{bmatrix}
 Y_{m}^{n}  & 0 & \cdots & 0 \\
 0 & \cdots & 0 & Y_{m}^{n}
\end{bmatrix},
\end{aligned}
\end{equation}
where each row block in \( S_{m}^{n} \) is padded with \( n+1 \) zeros accordingly.

The vector \( l \) in \equref{eq:error_form} is formulated as:
\begin{equation}\label{eq:l_form}
\begin{aligned}
l =
\begin{bmatrix}
x_{1,1} - a_{1} - AY_{1,1}^{1} - \cdots - AY_{1,1}^{n} \\
x_{1,2} - a_{2} - AY_{1,2}^{1} - \cdots - AY_{1,2}^{n} \\
\vdots \\
x_{m,1} - a_{1} - AY_{m,1}^{1} - \cdots - AY_{m,1}^{n} \\
x_{m,2} - a_{2} - AY_{m,2}^{1} - \cdots - AY_{m,2}^{n}
\end{bmatrix},
\end{aligned}
\end{equation}
noting that both mapping and center parameters share the same residual vector expression in \eqref{eq:error_form}.

We define \( AY_{m,1}^{n} \) and \( AY_{m,2}^{n} \) as:
\begin{equation}\label{eq:ay1_form_sigma}
AY_{m,1}^{n} = \sum_{k=0}^{n} \frac{1}{n!} \binom{n}{k} \, a_{n^2 + n + 2k} \, (y_{m,1} - \mathfrak{c}_1)^{n-k} (y_{m,2} - \mathfrak{c}_2)^k
\end{equation}

\begin{equation}\label{eq:ay2_form_sigma}
AY_{m,2}^{n} = \sum_{k=0}^{n} \frac{1}{n!} \binom{n}{k} \, a_{n^2 + n + 1 + 2k} \, (y_{m,1} - \mathfrak{c}_1)^{n-k} (y_{m,2} - \mathfrak{c}_2)^k
\end{equation}

The adjustment criterion is:
\begin{equation}\label{eq:criterion_form}
V^\top P V = \text{minimum},
\end{equation}
where \( P \) is the identity matrix in this case.

Define:
\begin{equation}\label{eq:NBB_W}
M = B^\top P B, \quad W = B^\top P l.
\end{equation}

Then the normal equation for the least-squares adjustment becomes~\cite{ogundare2018understanding}:
\begin{equation}\label{eq:normal_equation}
M \hat{v} = W.
\end{equation}

We solve the update:
\begin{equation}\label{eq:xm_form}
\hat{v} = 
\begin{bmatrix}
\delta a_{1} & \cdots & \delta a_{n^2 + n}
\end{bmatrix},
\end{equation}

and obtain the updated parameter vector:
\begin{equation}\label{eq:adjustment_result}
\hat{X} = X^{o} + \hat{v},
\end{equation}
where \( X^o \) and \( \hat{X} \) denote the initial and updated estimates of the parameter vector, respectively.

The above linear least squares algorithm is further enhanced by incorporating the BFGS\footnote{Broyden–Fletcher–Goldfarb–Shanno} update, which adaptively adjusts the weighting matrix during optimization. This dynamic reweighting improves the convergence rate by enabling the entire optimization procedure to follow a quasi-Newton strategy, thereby accelerating the solution of $\hat{x}$ and enhancing numerical stability. In addition, this quasi-Newton approach adapts well to the high-dimensional parameter space induced by the Taylor expansion, maintaining robustness and efficiency even in the presence of complex nonlinear mappings.

\begin{lemma}[Sufficient Condition for Unique and Stable Structured Taylor Fitting]
Let $f : \mathbb{R}^n \to \mathbb{R}^n$ be a smooth mapping, and let $T_m(y)$ denote its order-$m$ expansion in the structured Taylor basis centered at $\mathfrak{c}$:
\[
T_m(y) = \sum_{k=0}^m \frac{1}{k!} \mathcal{J}^{[k]}(f)(\mathfrak{c}) \cdot \phi^{[k]}(y - \mathfrak{c}).
\]

Suppose we are given $M$ correspondence pairs $\{(y_i, f(y_i))\}_{i=1}^M$ used to estimate the generalized derivative matrices $\mathcal{J}^{[k]} \in \mathbb{R}^{n \times N_k}$ for $k = 0, \dots, m$, where $N_k = \binom{n + k - 1}{n-1}$ is the number of $k$-th order monomial terms.

Then the structured least-squares fitting problem admits a unique and numerically stable solution if the following conditions are satisfied:
\begin{enumerate}
    \item[(C1)] (\textbf{Minimal data condition})  
    The total number of point correspondences satisfies:
    \[
    M \ge \left\lceil \frac{1}{n} \sum_{k = 0}^{m} n \cdot N_k \right\rceil = \left\lceil \sum_{k=0}^m N_k \right\rceil.
    \]
    \item[(C2)] (\textbf{Geometric distribution condition})  
    The data points $\{y_i\}$ are assumed to be quasi-isotropically distributed around the expansion center $\mathfrak{c}$, 
with a sufficiently small radius 
\( r = \max_i \|y_i - \mathfrak{c}\| \).
    \item[(C3)] (\textbf{Algebraic conditioning})  
    The monomial matrix $\Phi \in \mathbb{R}^{M \times \sum_k N_k}$ formed by evaluating $\phi^{[k]}(y_i - \mathfrak{c})$ is full rank and well-conditioned, i.e., $\kappa(\Phi) \ll 10^4$ \protect\footnote{%
Here, $\kappa(\Phi)$ denotes the condition number of the monomial matrix $\Phi$, defined as the ratio between its largest and smallest singular values. A large condition number indicates near-linear dependence among the columns of $\Phi$, which leads to numerical instability in least-squares estimation. In practice, $\kappa(\Phi) \gg 10^4$ suggests potential overfitting, poor excitation of higher-order terms, and amplification of noise in the fitted coefficients.%
}.
    \item[(C4)] (\textbf{Optimization initialization})  
    If an iterative solver (e.g., quasi-Newton) is used, the initial guess lies within the local basin of attraction and trust-region or line-search damping is applied.
\end{enumerate}

Under these conditions, the solution to the structured Taylor fitting problem is unique, numerically stable, and exhibits robust convergence under noise.
\end{lemma}

\begin{proof}[Sketch]
Condition (C1) ensures that the number of linear constraints $n \cdot M$ is no smaller than the number of unknowns $P_m = \sum_k n \cdot N_k$, and the monomial matrix is full rank, so the least-squares solution exists and is unique.

Conditions (C2) and (C3) guarantee that the monomial basis is well-excited and that the system is not ill-conditioned, preventing numerical amplification of noise. These are standard assumptions in polynomial approximation and inverse problem stability.

Condition (C4) ensures that for iterative solvers such as quasi-Newton methods, the optimization remains within a region where the cost surface is smooth and convex enough for convergence.

Together, these conditions guarantee the solvability and stability of the structured Taylor fitting process.
\end{proof}

The following lemma characterizes how the approximation order increases when analytic mappings are composed iteratively, as in our progressive fitting algorithm. It provides theoretical support for the expressiveness growth across iterations.

\begin{lemma}[Order Amplification in Composition of Analytic Mappings]\label{lemma:composition_order}
Let \( \tau_i : \mathbb{R}^n \to \mathbb{R}^n \) be a sequence of analytic mappings, each admitting an expansion of order at most \( O_i \) in the structured Taylor basis centered at \( \mathfrak{c}_i \). i.e.,
\[
\tau_i(y) = \sum_{k=0}^{O_i} \frac{1}{k!} \mathcal{J}^{[k]}(\tau_i)(\mathfrak{c}_i) \cdot \phi^{[k]}(y - \mathfrak{c}_i) + R_i(y),
\]
where \( \mathcal{J}^{[k]}(\tau_i)(\mathfrak{c}_i) \in \mathbb{R}^{n \times N_k} \) is the $k$-th order generalized derivative matrix, and \( \phi^{[k]}(y - \mathfrak{c}_i) \in \mathbb{R}^{N_k} \) is the corresponding generalized monomial vector.

Then, the composition
\[
\tau = \tau_n \circ \cdots \circ \tau_1
\]
is also analytic, and its Taylor expansion around the induced center admits a maximum order bounded by
\[
O = \prod_{i=1}^{n} O_i.
\]
\end{lemma}

\vspace{2mm}

\begin{proof}[Sketch]
The closure of analyticity under composition ensures that the composite function \( \tau \) remains analytic. The expansion of \( \tau = \tau_n \circ \cdots \circ \tau_1 \) can be systematically derived using the multivariate version of the Fa\`a di Bruno formula. Each composition multiplies the highest contributing monomial degree from the previous layer by the current layer's maximal order, leading to a total order bounded by the product \( \prod_i O_i \). The structure of each expansion in terms of generalized derivative matrices and monomial vectors is preserved, though the resulting coefficients grow combinatorially in complexity. Detailed coefficient forms are omitted here.
\end{proof}

This lemma demonstrates that iterative fitting using low-order Taylor approximations achieves exponential expressive growth. For example, using second-order expansions \( O_i = 2 \), the final approximation \( \tau \) effectively spans the function space of \( 2^n \)-order analytic maps. This enables our algorithm to approximate highly nonlinear deformations while maintaining per-iteration efficiency and avoiding the instability of direct high-order fitting.

\begin{remark}{operating regimes}
Low-order ($m\!=\!1$–$2$) Taylor models are accurate in several common settings:
(i) \emph{planar patterns under mild tilt}, where a homography is well-approximated by affine\,+\,quadratic when $|a_6 x + a_7 y|\!\ll\!1$;
(ii) \emph{camera lens distortion} near the principal point, where classical radial\,+\,tangential models (Brown--Conrady, division model, OpenCV) are analytic or rational and admit stable low-order Taylor approximations in $(x,y)$;
(iii) \emph{small, low-curvature surface patches}, where the in-plane warp has $O(r^3)$ remainder for patch radius $r$ and maximum curvature $\kappa_{\max}$ with $\kappa_{\max} r^2\!\ll\!1$.
Outside these regimes (e.g., strong fisheye, large nonlocal displacements, or heterogeneous warps), we \emph{do not} rely on a single global polynomial; instead we use either the restricted projective refinement (Sec.~\ref{sec:regist}) or \emph{structured Taylor lifting with a few centers} (Sec.~\ref{sec:experiment}), activating higher order only near convergence. This strategy retains numerical stability while preserving approximation power.
\end{remark}

\section{Fast Registration Using the AMVFF Framework}
\label{sec:regist}

AMVFF can be embedded into various registration frameworks to enable efficient and robust non-rigid point set registration, including ICP, RPM, CPD, BCPD, and others. In this work, we present a specific instantiation within the classical ICP framework, implementing a fast non-rigid registration algorithm named Analytic Iterative Closest Point (Analytic-ICP). Unlike the core theory in Section~\ref{sec:main}, where the analytic mapping is assumed known and correspondences are fixed, here we explicitly compute point correspondences via nearest-neighbor search, as is standard in ICP.

We alternate between (i) updating correspondences—either hard nearest-neighbor (closest-point) ICP matches or soft assignments (RPM/CPD-style)—and (ii) optimizing the AMVFF objective in Section~\ref{sec:main} conditioned on those correspondences to update the mapping parameters.

\subsection{Motivation}

Most non-rigid registration algorithms are built upon the theory of reproducing kernels~\cite{maiseli2017recent,2022Survey}, constructing deformation fields using local polynomial, spline, or Gaussian basis functions. While these models are expressive enough to capture large deformations, they are prone to overfitting in the presence of singularities, as highlighted in manifold spline theory~\cite{gu2005manifold}. Moreover, the computational complexity of non-rigid registration is significantly higher than that of rigid registration. In cases involving only minor deformations, existing non-rigid frameworks often fail to achieve both robustness and efficiency.

However, high efficiency, accuracy, and strong transferability are indispensable requirements in industrial production~\cite{bergstrom2011repeated}. High detection accuracy is essential for achieving the goal of ``replacing humans with machines,'' which not only improves automation but also reduces the costs associated with quality control and facilitates broader technological adoption. 

Moreover, to ensure that production throughput is not compromised, the algorithm must respond within a sufficiently low latency. Transferability is equally critical, as the algorithm may need to be seamlessly adapted across different production lines and heterogeneous operating environments.

\subsection{Analytic-ICP: Point Set Registration via AMVFF}\label{sec:analyticicp}

ICP is used to update point correspondences, while AMVFF is used to fit a smooth transport map conditioned on the current correspondences. After a rigid and an affine stage, the pipeline branches into two alternative refinement routes:
\begin{equation}\label{eq:orthogonal2Analytic}
Y \xrightarrow{\tau_{\mathrm{rig}}} \widehat{Y}_{\mathrm{rig}}
\xrightarrow{\tau_{\mathrm{aff}}} \widehat{Y}_{\mathrm{aff}}
\xrightarrow{\;\;\text{either}\;\;}
\begin{cases}
\widehat{Y}_{\mathrm{proj}}, & \text{(Route A: restricted projective refinement, } \widehat{Y}_{\mathrm{proj}}=\tau_{\mathrm{proj}}(\widehat{Y}_{\mathrm{aff}})\text{)}\\[2pt]
\widehat{Y}^{(n)}, & \text{(Route B: AMVFF / Taylor lifting, } \widehat{Y}^{(n)}=(\tau^{(n)}\circ\cdots\circ\tau^{(2)})(\widehat{Y}_{\mathrm{aff}})\text{)}
\end{cases}
\end{equation}

\noindent
Here $\widehat{Y}^{(n)}=\tau^{(n)}(\widehat{Y}_{\mathrm{aff}})$ denotes the transformed set under the $n$-th order Taylor model.

In practice, we start from a low truncation order (typically $m=2$) after the rigid/affine initialization, and increase $m$ up to the cap $m_w$ only when the residual stagnates. Although each AMVFF update fits only a low-order Taylor model, the approximation capacity of the overall deformation grows through iterative composition (cf.\ Lemma~\ref{lemma:composition_order}): the effective polynomial degree of the composed map is bounded by the product of the per-iteration truncation orders. This improves final accuracy while keeping each individual fitting step numerically stable and computationally lightweight.

\begin{algorithm}[t]
  \caption{Staged rigid/affine (optional projective) initialization used in Analytic-ICP}
  \label{alg:offline}
  \begin{algorithmic}[1]
    \Require 
      $X$~(fixed point set), 
      $Y$~(moving point set), 
      $\bar{\nu}$~(residual threshold), 
      $\psi$~(maximum iteration steps, typically $\leq 4$)
    \Ensure 
      $\mathcal{C}$~(correspondence),
      $(\Re,\mathbf t)$~(rigid estimate),
      $(R,\delta \mathbf t_A)$~(affine refinement),
      $(\Lambda,\delta \mathbf t_P)$~(optional projective refinement)

    \Procedure{StageInit}{$X$, $Y$}
      \State Initialize: first-order term $\gets$ identity; all other parameters $\gets 0$; $i \gets 1$
      \While{residual $\hat{v}$ from~\eqref{eq:adjustment_result} $\geq \bar{\nu}$ and $i < \psi$}
        \State Compute $\mathcal{C}$ and rigid $(\Re,\mathbf t)$ via a standard ICP step (NN correspondence by default)
        \If{$i > 1$}
          \State Estimate affine refinement $(R,\delta \mathbf t_A)$ (see Sec.~\ref{sec:linear})
        \EndIf
        \If{$i > 2$}
          \State Optionally estimate projective refinement $(\Lambda,\delta \mathbf t_P)$ (see Sec.~\ref{sec:linear})
        \EndIf
        \State $i \gets i + 1$
      \EndWhile
    \EndProcedure
  \end{algorithmic}
\end{algorithm}

\paragraph{Coarse-to-fine linear stages and a branching refinement}
For numerical stability, we first perform a coarse-to-fine linear initialization by estimating a rigid transform $\tau_{\mathrm{rig}}$ followed by an affine refinement $\tau_{\mathrm{aff}}$.
After the affine stage, we choose one of two refinement routes depending on the deformation regime:
(i) a restricted projective refinement $\tau_{\mathrm{proj}}$ for planar patterns under perspective effects; or
(ii) AMVFF-based Taylor lifting, where we apply a sequence of truncated Taylor updates to progressively increase the effective approximation order up to $n$.
Details of the linear parameterization and the translation update are given in Section~\ref{sec:linear}.

 \begin{algorithm}
          \caption{Fitting of mapping parameters ($a_{1},\cdots , a_{N}$) in $n$-dimensional $m$-order Truncated Taylor Series of Vector-Valued Functions (refer to \equref{eq:adjustment_form})}
          \label{alg:mappingParam}
          \algorithmicrequire{\hspace{5mm}
             $N \equiv {\textstyle \sum_{d=0}^{m}} \left ( n\cdot C_{d+n-1}^{n-1}  \right )$ \dots the number of parameters  \newline
            \textcolor{myWhite}{.\hspace{12mm}} $X$ \dots the fixed point set\newline  
            \textcolor{myWhite}{.\hspace{13mm}}$Y$ \dots the moved point set\newline  
            \textcolor{myWhite}{.\hspace{13mm}}$m$ \dots the order of Taylor series\newline 
            \textcolor{myWhite}{.\hspace{13mm}}$\bar{\nu}$ \dots threshold for residual\newline  
            \textcolor{myWhite}{.\hspace{13mm}}$\psi$ \dots the number of iterations\newline  
            \textcolor{myWhite}{.\hspace{13mm}}$a_{1},\cdots , a_{N}$ \dots mapping parameters\newline  
            \textcolor{myWhite}{.\hspace{13mm}}$\mathfrak{c}_{1},\cdots , \mathfrak{c}_{n}$ \dots location parameters\newline  
          }
          \algorithmicensure{\hspace{2mm} $a_{1},\cdots , a_{N}$}
          \begin{algorithmic}[1]
                \Procedure{quasi-newton solver}{}

                \State Initialization:\newline  
              \textcolor{myWhite}{.\hspace{9mm}}$N$-order matrix $H_{k}$ $\gets$ identity matrix \newline
                \textcolor{myWhite}{.\hspace{9mm}}$N$-order matrix $U_{k}$ $\gets$ zero matrix \newline
                \textcolor{myWhite}{.\hspace{9mm}}$N$-order vector $h_{k}$ $\gets$ zero vector \newline
                \textcolor{myWhite}{.\hspace{9mm}}$i  \gets 1$
                 \While {$\hat{v}$ of \equref{eq:adjustment_result} is not less than $\bar{\nu}$ and $i<\psi$}

                 \State $H_{k}$ $\gets$ $H_{k}+U_{k}$

                \State $\hat{v}$ $\gets$  execute \equref{eq:error_form} and \equref{eq:normal_equation} 

               \State $h_{k}$ $\gets$ $H_{k}\hat{v}$

               \State $\hat{X}$ $\gets$ $X^{o}+h_{k}$ 

                \State $U_{k}$ $\gets$ execute BFGS algorithm

                \State $i++$
                \EndWhile
               \State $a_{1},\cdots , a_{N} \gets \hat{X}$
                \EndProcedure
          \end{algorithmic}
      \end{algorithm}

Algorithm~\ref{alg:mappingParam} estimates the coefficients of the truncated multivariate Taylor model using a quasi-Newton scheme, where the BFGS update is used to accelerate convergence.
When a higher-order layer is activated, we optimize only the newly introduced parameters while keeping the previously fitted components fixed; this staged strategy improves stability
and avoids redundant recomputation.

\begin{algorithm}
          \caption{Fitting of location parameters ($\mathfrak{c}_{1},\cdots , \mathfrak{c}_{n}$) in $n$-dimensional $m$-order Truncated Taylor Series of Vector-Valued Functions (refer to \equref{eq:adjustment_form})}
          \label{alg:locationParam}
          \algorithmicrequire{\hspace{5mm}
               $X$ \dots the fixed point set\newline  
            \textcolor{myWhite}{.\hspace{13mm}}$Y$ \dots the moved point set\newline  
            \textcolor{myWhite}{.\hspace{13mm}}$m$ \dots the order of Taylor series\newline 
            \textcolor{myWhite}{.\hspace{13mm}}$\bar{\nu}$ \dots threshold for residual\newline  
            \textcolor{myWhite}{.\hspace{13mm}}$\psi$ \dots the number of iterations\newline  
            \textcolor{myWhite}{.\hspace{13mm}}$a_{1},\cdots , a_{N}$ \dots mapping parameters\newline  
            \textcolor{myWhite}{.\hspace{13mm}}$\mathfrak{c}_{1},\cdots , \mathfrak{c}_{n}$ \dots location parameters\newline  
          }
          \algorithmicensure{\hspace{2mm} $\mathfrak{c}_{1},\cdots , \mathfrak{c}_{n}$}
          \begin{algorithmic}[1]
                \Procedure{linear least square solver}{}

                \State Initialization:\newline  
                \textcolor{myWhite}{.\hspace{9mm}}$i  \gets 1$
                 \While {$\hat{v}$ of \equref{eq:adjustment_result} is not less than $\bar{\nu}$ and $i<\psi$}

                \State $\hat{v}$ $\gets$  execute \equref{eq:error_form} and \equref{eq:normal_equation} 

               \State $\hat{X}$ $\gets$ $X^{o}+\hat{v}$ \Comment{refer to \equref{eq:adjustment_result}}

                \State $i++$
                \EndWhile
               \State $\mathfrak{c}_{1},\cdots , \mathfrak{c}_{n} \gets \hat{X}$
                \EndProcedure
          \end{algorithmic}
      \end{algorithm}

The final output of Algorithm~\ref{alg:analyticcode} is a \emph{composition} of multiple low-order
Taylor mappings (one produced per outer iteration). As these analytic maps are composed, the
effective polynomial degree of the composite mapping grows rapidly; see Lemma~\ref{lemma:composition_order}.
Importantly, each outer iteration fits only a truncated model of modest size involving
\(
N = {\textstyle \sum_{d=0}^{m}} \bigl( n \cdot C_{d+n-1}^{\,n-1} \bigr)
\)
parameters, so the per-iteration optimization remains lightweight. This progressive lifting
avoids an aggressively high-order global fit at early iterations and activates additional
expressiveness only as the residual decreases.

\begin{algorithm}[t]
\caption{Analytic Iterative Closest Point}
\label{alg:analyticcode}
\algorithmicrequire{\hspace{5mm}
  $N \equiv {\textstyle \sum_{d=0}^{m}} \bigl( n \cdot {C}_{d+n-1}^{\,n-1} \bigr)$ \dots number of parameters\newline
  \textcolor{myWhite}{.\hspace{13mm}} $X$ \dots fixed point set \qquad
  $Y$ \dots moved point set\newline
  \textcolor{myWhite}{.\hspace{13mm}} $\phi$ \dots RMSE threshold \qquad
  $m_w$ \dots order cap\newline
  \textcolor{myWhite}{.\hspace{13mm}} $\psi$ \dots max outer iterations \qquad
  $k_{\mathrm{step}}$ \dots order--lifting period (default $=2$; $=1$ means every iter)
}

\algorithmicensure{\hspace{2mm} \newline
  \textcolor{myWhite}{.\hspace{9mm}}$\mathcal{C}$ \dots one-to-one correspondence (final)\newline
  \textcolor{myWhite}{.\hspace{9mm}}
  $(a_{1}^{(1)},\ldots,a_{N}^{(1)}),\,\ldots,\,(a_{1}^{(\psi)},\ldots,a_{N}^{(\psi)})$ \dots mapping parameters across outer iterations\newline
  \textcolor{myWhite}{.\hspace{9mm}}
  $(\mathfrak{c}_{1}^{(1)},\ldots,\mathfrak{c}_{n}^{(1)}),\,\ldots,\,
   (\mathfrak{c}_{1}^{(\psi)},\ldots,\mathfrak{c}_{n}^{(\psi)})$ \dots centers across outer iterations
}
\begin{algorithmic}[1]
\Procedure{Alternate iteration solver}{}
  \State $i \gets 1$
  \While {$i \le \psi$}
    \State $a_{1}^{\mathbf{(i)}},\ldots , a_{N}^{\mathbf{(i)}} \gets 0$
    \State $a_{n}^{\mathbf{(i)}},\ldots , a_{n^2+n-1}^{\mathbf{(i)}} \gets I$ \Comment{first-order block init. to identity}
    \State $\mathfrak{c}_{1}^{\mathbf{(i)}},\ldots , \mathfrak{c}_{n}^{\mathbf{(i)}} \gets 0$
    \State $i \gets i{+}1$
  \EndWhile
  \State $X, Y \gets$ data normalization
  \State $i \gets 1$;\quad \textbf{deg} $\gets 2$
  \State $\mathrm{err} \gets \textsc{RMSE}(X, Y)$ \Comment{mean NN distance (or w.r.t.\ current $\mathcal{C}$)}

  \While {$\mathrm{err} \ge \phi$ \textbf{and} $i<\psi$}

    \State $\mathcal{C},\, \Re,\, \mathbf t,\, R,\, \delta \mathbf t_{A} \gets$ execute Algorithm~\ref{alg:offline}
    \State $a_{1}^{\mathbf{(i)}},\ldots , a_{N}^{\mathbf{(i)}} \gets \Re,\, \mathbf t,\, R,\, \delta \mathbf t_{A}$

    \If{$i>1$}
      \If{$i \bmod k_{\mathrm{step}} = 0$}
        \State \textbf{deg} $\gets \min(\textbf{deg}{+}1,\, m_w)$
      \EndIf

      \If{$m_{w} \ge i$}
        \State $a_{1}^{\mathbf{(i)}},\ldots , a_{N}^{\mathbf{(i)}} \gets$ execute Algorithm~\ref{alg:mappingParam} \ \ ($m \gets$ \textbf{deg})
        \State $\mathfrak{c}_{1}^{\mathbf{(i)}},\ldots , \mathfrak{c}_{n}^{\mathbf{(i)}} \gets$ execute Algorithm~\ref{alg:locationParam} \ \ ($m \gets$ \textbf{deg})
      \Else
        \State $a_{1}^{\mathbf{(i)}},\ldots , a_{N}^{\mathbf{(i)}} \gets$ execute Algorithm~\ref{alg:mappingParam} \ \ ($m \gets m_w$)
        \State $\mathfrak{c}_{1}^{\mathbf{(i)}},\ldots , \mathfrak{c}_{n}^{\mathbf{(i)}} \gets$ execute Algorithm~\ref{alg:locationParam} \ \ ($m \gets m_w$)
      \EndIf
    \EndIf

    \State \textbf{$Y_{\mathrm{nr}} \gets \textsc{ApplyTaylor}(a^{\mathbf{(i)}}, \mathfrak c^{\mathbf{(i)}},\, \textbf{deg};\, Y)$}
    \State \textbf{$Y \gets Y_{\mathrm{nr}}$} \Comment{update moved set by current (composed) mapping}
    \State $\mathrm{err} \gets \textsc{RMSE}(X, Y)$ \Comment{recompute mean distance after update}
    \State $i \gets i{+}1$
  \EndWhile
\EndProcedure
\end{algorithmic}
\end{algorithm}

Another advantage of \textbf{Analytic-ICP} is its minimal dependence on tunable thresholds, which
improves transferability across datasets and operating conditions. In Algorithm~\ref{alg:analyticcode},
only a small set of control parameters is required in practice: the stopping tolerance $\phi$ (measured
by RMSE), the Taylor truncation order (and its cap $m_w$), and the maximum number of outer iterations
$\psi$ (with an optional order-lifting period $k_{\mathrm{step}}$). Unless otherwise stated, the residual
used for convergence checking is the RMSE between the fixed set $X$ and the currently transformed moving
set $Y$ under the current correspondence $\mathcal{C}$.

The time complexity of the proposed Analytic-ICP method is as follows:
\begin{eqnarray}\label{eq:complexity}
\begin{aligned}
\mathcal{O}\left( N_{r}\log N_{r} + p^{2}N_{r} + p^{3} \right),
\end{aligned}
\end{eqnarray}
where $N_{r}$ is the number of points in the reference point set and $p$ is the number of parameters to be fitted. The first term arises from the nearest neighbor search (e.g., using KD-tree), the second from computing and updating mapping parameters across all point pairs, and the third from BFGS-based global optimization in parameter space.

In contrast, classical methods such as CPD and TPS-RPM exhibit time complexity on the order of $\mathcal{O}(N_{r}^{3})$ due to dense kernel matrix construction and inversion. Additionally, the spatial complexity of Analytic-ICP is approximately $\mathcal{O}(N_{r}p)$, while CPD and TPS-RPM require $\mathcal{O}(N_{r}^{2})$ memory, which becomes prohibitive for large-scale data. Since in practice $N_r \gg p$, our method is highly scalable.

\vspace{1mm} \noindent\textbf{Remark.}
In practice, the spatial complexity becomes a bottleneck in large-scale scenarios. For instance, CPD and TPS-RPM often fail or significantly slow down when $N_r > 10^5$ due to their $\mathcal{O}(N_r^2)$ memory requirements. In contrast, Analytic-ICP successfully handles such large point sets on modest hardware (e.g., a laptop with limited memory), thanks to its $\mathcal{O}(N_r \cdot p)$ spatial complexity.

\subsection*{Summary of Practical Value of Analytic-ICP}

Before moving to experimental evaluations, we briefly summarize the key advantages of the proposed \textbf{Analytic-ICP} algorithm from a practical perspective.

Unlike conventional non-rigid registration algorithms that rely on intricate kernel designs or parameter-heavy models, Analytic-ICP exhibits a unique combination of mathematical simplicity and engineering utility:

\begin{itemize}
    \item It achieves \textbf{fine-grained control of deformation complexity} through the iterative composition of low-order Taylor mappings, enabling fast and stable approximation of complex geometric transformations.
    \item Its modular architecture, based on \textbf{orthogonal $\to$ affine $\to$ perspective} decomposition, ensures that each stage is computationally lightweight, interpretable, and separable.
    \item The algorithm has \textbf{minimal dependence on hyperparameters} — with only a handful of thresholds controlling residual, series order, and iteration count — making it well-suited for industrial applications with high demands on robustness and transferability.
    \item Thanks to its analytic formulation and convergence-friendly design, the method delivers reliable results even on challenging datasets, without the need for delicate tuning or data-specific heuristics.
\end{itemize}

\subsection{Progressive Lifting with Residual-Gated Order Activation}

To balance the expressive power of high-order structured Taylor models with numerical stability and computational control, we adopt a progressive lifting strategy combined with residual-based gating.

Instead of activating all Taylor terms up to a fixed order $m$ from the start, we incrementally increase the expansion order in a staged manner. At each step, higher-order terms are activated only when the residual error from the current approximation becomes sufficiently small. This design reflects the empirical observation that early activation of high-order polynomial terms can destabilize the iteration process, whereas controlled late-stage activation improves final precision.

Let $T_m(y)$ denote the structured Taylor approximation of order $m$, and define the current residual:
\[
\mathcal{R}^{(m)}(y) = \|f(y) - T_m(y)\|.
\]
We introduce a residual gating function such that the $(m+1)$-th order derivative matrix $\mathcal{J}^{[m+1]}$ is activated only if:
\[
\mathcal{R}^{(m)} < \delta_m,
\]
where $\delta_m = \delta_0 \cdot \alpha^m$ with fixed decay rate $0 < \alpha < 1$. This ensures that:
\begin{itemize}
  \item The center $\mathfrak{c}$ is close to $y$, keeping $\|\phi^{[m]}(y - \mathfrak{c})\|$ under control;
  \item The deformation field has already stabilized in lower-order modes before injecting additional degrees of freedom;
  \item Numerical blow-up due to early-stage high-order oscillation is avoided.
\end{itemize}

From a theoretical perspective, this strategy is consistent with the composition behavior discussed in Lemma~\ref{lemma:composition_order}, where uncontrolled stacking of high-order maps can lead to exponential growth in polynomial degree. In contrast, the gated lifting mechanism introduces new basis terms only when they can safely contribute to refining the residual.

\begin{remark}
In practice, we find that orders $m=2$ or $3$ are often sufficient to capture the majority of geometric deformation, and higher orders should be enabled only near convergence. The residual-gated strategy leads to smoother convergence, avoids early-stage overfitting, and enhances robustness against noise and under-sampled regions.
\end{remark}

In the following section, we empirically evaluate the performance of Analytic-ICP and compare it with state-of-the-art baselines under both synthetic and real-world conditions.

\section{Empirical Validation of Analytic-ICP}
\label{sec:experiment}

In this section we empirically evaluate \textbf{Analytic-ICP} for accuracy, robustness, and runtime across rigid, affine, and smooth non-rigid deformations (including heterogeneous pairs) generated via truncated Taylor maps. We compare against two representative baselines—\textbf{CPD} and \textbf{TPS-RPM}—using both quantitative metrics and visual results.

Unless otherwise noted, experiments run on a standard laptop (Intel i5-6200U, 8\,GB RAM). Our implementation is single-threaded C++ (OpenCV + Eigen); ICP uses the public point-to-point variant by Andreas Geiger (2011). Before registration, source (moving) and target (fixed) point sets are zero-centered and scaled to unit size. For experiments executed on a different workstation (explicitly indicated in their captions/subsections), wall-clock times are comparable \emph{within} that subsection only.

\subsection{Stepwise Evaluation of Analytic-ICP}

To illustrate the theoretical and practical advantages of our Analytic-ICP framework, we performed a stepwise analysis of the fitting process. Specifically, instead of conventional ablation studies, we directly present intermediate fitting results obtained at each iteration stage. This step-by-step presentation clearly reveals how the proposed AMVFF-based method progressively improves registration accuracy by systematically expanding the mapping space—from the identity transformation, through rigid and affine mappings, and ultimately towards higher-order Taylor expansions for full non-rigid transformations.

Importantly, this hierarchical process not only highlights how lower-order expansions effectively approximate more complex transformations in a computationally efficient manner but also demonstrates the practical benefit of incremental fitting, where the algorithm retains efficiency comparable to affine methods even as complexity grows. This progression underscores the computational robustness of Analytic-ICP, emphasizing its capability to handle complex non-rigid deformation scenarios while maintaining the efficiency and reliability characteristic of rigid registration algorithms.

In the following sections, we present detailed experimental results in both 2D and 3D settings, demonstrating clearly how each incremental stage enhances registration quality and effectiveness.

\subsubsection{Performance of 2D Analytic-ICP}
We begin with 2D registration tasks to evaluate the stepwise deformation fitting capability of Analytic-ICP. The experiments demonstrate how the method transitions from rigid to high-order non-rigid mapping with stable convergence and increasing accuracy.

\begin{figure*}[htb]
  \centering
  \includegraphics[width=0.9\textwidth]{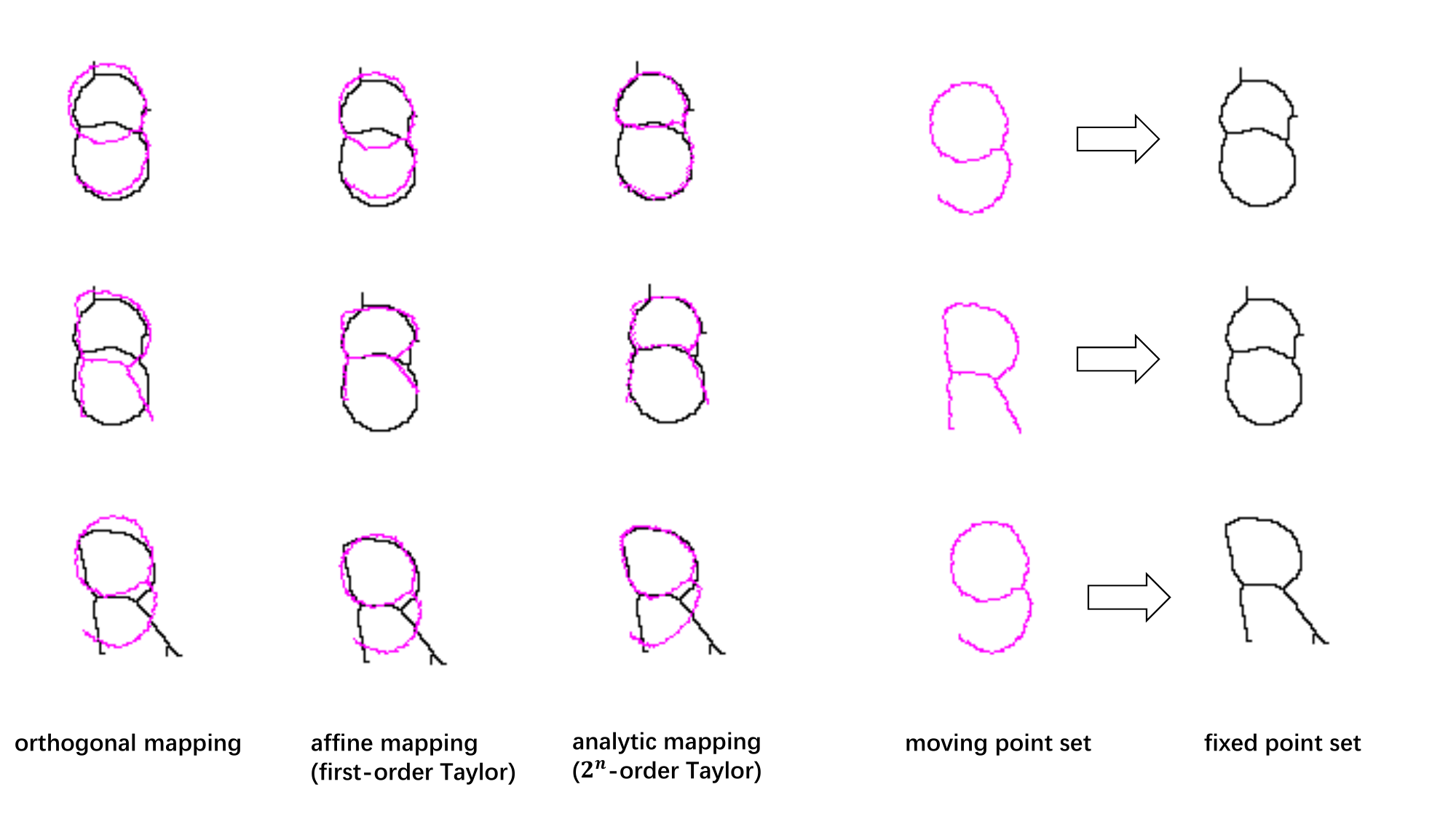}%
  \caption{\textbf{Performance of 2D stepwise Analytic-ICP.}
  The registration is decomposed into three stages: \emph{orthogonal} (rigid), \emph{affine}, and a \emph{structured analytic} refinement.
  The analytic stage is implemented as a composition of \emph{quadratic} (second-order) Taylor lifts, yielding an overall low-order but expressive mapping.
  In each row, the right panel shows the fixed (red) and moving (green) point sets; the left panels visualize the intermediate results after each stage.}
  \label{fig:decompo}
\end{figure*}

\paragraph{For heterogeneous pairs of typical point sets} 
As illustrated in Fig.~\ref{fig:decompo}, the 2D Analytic-ICP proceeds with order lifting across stages: 
each stage uses a second-order structured Taylor map\footnote{Composing $n$ second-order stages yields an \emph{upper bound} on the effective polynomial degree of $\le 2^{n}$. 
This is only a worst-case bound; cancellations often reduce the practical degree. In our runs $n$ is small.} 
and the maps are composed across $n$ stages.

\begin{figure*}[htbp]
\centering
\scriptsize
\centerline{Increasing the (effective) order via \emph{stagewise order lifting} reduces residuals.}
\vspace{2mm}

\begin{minipage}{0.15\linewidth}\centering
\includegraphics[width=\linewidth]{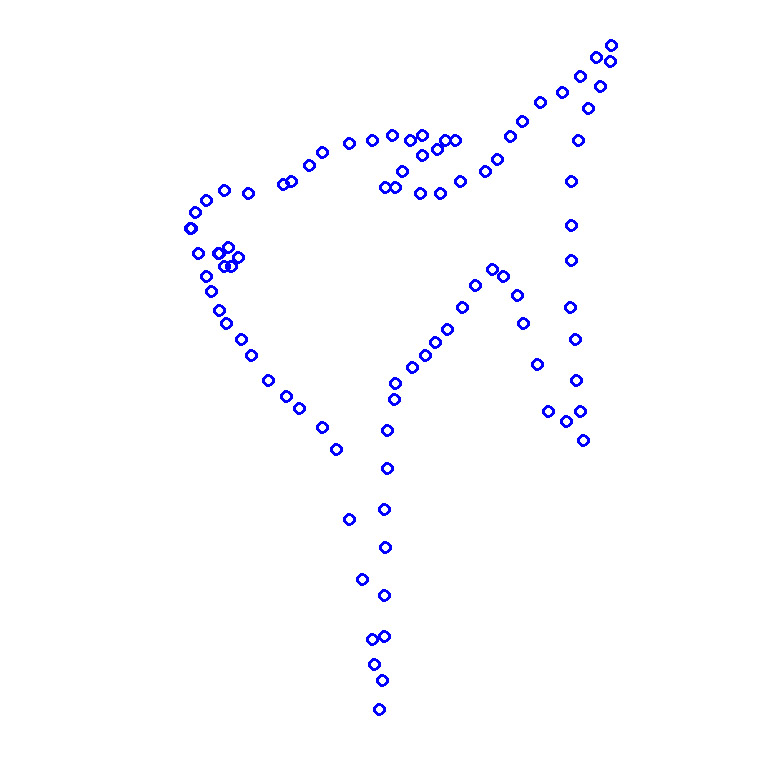}
\end{minipage}\hspace{1mm}
\begin{minipage}{0.15\linewidth}\centering
\includegraphics[width=\linewidth]{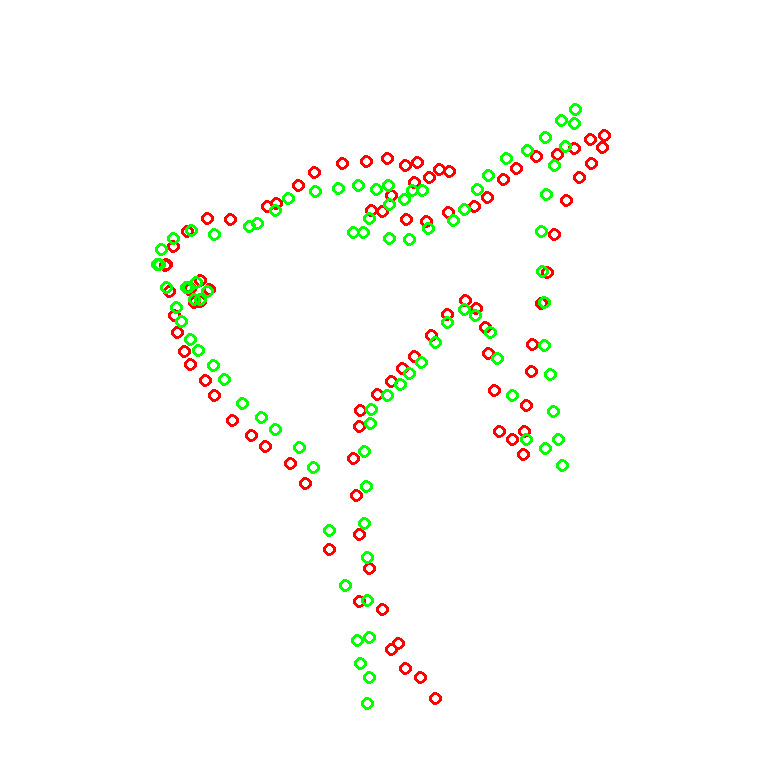}
\end{minipage}\hspace{1mm}
\begin{minipage}{0.15\linewidth}\centering
\includegraphics[width=\linewidth]{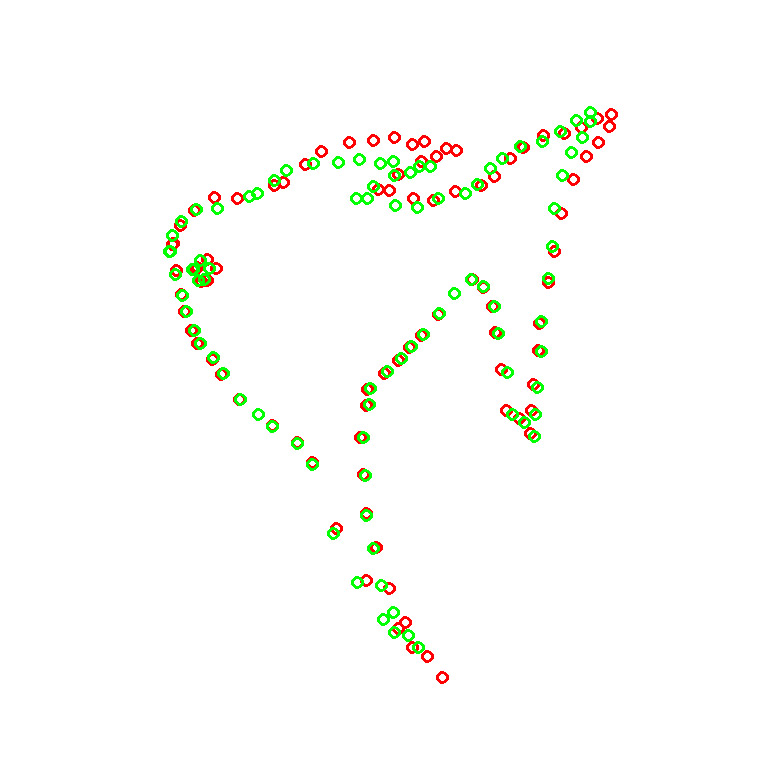}
\end{minipage}\hspace{1mm}
\begin{minipage}{0.15\linewidth}\centering
\includegraphics[width=\linewidth]{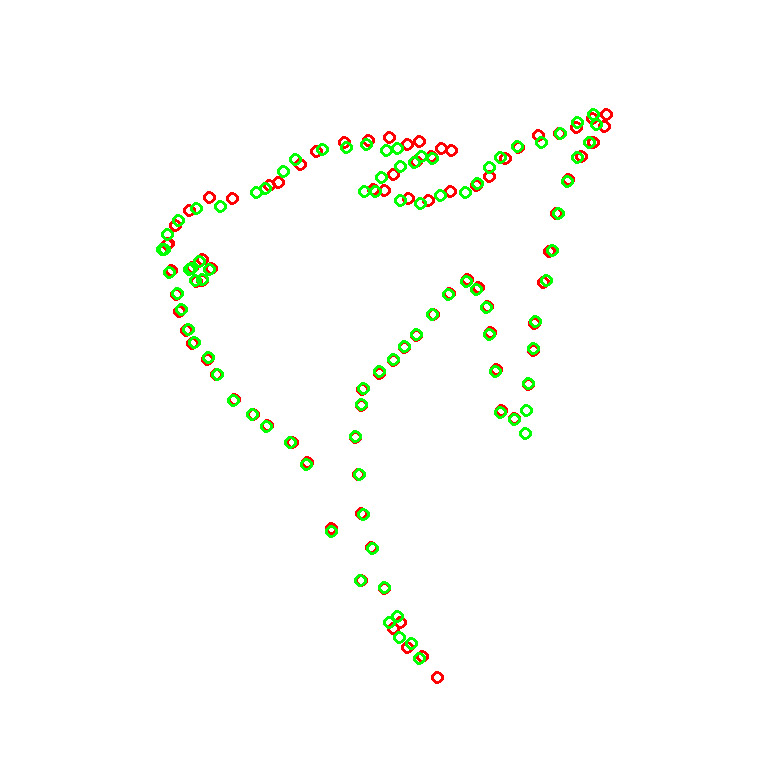}
\end{minipage}\hspace{1mm}
\begin{minipage}{0.15\linewidth}\centering
\includegraphics[width=\linewidth]{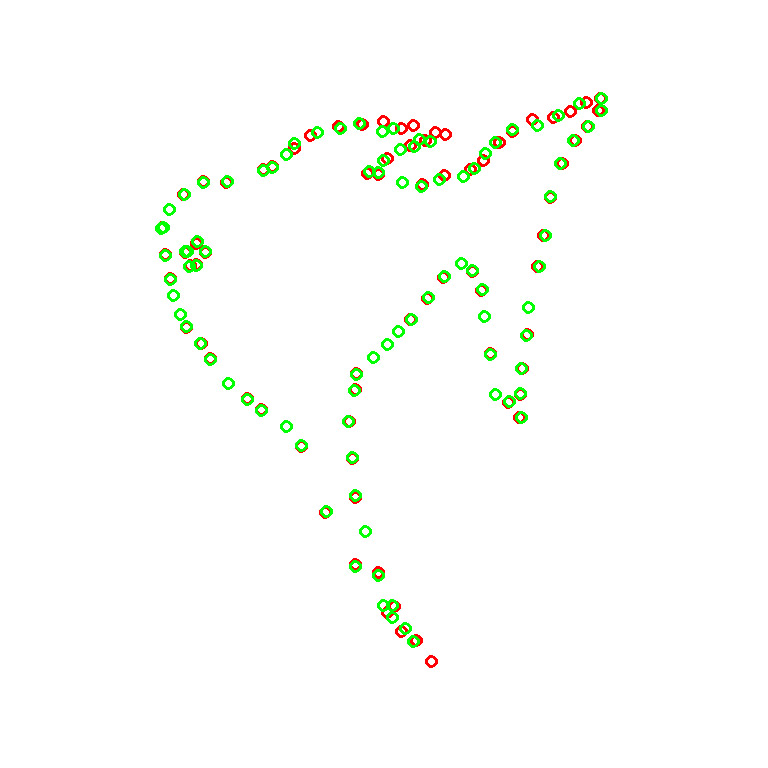} 
\end{minipage}\hspace{1mm}
\begin{minipage}{0.15\linewidth}\centering
\includegraphics[width=\linewidth]{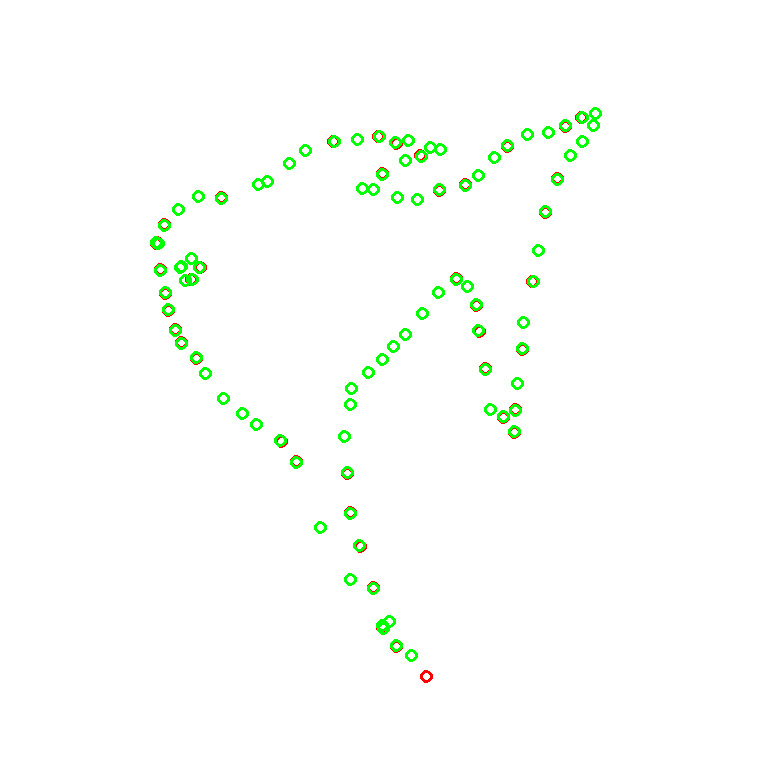}
\end{minipage}

\vspace{2mm}

\begin{minipage}{0.15\linewidth}\centering
\includegraphics[width=\linewidth]{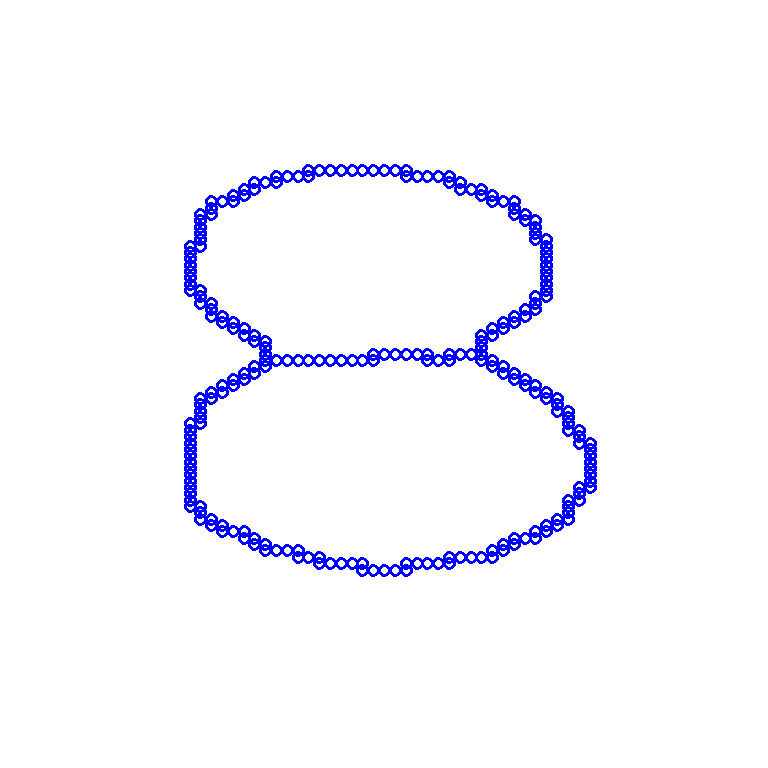}
\end{minipage}\hspace{1mm}
\begin{minipage}{0.15\linewidth}\centering
\includegraphics[width=\linewidth]{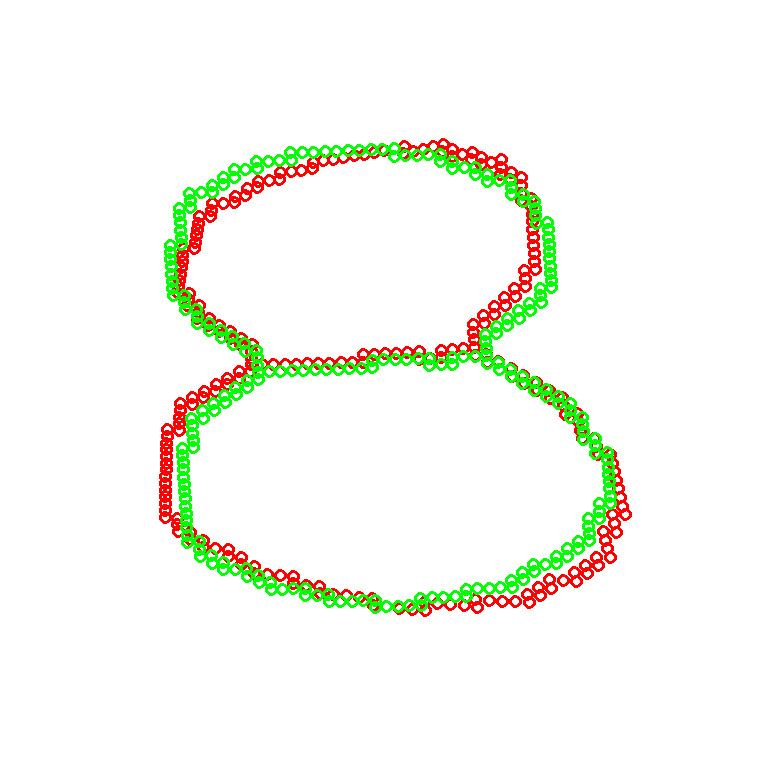}
\end{minipage}\hspace{1mm}
\begin{minipage}{0.15\linewidth}\centering
\includegraphics[width=\linewidth]{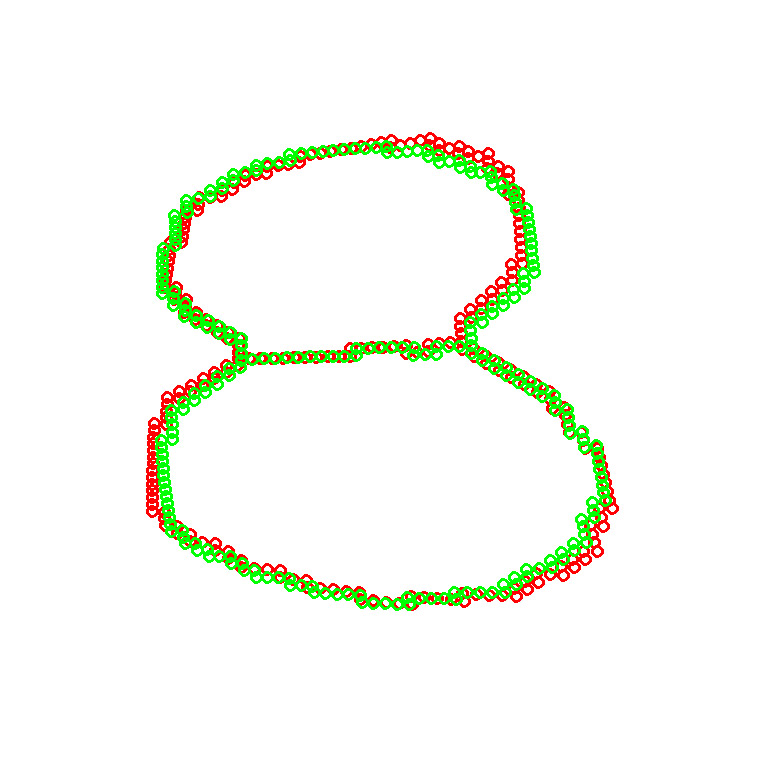} 
\end{minipage}\hspace{1mm}
\begin{minipage}{0.15\linewidth}\centering
\includegraphics[width=\linewidth]{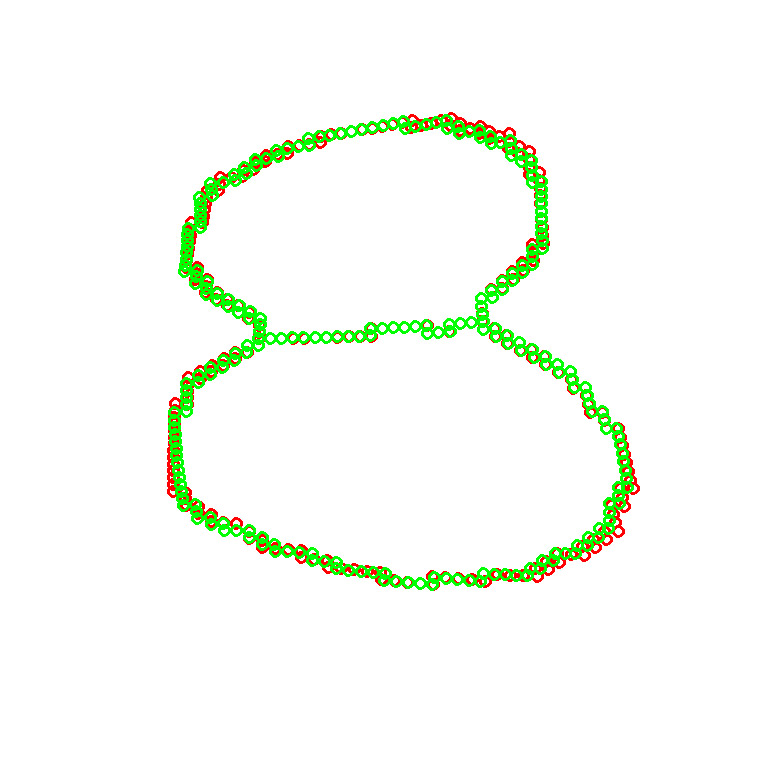}
\end{minipage}\hspace{1mm}
\begin{minipage}{0.15\linewidth}\centering
\includegraphics[width=\linewidth]{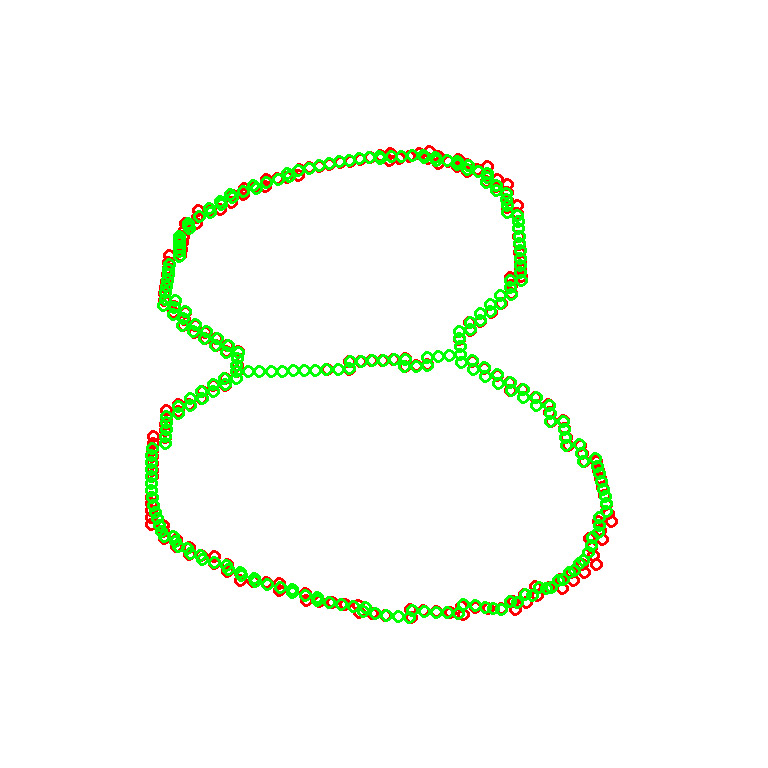}
\end{minipage}\hspace{1mm}
\begin{minipage}{0.15\linewidth}\centering
\includegraphics[width=\linewidth]{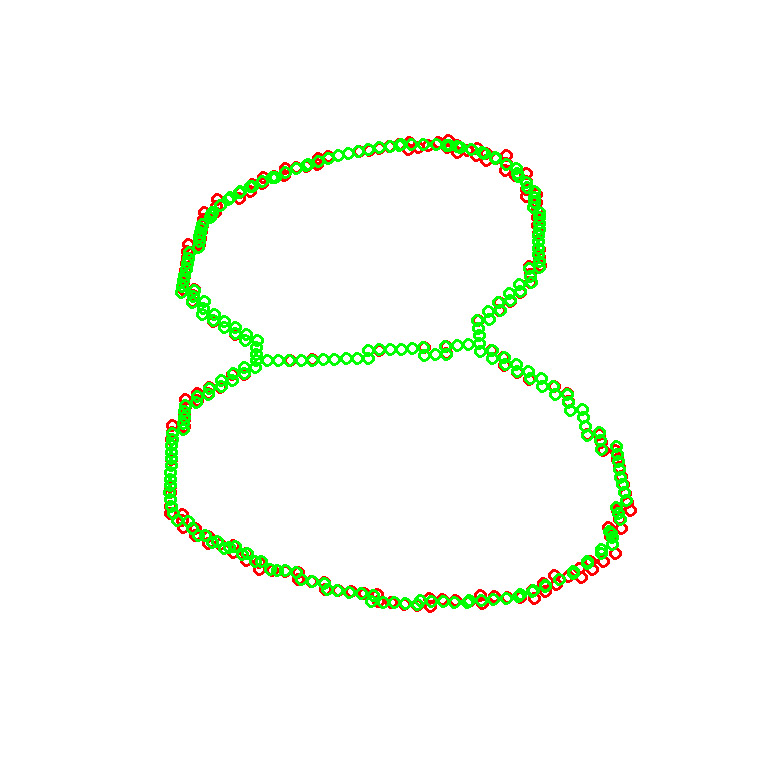}
\end{minipage}

\vspace{2mm}

\begin{minipage}{0.15\linewidth}\centering
\includegraphics[width=\linewidth]{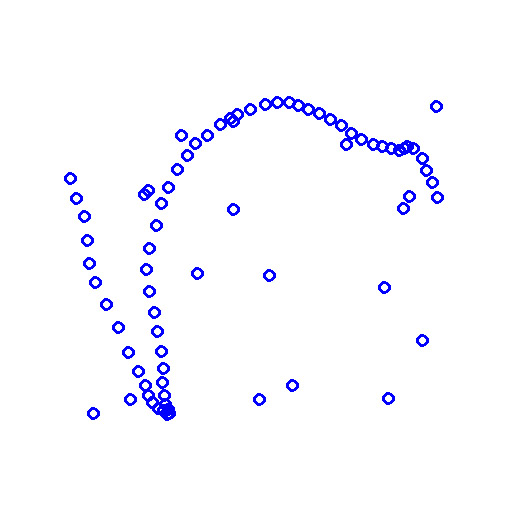}
\newline moving point set
\end{minipage}\hspace{1mm}
\begin{minipage}{0.15\linewidth}\centering
\includegraphics[width=\linewidth]{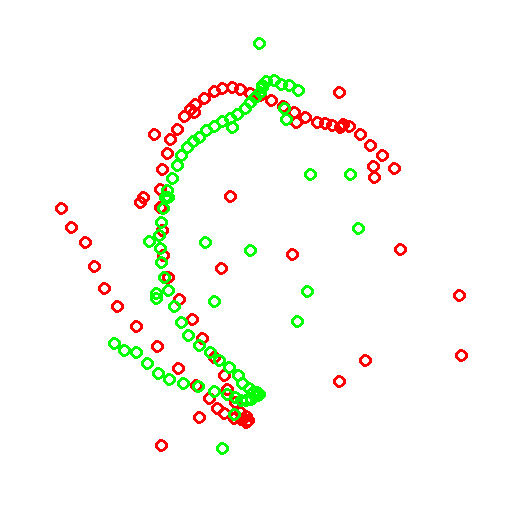}
\newline affine mapping
\end{minipage}\hspace{1mm}
\begin{minipage}{0.15\linewidth}\centering
\includegraphics[width=\linewidth]{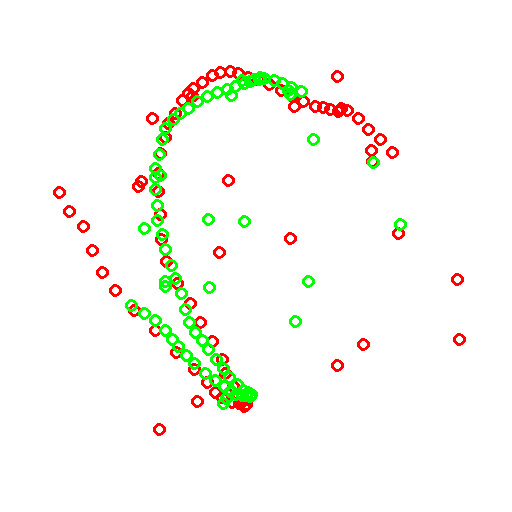}
\newline stage order: $3$
\end{minipage}\hspace{1mm}
\begin{minipage}{0.15\linewidth}\centering
\includegraphics[width=\linewidth]{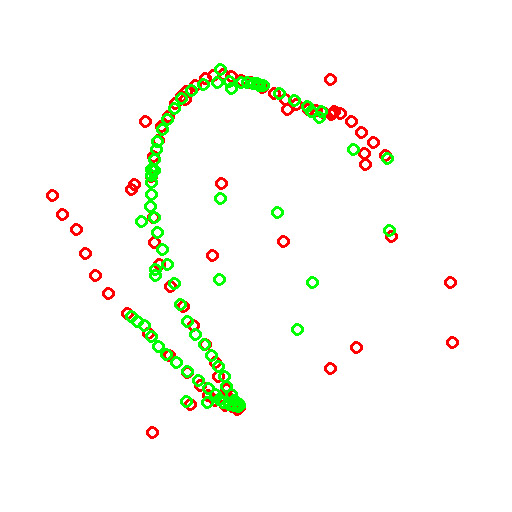}
\newline $2 \!\rightarrow\! 3 \!\rightarrow\! 4$ \\(eff.\ $\le 24$)
\end{minipage}\hspace{1mm}
\begin{minipage}{0.15\linewidth}\centering
\includegraphics[width=\linewidth]{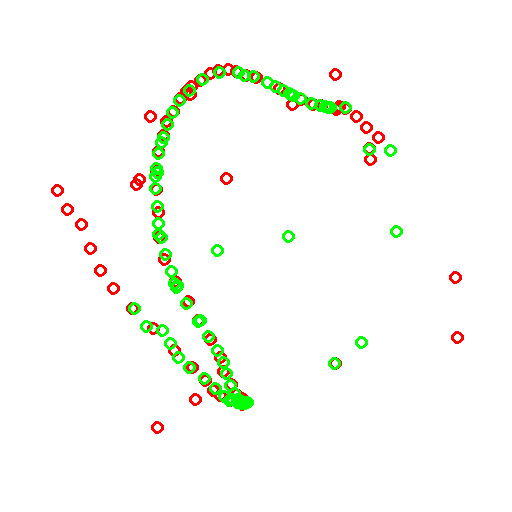}
\newline $2 \!\rightarrow\! 3 \!\rightarrow\! \cdots \!\rightarrow\! 6$\\ (eff.\ $\le 720$)
\end{minipage}\hspace{1mm}
\begin{minipage}{0.15\linewidth}\centering
\includegraphics[width=\linewidth]{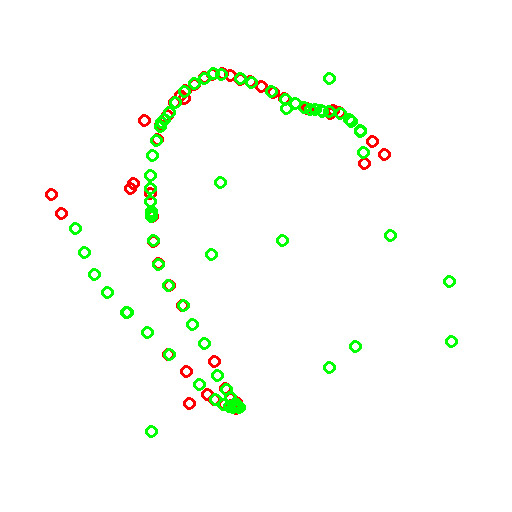}
\newline $2 \!\rightarrow\! 3 \!\rightarrow\! \cdots \!\rightarrow\! 8$ \\(eff.\ $\le 40320$)
\end{minipage}

\vspace{2mm}
\caption{\textbf{Residual visualization under stagewise order lifting.}
Each row shows a 2D point set (fish, numeral 8, curve). Left to right: moving set, affine result, and compositions of low-order stages with increasing \emph{stage orders} (labels indicate a possible order sequence and the corresponding \emph{effective degree upper bound} \(\le\!\prod_s o_s\)). Residuals decrease as effective order increases.}
\label{fig:difforder}
\end{figure*}

\paragraph{For isomorphic pairs of typical point sets} 
To further evaluate the approximation capability of Analytic-ICP, we investigate how the registration residual evolves as the order of the truncated Taylor series increases. Specifically, we consider three representative point cloud models: a \textit{fish}, the Arabic numeral \textit{8}, and a \textit{curve}. The fish and curve datasets are widely used in the point set registration literature, while the Arabic numeral 8 corresponds to the medial axis of the character ``8'' in bitmap images.

To generate test data, we apply a slight smooth distortion to each original model using a Taylor series transformation (see \equref{eq:adjustment_form}). The mapping coefficients $a_{0}, \cdots, a_{n^{2}+3n+1}$ are randomly sampled in the range $[-0.3, 0.3]$, with the exception that the diagonal entries of the first-order matrix are fixed at 1 to preserve approximate scale. Since the applied transformation is at least third-order, it introduces non-affine deformations. As a result, affine registration alone will leave substantial residual errors.

We then use Analytic-ICP to align the original point set to its distorted version, and observe the reduction of residual as the Taylor series order increases. This experiment provides a controlled setting to showcase the expressive power of the proposed framework in capturing subtle, smooth, and non-affine geometric variations.

\begin{figure}[!ht]
\centering
\includegraphics[width=5in]{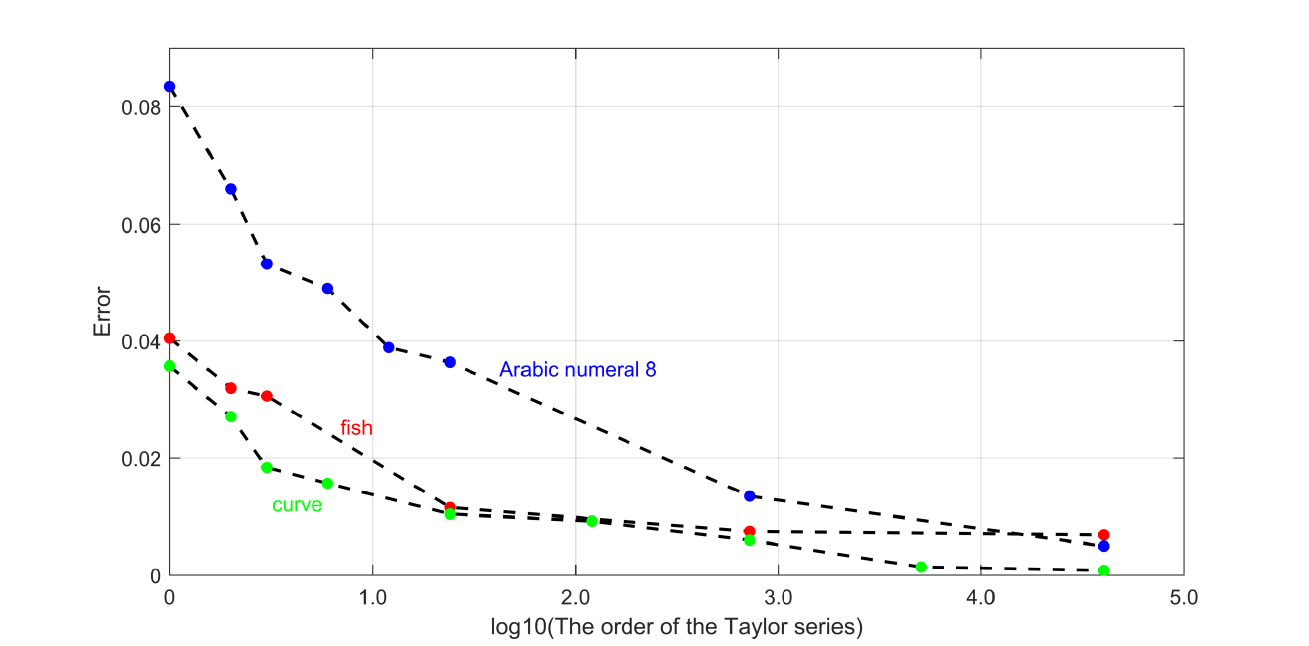}
\caption{
\textbf{Error vs. Order Curve for 2D Analytic-ICP.}  
This plot shows the registration residuals corresponding to different truncation orders of the Taylor series for the three 2D point sets presented in \figref{fig:difforder}. For improved readability, the horizontal axis uses a logarithmic scale (base 10), representing the Taylor series order. The results demonstrate a general trend of improved fitting accuracy as the order increases.
}
\label{fig:taylor_error}
\end{figure}

In Fig.~\ref{fig:difforder} we assess approximation quality as the model order increases.
We adopt \emph{stagewise order lifting}: at stage $s$ we fit a low-order structured Taylor map of order $o_s$
(e.g., $o_1{=}2,\,o_2{=}3,\,o_3{=}4,\ldots$) and compose the stages. The resulting
\emph{effective polynomial degree} of the composed map is upper-bounded by $\prod_s o_s$
(e.g., $2\!\times\!3\!\times\!4{=}24$), so increasing the stage orders (or the number of stages)
systematically enlarges model capacity.

Across the three heterogeneous 2D cases (fish, numeral ``8'', curve), the residuals consistently
decrease as the effective order grows, demonstrating the expressive power of the structured analytic model.
This trend also appears in Fig.~\ref{fig:taylor_error} (log-scale residuals vs.\ order). We note that
for \emph{adjacent} orders local non-monotonicity may occur, mainly due to ICP correspondence updates and
shape geometry, but the overall order–accuracy relation remains clear.

\subsubsection{Performance of 3D Analytic-ICP}

Next, we evaluate the \emph{stagewise} 3D Analytic-ICP. We study how accuracy improves with higher
effective orders (via composition of low-order stages) and briefly examine scalability by increasing
the registered points to $>10^5$, confirming both expressiveness and efficiency in 3D.

\begin{figure*}[htb]
\centering
\includegraphics[width=5.5in]{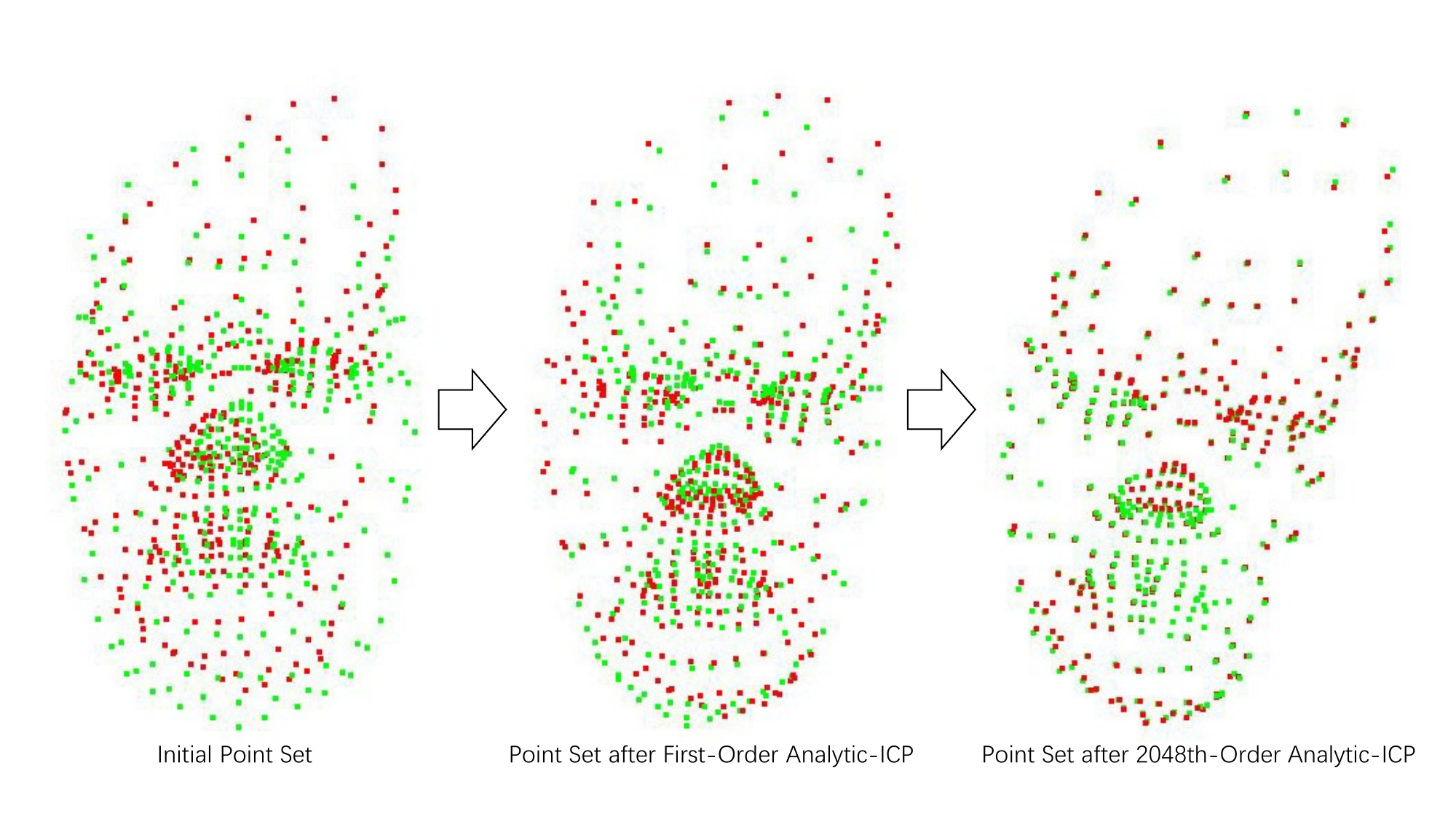}
\caption{\textbf{Performance of 3D stagewise Analytic-ICP.}
We synthesize a smooth non-affine deformation on a 3D facial point cloud (392 points) using our
structured analytic model: the first-order (affine) matrix has ones on the diagonal and all other
coefficients are sampled i.i.d.\ from $[-0.3,\,0.3]$. We then run second-order Analytic-ICP for
11 stages; composing $11$ quadratic stages yields an effective degree $\le 2^{11}=2048$.
Left: initial source (green) and target (red). Middle: affine result. Right: result after the
stagewise quadratic lifting (effective $\sim\!2048$th order), showing markedly reduced residuals.}
\label{fig:3Dface}
\end{figure*}

For this small 3D case, a single run of Analytic-ICP with quadratic stages ($o_s=2$) composed over
11 stages achieves an effective degree bounded by $2^{11}$, substantially improving upon the affine
baseline (middle of Fig.~\ref{fig:3Dface}). The clear residual reduction illustrates how composing
low-order structured Taylor maps quickly yields high expressive power while keeping each stage
numerically well conditioned.

\begin{figure*}[htb]
\centering
\scriptsize \textbf{By increasing the order of the 3D Taylor series, the accuracy of the fitting is improved.}

\vspace{1mm}

\begin{minipage}[c]{0.48\textwidth}
  \centering
  \includegraphics[width=0.45\linewidth]{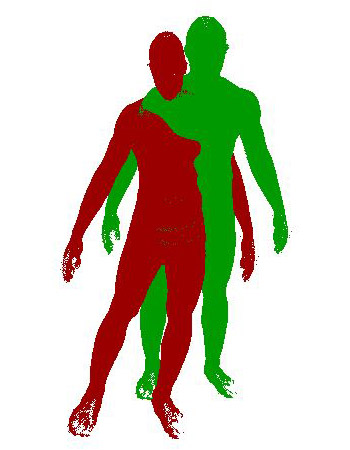}
  \hspace{2mm}
  \includegraphics[width=0.45\linewidth]{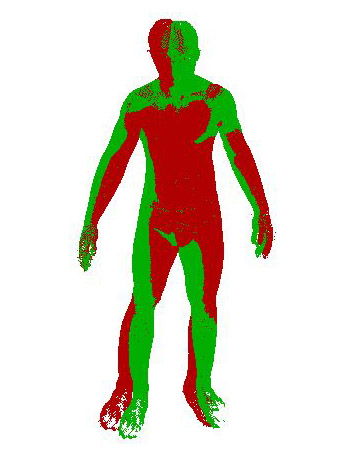}
\end{minipage}
\hfill
\begin{minipage}[c]{0.48\textwidth}
  \centering
  \includegraphics[width=0.45\linewidth]{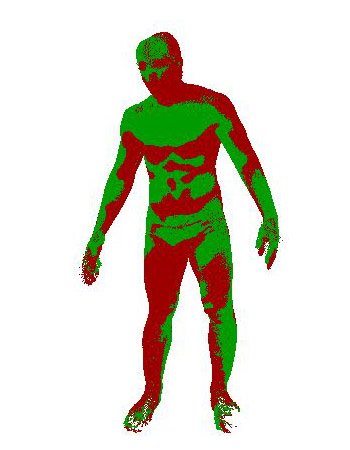}
  \hspace{2mm}
  \includegraphics[width=0.45\linewidth]{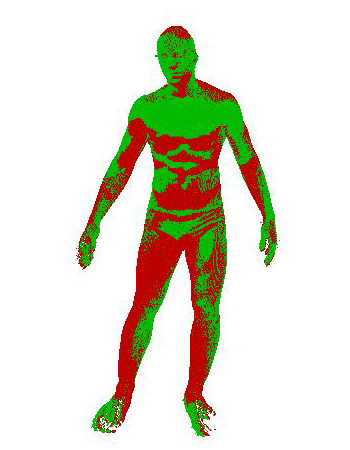}
\end{minipage}

\vspace{3mm}

\begin{minipage}[c]{0.48\textwidth}
  \centering
  \includegraphics[width=0.45\linewidth]{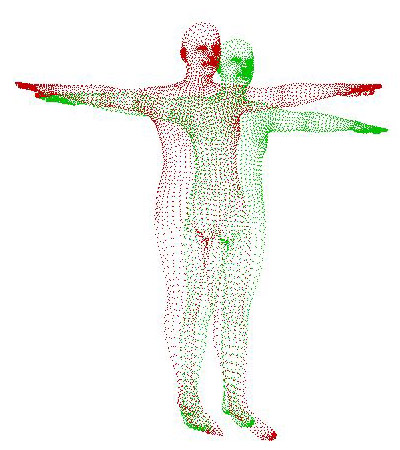}
  \hspace{2mm}
  \includegraphics[width=0.45\linewidth]{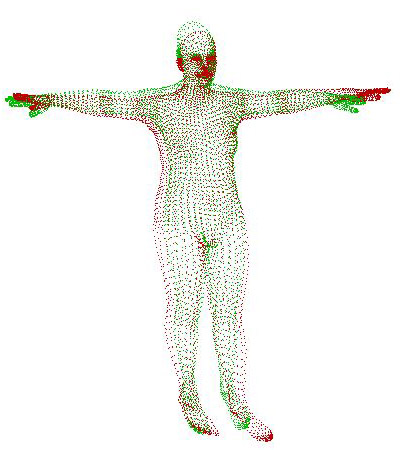} \\
  \vspace{1mm}
  original models \hspace{18mm} affine (first-order)
\end{minipage}
\hfill
\begin{minipage}[c]{0.48\textwidth}
  \centering
  \includegraphics[width=0.45\linewidth]{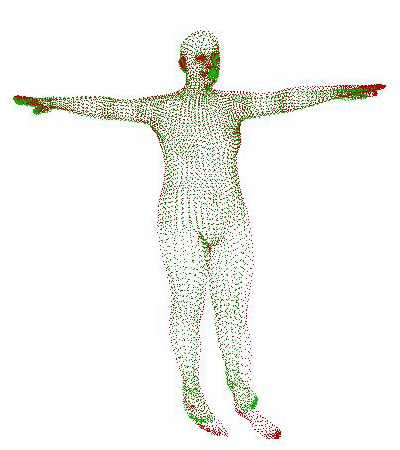}
  \hspace{2mm}
  \includegraphics[width=0.45\linewidth]{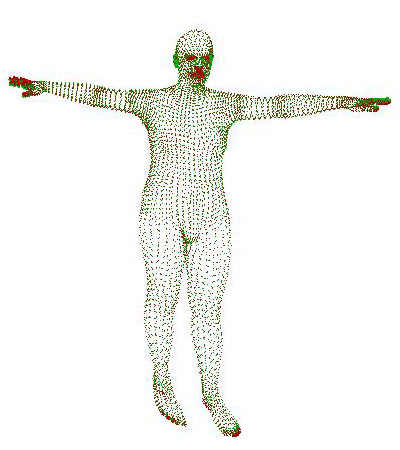} \\
  \vspace{1mm}
  256th-order Taylor \hspace{16mm} 32{,}768th-order Taylor
\end{minipage}

\caption{\textbf{3D Analytic-ICP registration performance under increasing Taylor series orders.}
Two different 3D human body point clouds (top and bottom rows) are registered using truncated Taylor expansions of increasing effective order. From affine (first-order) to very high-order Taylor approximations, residuals decrease markedly, demonstrating the expressive power of the analytic model.
\emph{Note: composing a second-order map $L$ times yields an effective degree $\le 2^{L}$ (e.g., $L{=}8\!\Rightarrow\!256$, $L{=}15\!\Rightarrow\!32{,}768$).}}
\label{fig:body4}
\end{figure*}

\begin{figure}[!ht]
\centering
\includegraphics[width=4in]{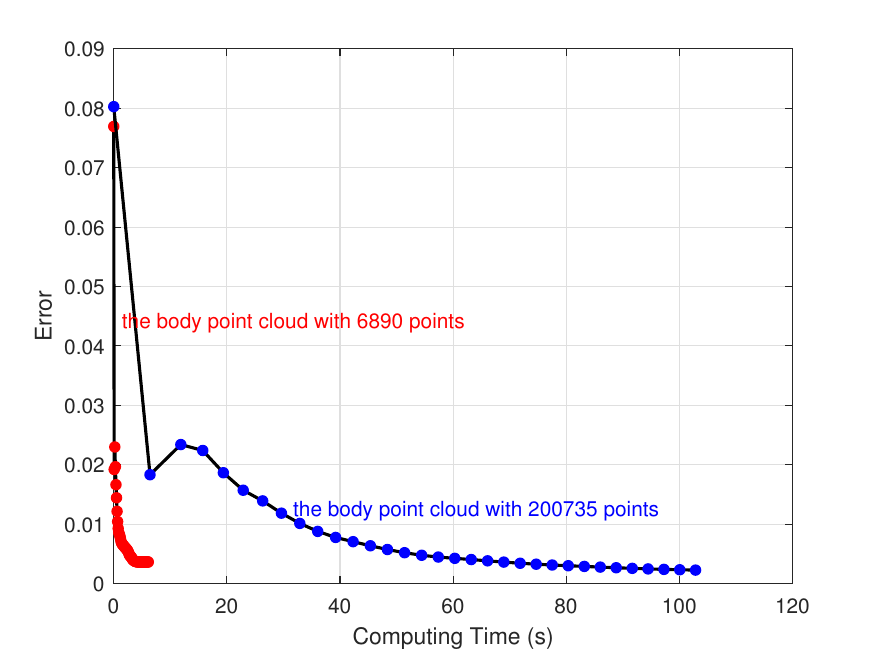}%
\hfil
\caption{ \textbf{Computation time vs. registration residuals for 3D Analytic-ICP.} 
    The plots show the trade-off between computation time and registration accuracy for two human body point clouds used in \figref{fig:body4}. Higher-order Taylor expansions lead to lower residuals at the cost of increased computation time.}
\label{fig:Compare3DPlot}
\end{figure}

Although our primary focus lies in the theoretical formulation and analytic expressiveness of the proposed registration model, we also conducted large-scale registration experiments to demonstrate its scalability and numerical robustness in practical settings.

\subsubsection{Performance on Large-Scale 3D Point Clouds}

To further demonstrate efficiency and scalability, we applied Analytic-ICP to two SHREC’19 models
with $6{,}890$ and $200{,}735$ points. Each model was perturbed by a structured analytic warp of
controlled complexity, generated by composing low-order (first/second) truncated Taylor maps. We
consider three complexity levels: (i) compositions of first-order maps, (ii) compositions of
second-order maps to an effective degree $\le 2^{8}$, and (iii) to an effective degree $\le
2^{15}$.

We then registered the original (fixed) model to each perturbed (moving) version using
Analytic-ICP. As the effective Taylor order increases (via stagewise order lifting), residuals
steadily decrease, indicating stable, expressive fits on large point sets; see
Fig.~\ref{fig:body4} for visual results across orders and Fig.~\ref{fig:Compare3DPlot} for a
representative qualitative summary.

On the $200{,}735$-point model, Analytic-ICP converged in \textbf{102.86\,s} (final RMSE:
\textbf{0.0022}); on the $6{,}890$-point model, it converged in \textbf{4.6\,s} (final RMSE:
\textbf{0.0036}). These results corroborate both the expressive power (via order lifting) and the
favorable scalability of our approach.

\begin{figure*}[htb]
\centering
\scriptsize \textbf{Heterogeneous registration: Analytic-ICP vs.\ CPD (2D)}

\vspace{2mm}

\begin{minipage}{0.30\linewidth}
  \centering
  \includegraphics[width=0.58\linewidth]{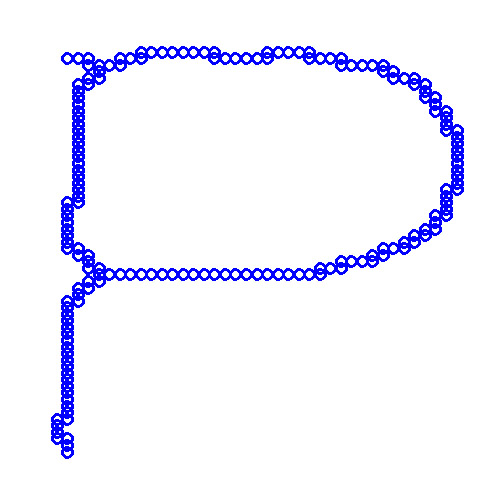}\\[0.8mm]
  \textbf{original (moving)}
\end{minipage}
\begin{minipage}{0.30\linewidth}
  \centering
  \includegraphics[width=0.58\linewidth]{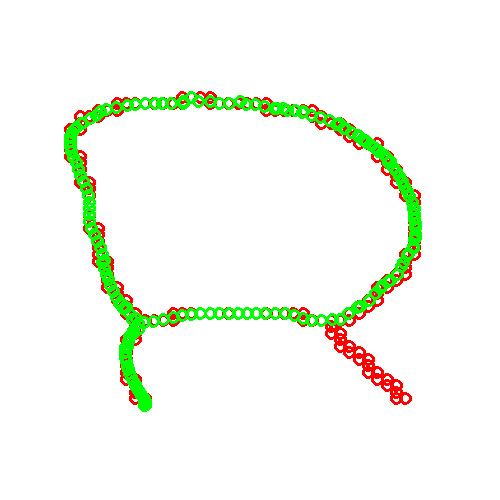}\\[0.8mm]
  \textbf{Analytic-ICP}\\
  residual: 0.0084,\quad time: 0.096\,s
\end{minipage}
\begin{minipage}{0.30\linewidth}
  \centering
  \includegraphics[width=0.58\linewidth]{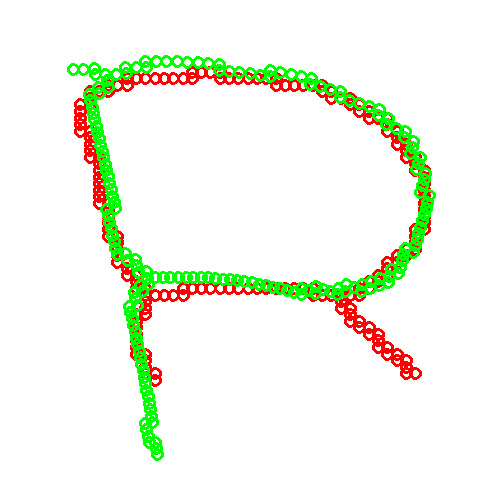}\\[0.8mm]
  \textbf{CPD}\\
  residual: 0.6345,\quad time: 1.331\,s
\end{minipage}

\vspace{3mm}

\begin{minipage}{0.30\linewidth}
  \centering
  \includegraphics[width=0.58\linewidth]{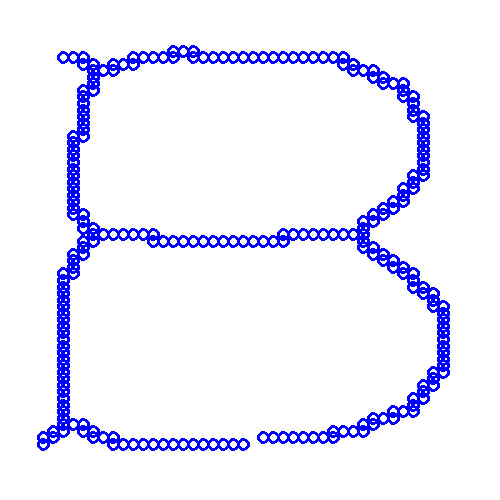}\\[0.8mm]
  \textbf{original (moving)}
\end{minipage}
\begin{minipage}{0.30\linewidth}
  \centering
  \includegraphics[width=0.58\linewidth]{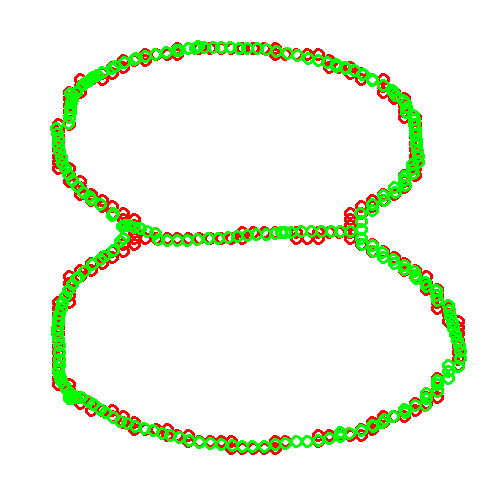}\\[0.8mm]
  \textbf{Analytic-ICP}\\
  residual: 0.0119,\quad time: 0.129\,s
\end{minipage}
\begin{minipage}{0.30\linewidth}
  \centering
  \includegraphics[width=0.58\linewidth]{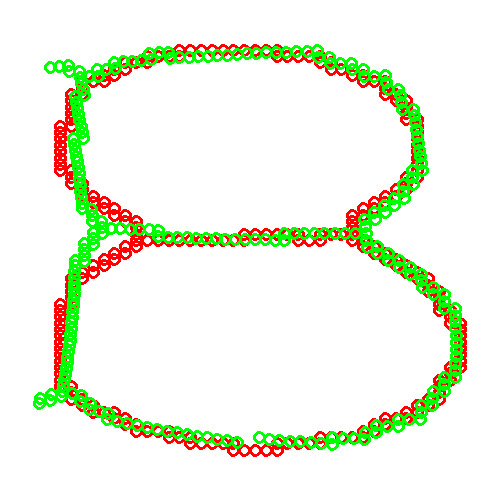}\\[0.8mm]
  \textbf{CPD}\\
  residual: 0.8420,\quad time: 2.507\,s
\end{minipage}

\vspace{2mm}
\caption{\textbf{Heterogeneous 2D registration under identical settings.}
Red points: fixed set; green: transformed moving set. Analytic-ICP uses staged fitting
(rigid $\rightarrow$ affine $\rightarrow$ \emph{low-order Taylor lifting across iterations}),
whereas CPD performs probabilistic soft matching. On both pairs, Analytic-ICP attains
lower residuals with markedly shorter runtime. Wall-clock times are comparable \emph{within
this figure} (same single-threaded build and hardware).}
\label{fig:heteroCompare}
\end{figure*}

\subsection{Comparative Experiments}
\label{sec:bow_results}
\paragraph{Scope and rationale}
Our comparative experiments primarily target \textbf{small-scale point sets}, consistent with the paper’s goal of developing a robust, principled \emph{foundational registration algorithm}. Focusing on small instances isolates algorithmic behavior from confounders such as GPU acceleration, batch parallelism, and platform-dependent throughput, while remaining practically relevant for industrial/embedded/edge scenarios~\cite{vizzo2023kiss}. 

\paragraph{Baselines}
We compare against two widely recognized, theoretically compact methods: \textbf{TPS–RPM} and \textbf{CPD}. TPS–RPM models deformation with thin-plate spline regularization (a solid-mechanics view), whereas CPD enforces motion coherence via a probabilistic Gaussian mixture formulation. Our \textbf{Analytic-ICP} instead fits a structured multivariate analytic map. From a continuum perspective, these span complementary priors (solid-like, flow-like, and analytic), providing a balanced benchmark suite.\footnote{BCPD~\cite{BCPD} extends CPD with Bayesian priors but retains the $O(N^{2})$ correspondence structure and adds overhead from prior inference. Because our focus is on the underlying analytic mapping rather than probabilistic modeling, we do not include BCPD in the main comparisons.}

\paragraph{Evaluation protocol}
All methods operate on zero-centered, unit-scale data and use identical stopping rules where applicable. Unless stated otherwise, CPD uses $\lambda = 3.0,\ \beta = 3.0,\ \omega = 0.1$; TPS–RPM uses $\lambda\!=\!1$, annealing temperature $T\!=\!1$, schedule ratio $r\!=\!0.9$. Analytic-ICP follows the staged pipeline (rigid $\rightarrow$ affine $\rightarrow$ structured Taylor) with residual–gated order lifting and shared ICP correspondence updates. Timing within each subsection is measured on the \emph{same} machine and single-threaded build to ensure fairness.

\subsubsection{Comparison on 2D Non-rigid Registration}
\label{sec:experiment-2d}
\begin{figure*}[ht]
    \centering
    \scriptsize \textbf{Comparison of registration performance for small 2D deformations using Analytic-ICP and CPD.}
    \vspace{2mm}

    \begin{minipage}{0.3\linewidth}
        \centering
        \includegraphics[width=\linewidth]{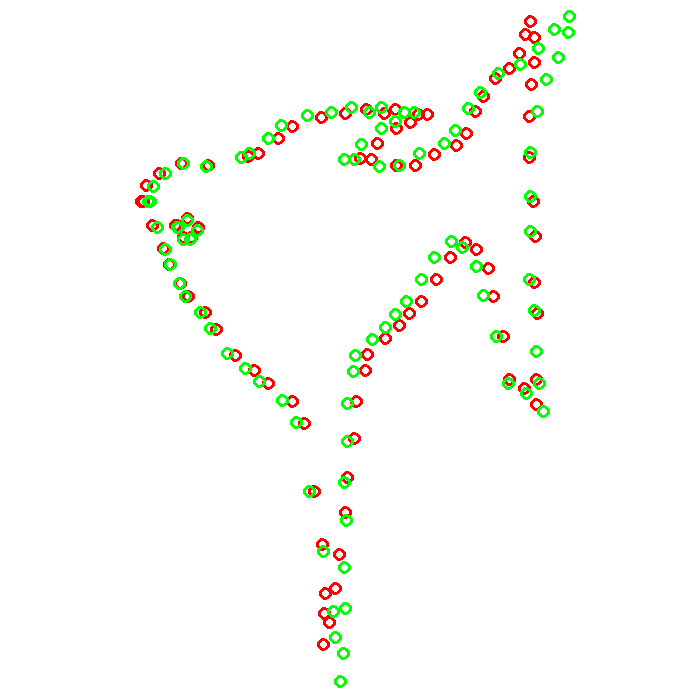}
    \end{minipage}
    \begin{minipage}{0.3\linewidth}
        \centering
        \includegraphics[width=\linewidth]{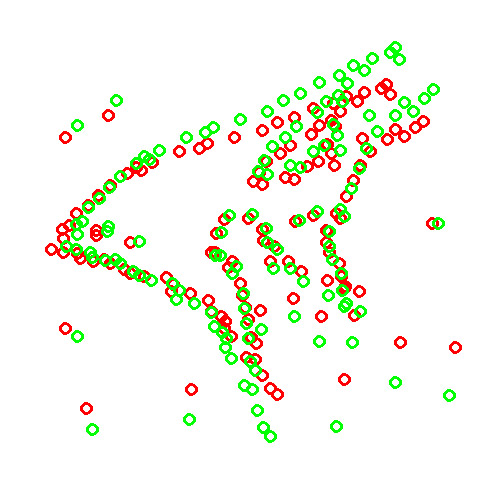}
    \end{minipage}
    \begin{minipage}{0.3\linewidth}
        \centering
        \includegraphics[width=\linewidth]{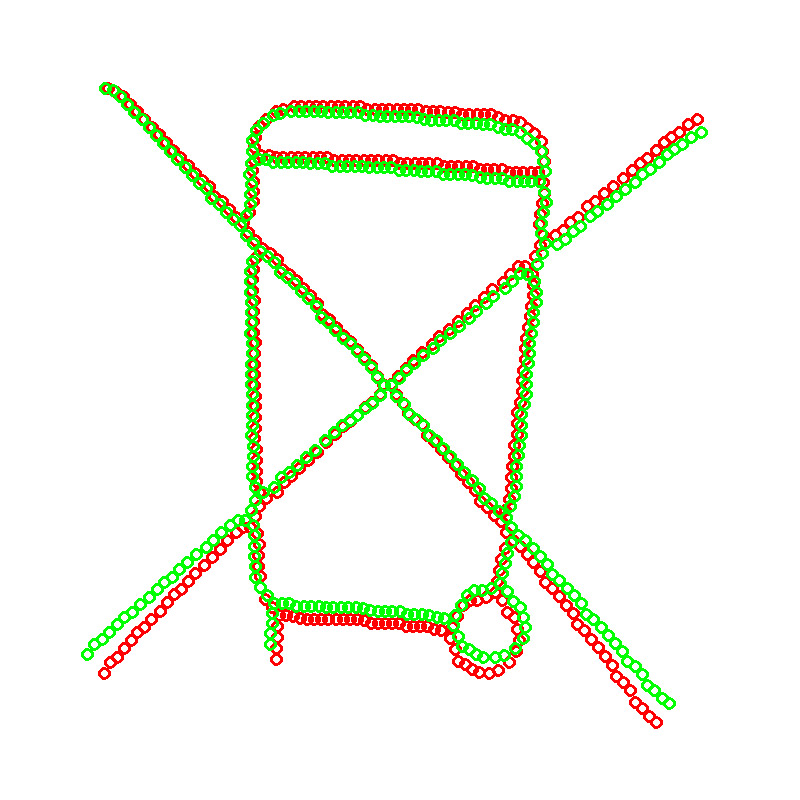}
    \end{minipage}

    \vspace{2mm}

    \begin{minipage}{0.3\linewidth}
        \centering
        \includegraphics[width=\linewidth]{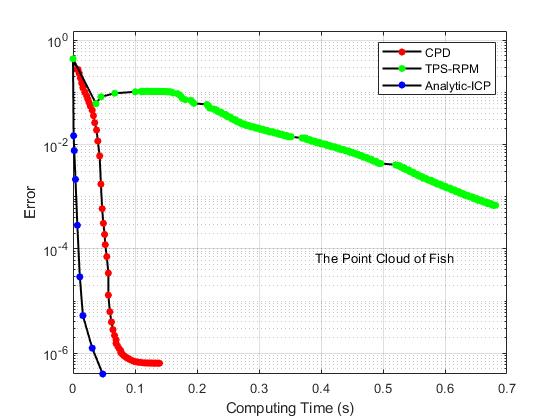}
        \vspace{1mm}
        
        \textbf{\scriptsize registration of a fish \\for small deformation}
    \end{minipage}
    \begin{minipage}{0.3\linewidth}
        \centering
        \includegraphics[width=\linewidth]{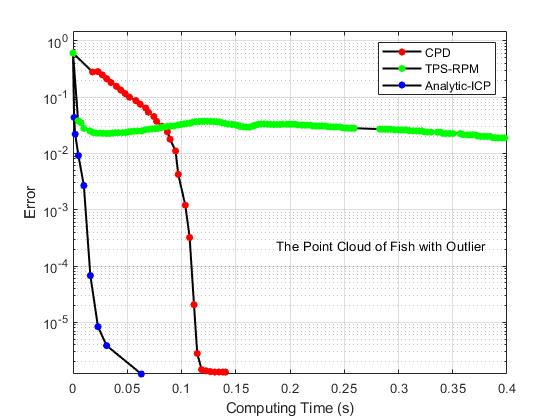}
        \vspace{1mm}
        
        \textbf{\scriptsize registration of a fish with noise \\for small deformation}
    \end{minipage}
    \begin{minipage}{0.3\linewidth}
        \centering
        \includegraphics[width=\linewidth]{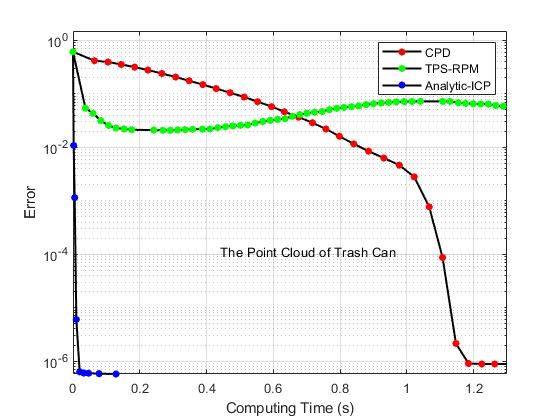}
        \vspace{1mm}
        
        \textbf{\scriptsize registration of a trash can \\for small deformation}
    \end{minipage}

   \caption{
\textbf{Error vs. Computation Time for Analytic-ICP, CPD, and TPS-RPM on 2D Registration Tasks.} 
The top row displays the registration targets: three representative 2D point sets (green) and their perturbed versions (red). 
The bottom row plots the residual error against computation time for the three algorithms. 
Analytic-ICP demonstrates superior accuracy and speed in handling small, smooth deformations.
}
    \label{fig:isoContrast2d}
\end{figure*}

\paragraph{Heterogeneous pairs (Analytic-ICP vs.\ CPD)}
We evaluated \textbf{two heterogeneous pairs} of point sets and compared \textbf{Analytic-ICP} with \textbf{CPD} under identical preprocessing, stopping rules, and \emph{the same} single-threaded hardware build. Analytic-ICP followed the staged pipeline (rigid $\rightarrow$ affine $\rightarrow$ structured Taylor) with \emph{order lifting across outer iterations}. In this run we set an order cap $m_w=8$ and used $k_{\mathrm{step}}{=}1$, executing eight outer iterations; the resulting composed map is effectively high-order (factorial-like growth when each stage contributes a quadratic lift).\footnote{With quadratic per-stage lifts, the effective degree is upper-bounded by $2\cdot3\cdot\cdots\cdot(1+\text{\#lifts})$, e.g., $\le 8!$ after eight lifts; we report this qualitatively to avoid over-specifying run-dependent details.}

As shown in \figref{fig:heteroCompare}, \textbf{Analytic-ICP} yields \emph{lower residuals} and \emph{substantially shorter runtimes} than \textbf{CPD} on these heterogeneous pairs. We attribute this to two factors: (i) CPD’s dense $N{\times}N$ correspondence matrix (E-step) is both computationally heavy and sensitive to ambiguous matches that arise in heterogeneous regions, and (ii) our structured analytic model (AMVFF) provides a strong, low-parameter prior that, when combined with staged fitting, produces sharper, better-conditioned updates.

\emph{Scope.} The principal limitation in this experiment is the \emph{nearest-neighbor correspondence} used in the ICP outer loop: the method is most reliable when the initial pose places a majority of pairs within the ICP basin. This limitation stems from the correspondence mechanism rather than from the analytic model itself; the AMVFF representation remains expressive and is compatible with soft/joint correspondence strategies if desired.

\begin{figure*}[!t]
\centering
\includegraphics[width=5.5in]{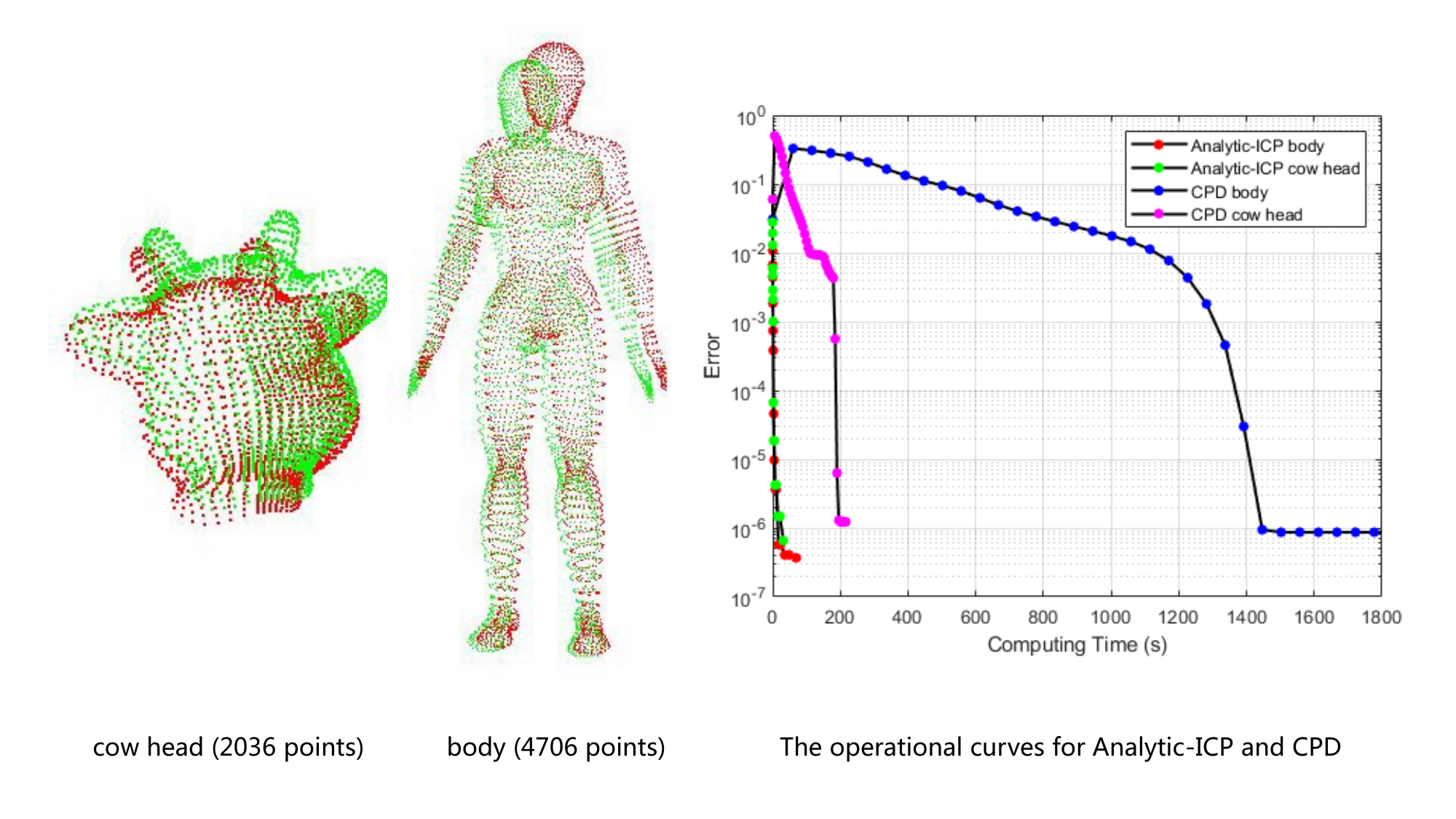}%
\hfil
\caption{%
\textbf{Comparison of the registration performance of Analytic-ICP and CPD on distorted 3D point cloud models.} %
The original point cloud is shown in red, and the perturbed version is shown in green. Each algorithm attempts to align the original point cloud to its distorted counterpart. The figure illustrates the residuals across iterations, highlighting the differences in convergence speed and final accuracy between the two methods.
}
\label{fig:isoCompare3d}
\end{figure*}

\paragraph{Isomorphic pairs of point sets}
{
Analytic-ICP is not designed for very large deformations—where CPD and similar probabilistic matchers often prevail—but it shows clear strengths under \emph{smooth, small} deformations: high accuracy at much lower runtime. This makes it well suited as a reliable local optimizer (fine registration) in practical pipelines.

In this experiment we use the staged pipeline in Algorithm~\ref{alg:analyticcode} with \emph{order lifting across iterations}: each outer iteration fits a low-order (e.g., quadratic) structured Taylor \emph{factor}, and the current map is updated by composing these factors. Thus the effective model order increases progressively without ever fitting a single high-degree polynomial in one shot.

We evaluated three small 2D sets—a fish (91 points), a noisy fish (127 points), and a trash can (359 points). Each target was generated by applying a \emph{small} truncated-Taylor perturbation to the source (first row of \figref{fig:isoContrast2d}; green: original, red: perturbed). We then registered the originals to their perturbed counterparts using Analytic-ICP, CPD, and TPS-RPM.

As shown in the second row of \figref{fig:isoContrast2d} and summarized in \tabref{tab:computingtime2d} and \tabref{tab:residual2d}, Analytic-ICP consistently attains lower residuals with substantially shorter runtime. The advantage grows with point count: for the fish set Analytic-ICP is about \( \sim\!3\times \) faster than CPD, and for the trash-can set (359 points) the speedup exceeds an order of magnitude, while maintaining better or comparable accuracy.

}
 
\begin{table}[ht]
\centering
\renewcommand{\arraystretch}{1.3}
\setlength{\tabcolsep}{10pt}
\caption{\textbf{Registration Residuals on 2D Point Sets.} Comparison of residual errors for CPD, TPS-RPM, and Analytic-ICP on three 2D datasets with small deformations. Lower is better.}
\label{tab:residual2d}
\begin{tabular}{lccc}
\toprule
\rowcolor{maroon!50}
\textbf{Algorithm} & \textbf{Fish} & \textbf{Fish+Noise} & \textbf{Trash Can} \\
\midrule
\rowcolor{maroon!10}
CPD & 0.139 & 0.141 & 1.301 \\
\rowcolor{maroon!10}
TPS-RPM & 0.681 & 0.953 & 7.442 \\
\rowcolor{maroon!10}
\textbf{Analytic-ICP} & \textbf{\textcolor{petr}{0.048}} & \textbf{\textcolor{petr}{0.063}} & \textbf{\textcolor{petr}{0.129}} \\
\bottomrule
\end{tabular}
\end{table}

\begin{table}[ht]
\centering
\renewcommand{\arraystretch}{1.3}
\setlength{\tabcolsep}{10pt}
\caption{\textbf{Registration Time on 2D Point Sets (in seconds).} Comparison of execution time for CPD, TPS-RPM, and Analytic-ICP on three 2D datasets. Lower is better.}
\label{tab:computingtime2d}
\begin{tabular}{lccc}
\toprule
\rowcolor{maroon!50}
\textbf{Algorithm} & \textbf{Fish} & \textbf{Fish+Noise} & \textbf{Trash Can} \\
\midrule
\rowcolor{maroon!10}
CPD & 0.00000064 & 0.0000013 & 0.00000088 \\
\rowcolor{maroon!10}
TPS-RPM & 0.000686 & 0.000387 & 0.000576 \\
\rowcolor{maroon!10}
\textbf{Analytic-ICP} & \textbf{\textcolor{petr}{0.00000040}} & \textbf{\textcolor{petr}{0.0000012}} & \textbf{\textcolor{petr}{0.00000057}} \\
\bottomrule
\end{tabular}
\end{table}

\paragraph{Statistical Analysis: Deformation vs.\ Residual}
{
To examine robustness across deformation levels, we studied how the final residual varies with the \emph{initial} post-rigid RMSD. Starting from the normalized fish point set, we generated random perturbations using a fixed fifth-order truncated Taylor map. After rigid pre-alignment, the resulting RMSD served as a proxy for deformation magnitude. Representative perturbations at the six RMSD levels are visualized in Fig.~\ref{fig:deformandresidual}.

We binned the initial RMSD into six narrow intervals (each with $\approx 10$ retained samples). For every sample, Analytic-ICP, CPD, and TPS-RPM were run to register the original to the perturbed shape; per-bin residuals were then aggregated to produce Fig.~\ref{fig:errorandresidual}. As expected, when the initial error is small (within the ICP convergence basin), Analytic-ICP attains the lowest residuals. As the initial error grows and correspondences become ambiguous, CPD/TPS-RPM can overtake in residual accuracy.
}

\begin{figure}[!ht]
\centering
\includegraphics[width=4in]{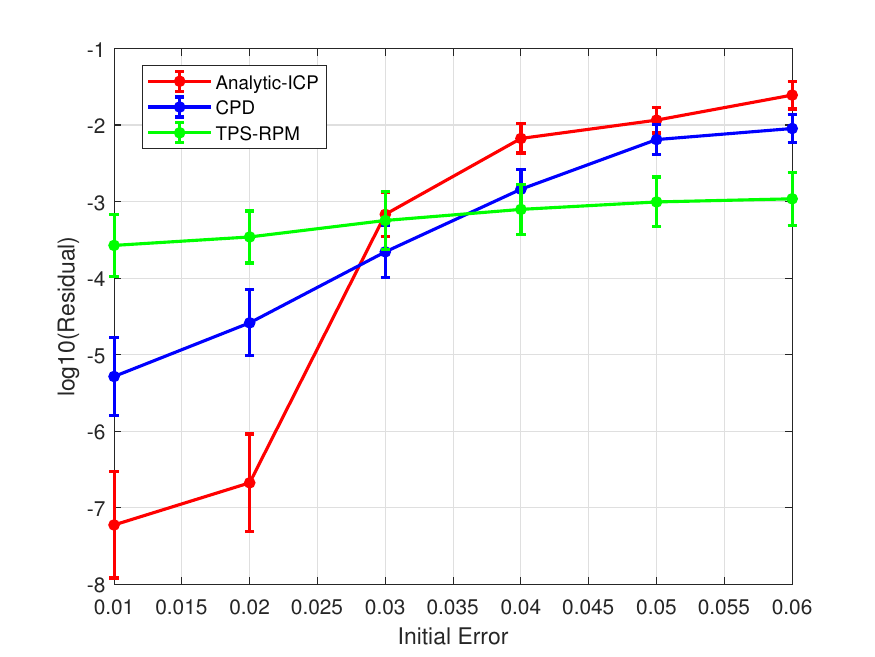}
\caption{\textbf{Residuals vs.\ initial RMSD (log-scale $y$).}
Each curve aggregates results over RMSD bins ($\approx\!10$ retained samples/bin). Minor discrepancies in bin widths/normalization and sampling variability may slightly shift individual points, but the overall trends and method ranking remain consistent.}
\label{fig:errorandresidual}
\end{figure}

\begin{figure*}[ht]
    \centering
    \scriptsize \textbf{The magnitude of model deformation due to random perturbation}
    \vspace{2mm}

    \begin{minipage}{0.15\linewidth}
        \centering
        \includegraphics[width=\linewidth]{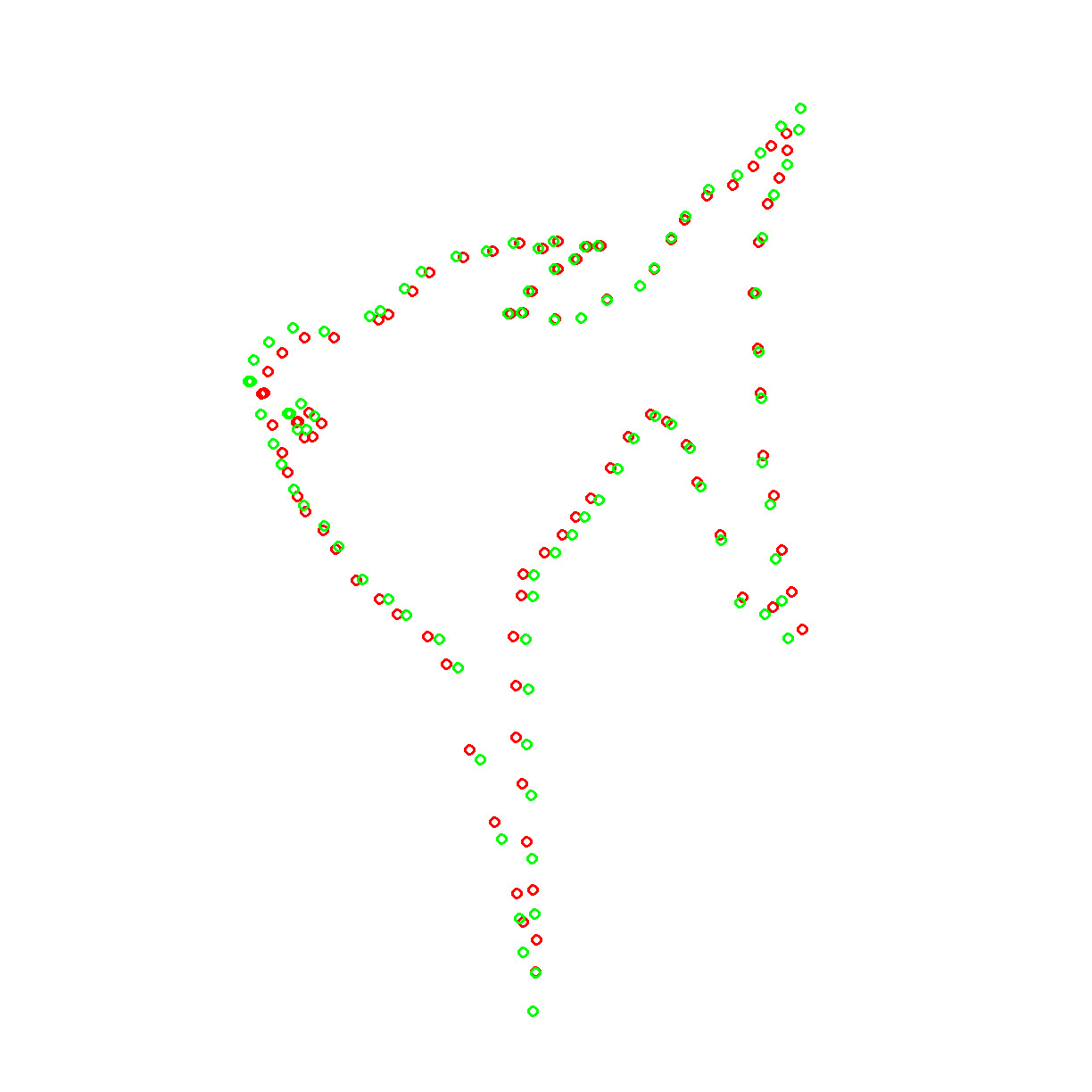}
        \vspace{1mm}
        
        RMSD $\approx$ 0.01
    \end{minipage}
    \hspace{1mm}
    \begin{minipage}{0.15\linewidth}
        \centering
        \includegraphics[width=\linewidth]{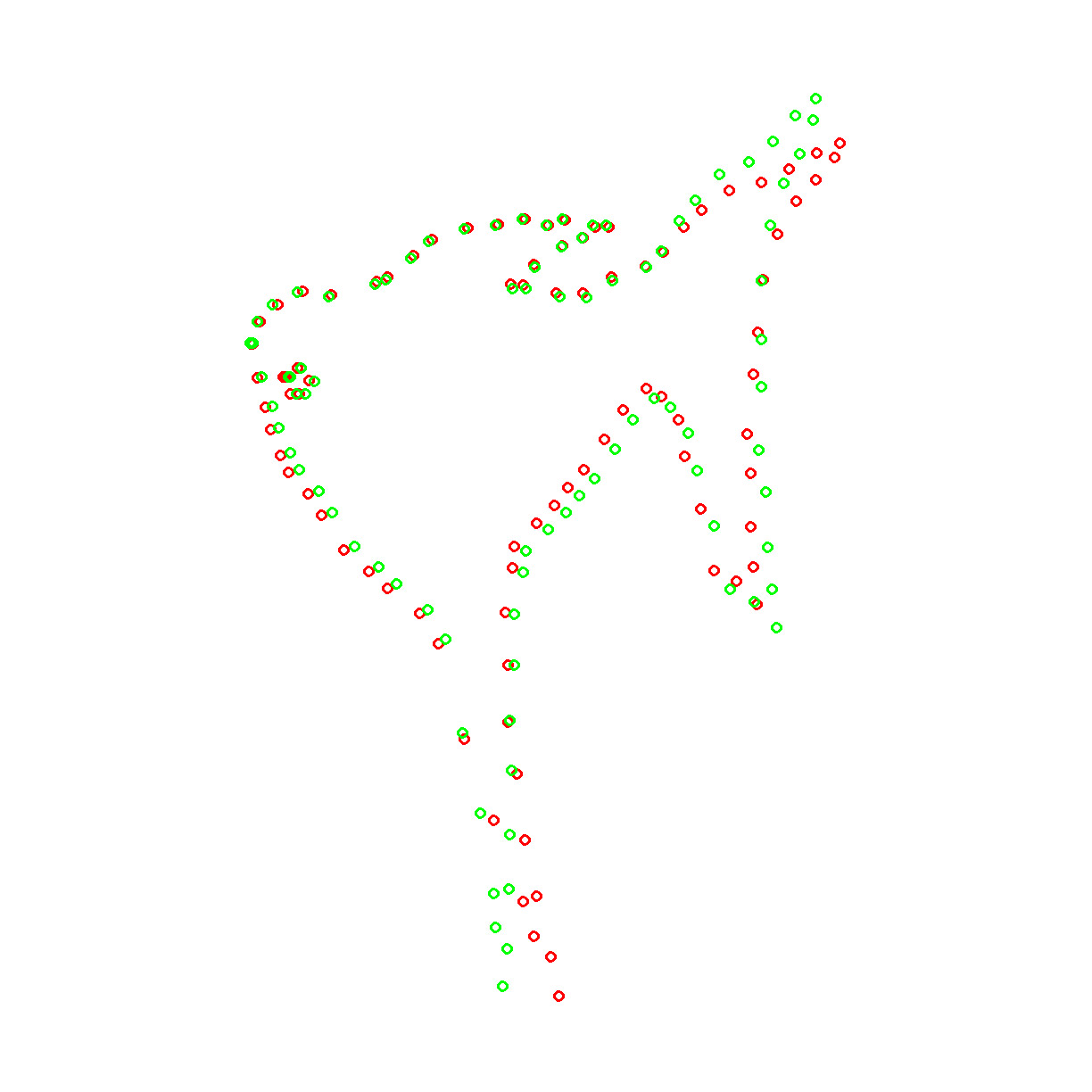}
        \vspace{1mm}
        
        RMSD $\approx$ 0.02
    \end{minipage}
    \hspace{1mm}
    \begin{minipage}{0.15\linewidth}
        \centering
        \includegraphics[width=\linewidth]{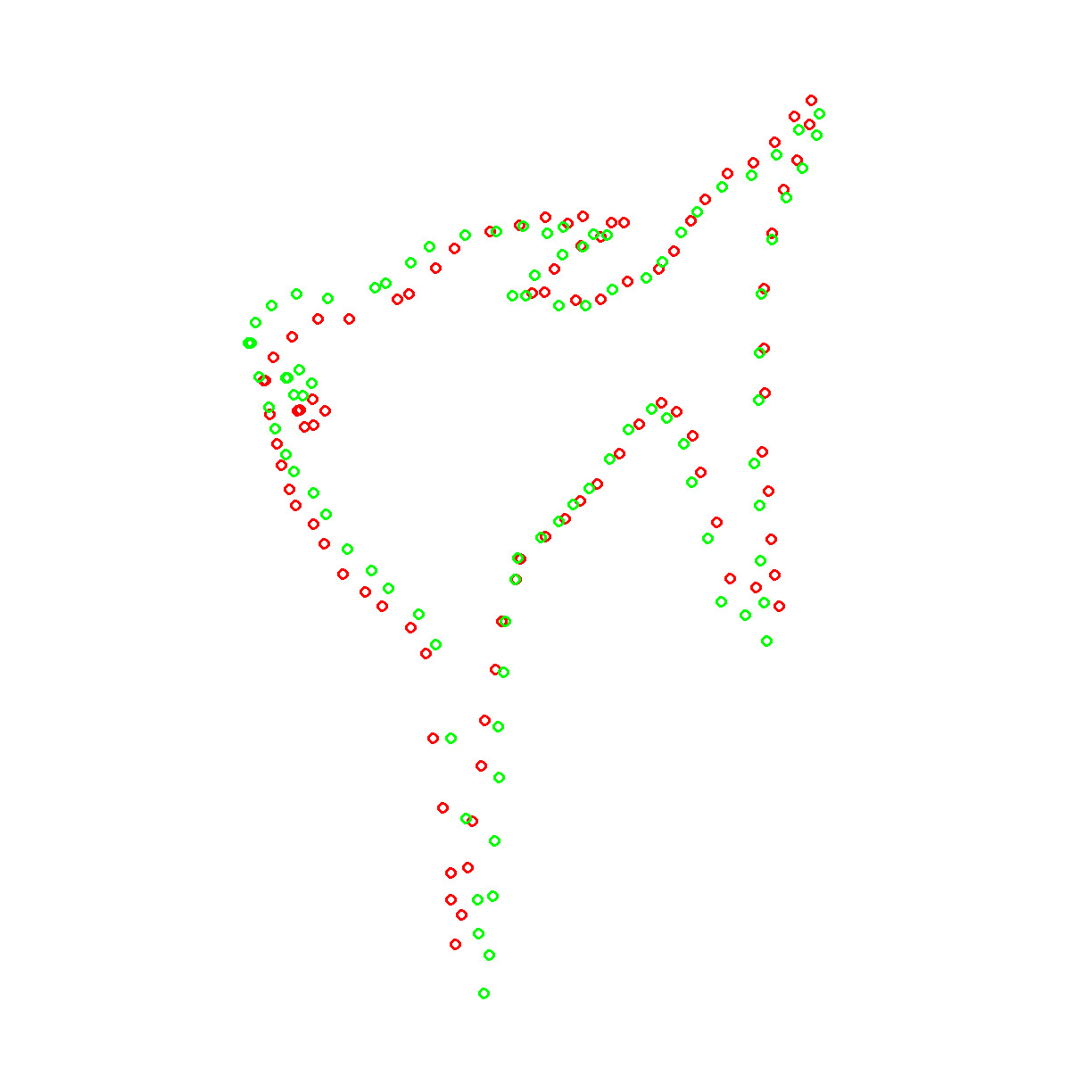}
        \vspace{1mm}
        
        RMSD $\approx$ 0.03
    \end{minipage}
    \hspace{1mm}
    \begin{minipage}{0.15\linewidth}
        \centering
        \includegraphics[width=\linewidth]{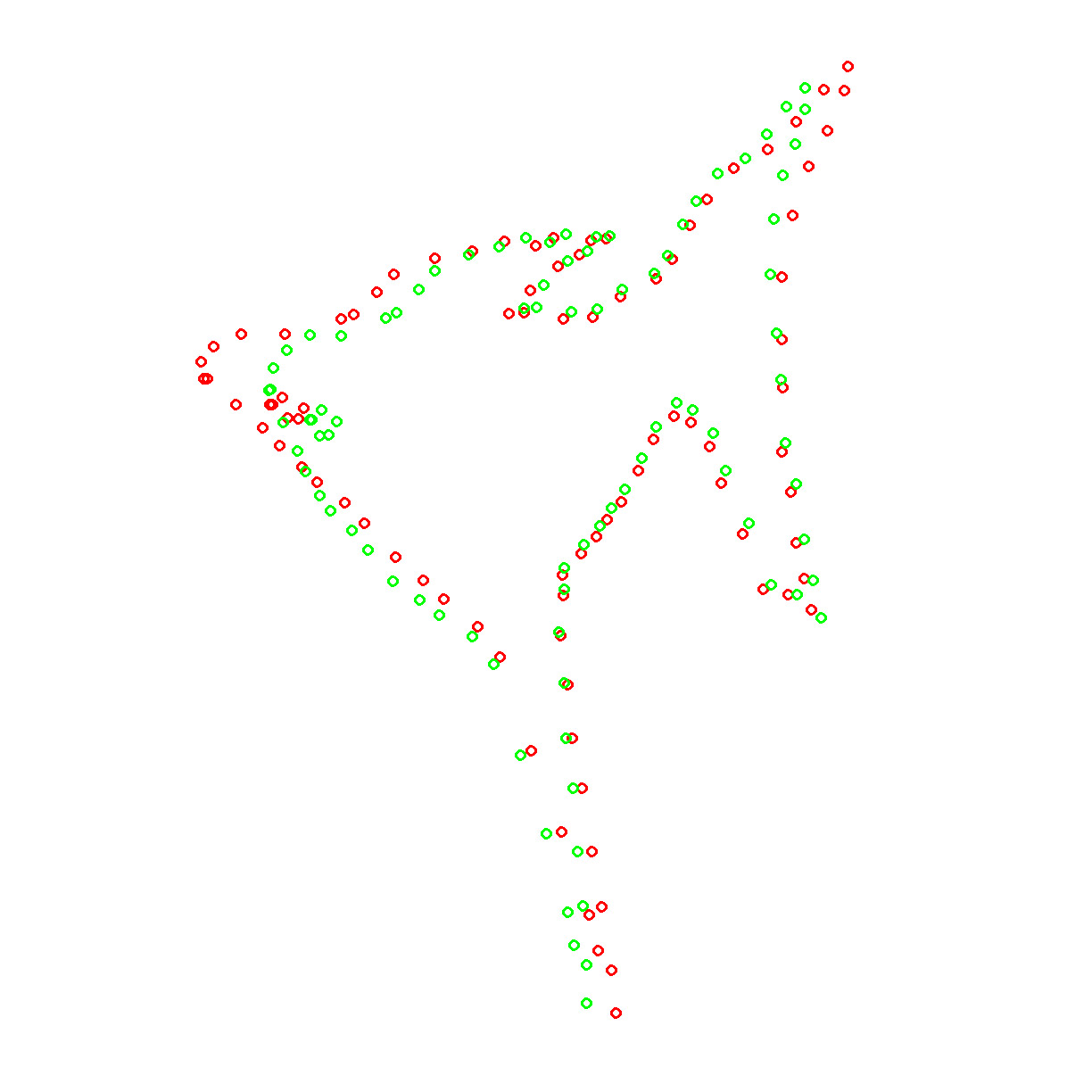}
        \vspace{1mm}
        
        RMSD $\approx$ 0.04
    \end{minipage}
    \hspace{1mm}
    \begin{minipage}{0.15\linewidth}
        \centering
        \includegraphics[width=\linewidth]{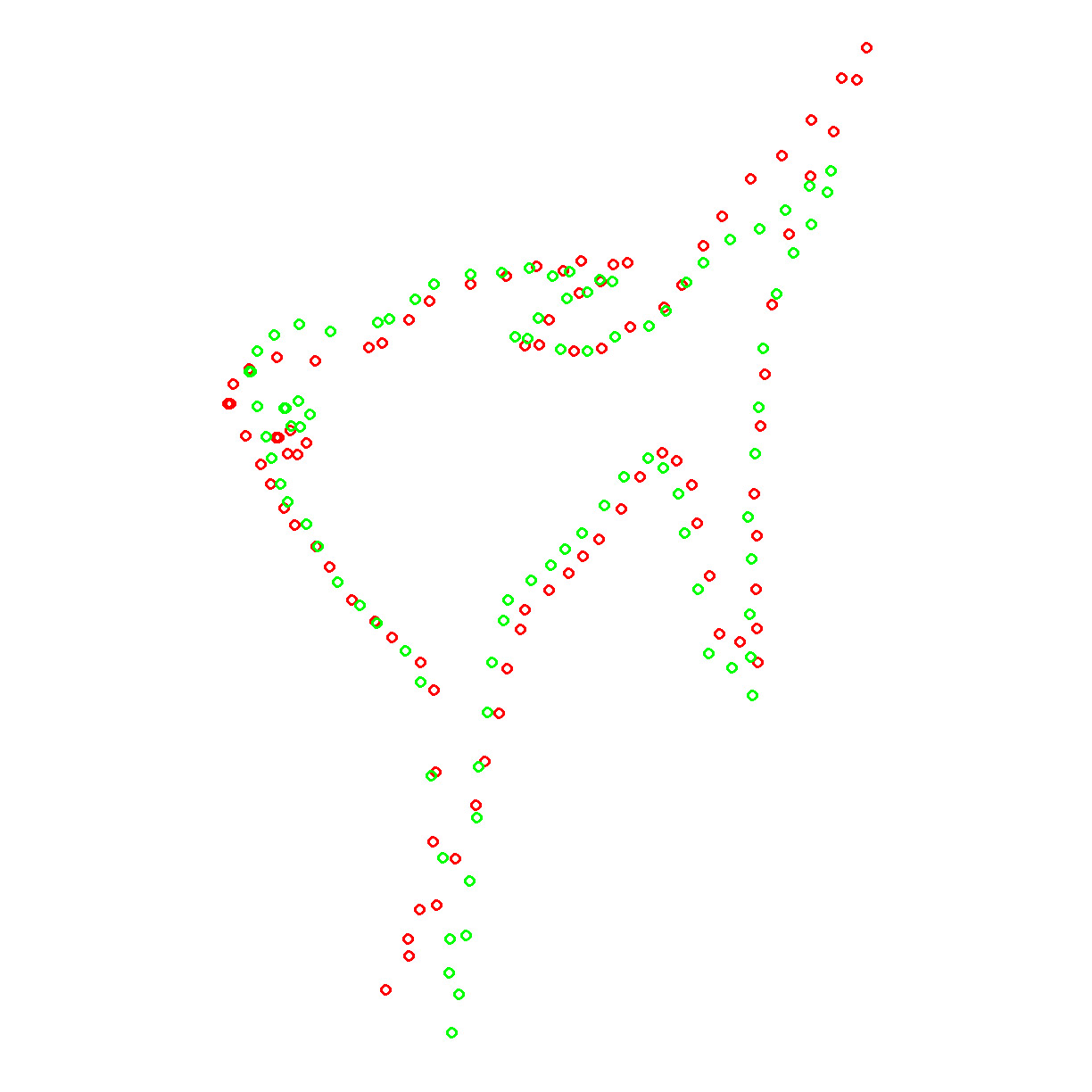}
        \vspace{1mm}
        
        RMSD $\approx$ 0.05
    \end{minipage}
    \hspace{1mm}
    \begin{minipage}{0.15\linewidth}
        \centering
        \includegraphics[width=\linewidth]{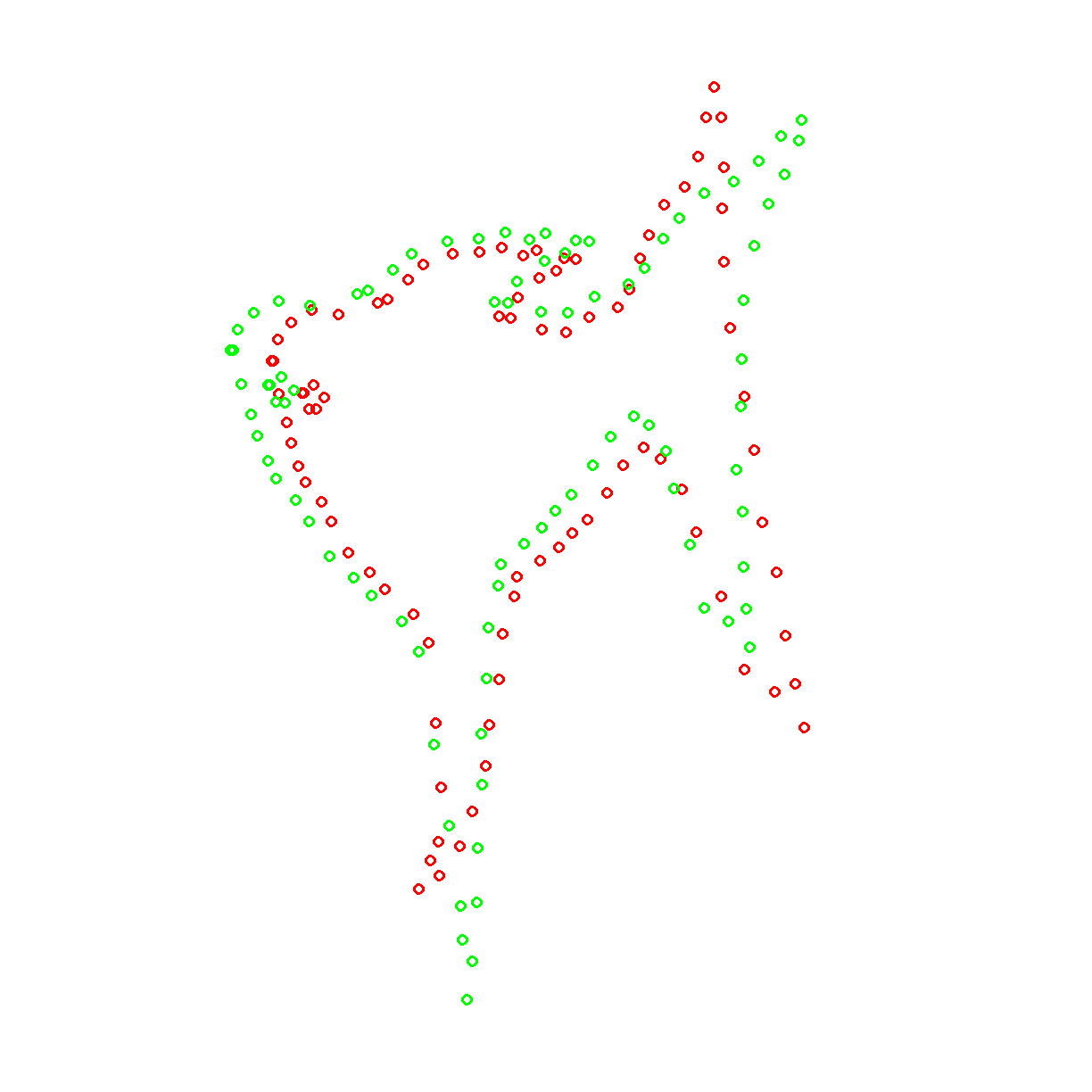}
        \vspace{1mm}
        
        RMSD $\approx$ 0.06
    \end{minipage}

    \caption{\textbf{Visualization of increasing deformation magnitude in 2D point sets.} All deformations are generated using randomly truncated Taylor expansions, with increasing root mean square distances (RMSD) from left to right.}
    \label{fig:deformandresidual}
\end{figure*}

\begin{table}[ht]
\centering
\renewcommand{\arraystretch}{1.4}
\setlength{\tabcolsep}{10pt}
\caption{\textbf{Registration performance on large 3D point clouds.} Both residuals and computation time are reported for CPD and Analytic-ICP on two datasets.}
\label{tab:ComputingTime3d}
\begin{tabular}{lcc}
\toprule
\rowcolor{maroon!50}
\textbf{Algorithm} & \textbf{Cow head (2036 pts)} & \textbf{Body (4706 pts)} \\
\midrule
\rowcolor{maroon!10}
CPD & Time: 215.86s, Residual: 1.23e-6 & Time: 2108.01s, Residual: 8.7e-7 \\
\rowcolor{maroon!10}
\textbf{Analytic-ICP} & \textcolor{petr}{Time: 31.86s, Residual: 6.6e-7} & \textcolor{petr}{Time: 69.57s, Residual: 3.7e-7} \\
\bottomrule
\end{tabular}
\end{table}

\subsubsection{Comparison on 3D Non-rigid Registration}

We evaluate 3D non-rigid registration on two small-scale point clouds: a cow head (2{,}036 points; course assignment data) and a mannequin (4{,}706 points; SHREC’19). Following the 2D protocol, we use the same Analytic-ICP settings as in Sec.~\ref{sec:experiment-2d} and generate \emph{small, smooth} deformations by a third-order truncated Taylor map: the first-order diagonal is fixed to $1$, while other coefficients are sampled i.i.d.\ from $[-0.2,\,0.2]$. Each original cloud is registered to its perturbed counterpart by \textbf{Analytic-ICP} and \textbf{CPD}.

As shown in Fig.~\ref{fig:isoCompare3d} and Table~\ref{tab:ComputingTime3d}, under small deformations Analytic-ICP achieves lower residuals and shorter runtimes than CPD on both models. The same settings used in 2D transfer directly to 3D, indicating good robustness across dimensions. Moreover, the gap widens with the number of points, consistent with the lower time/space complexity of Analytic-ICP.

\begin{figure}[!ht]
\centering
\includegraphics[width=0.9\linewidth]{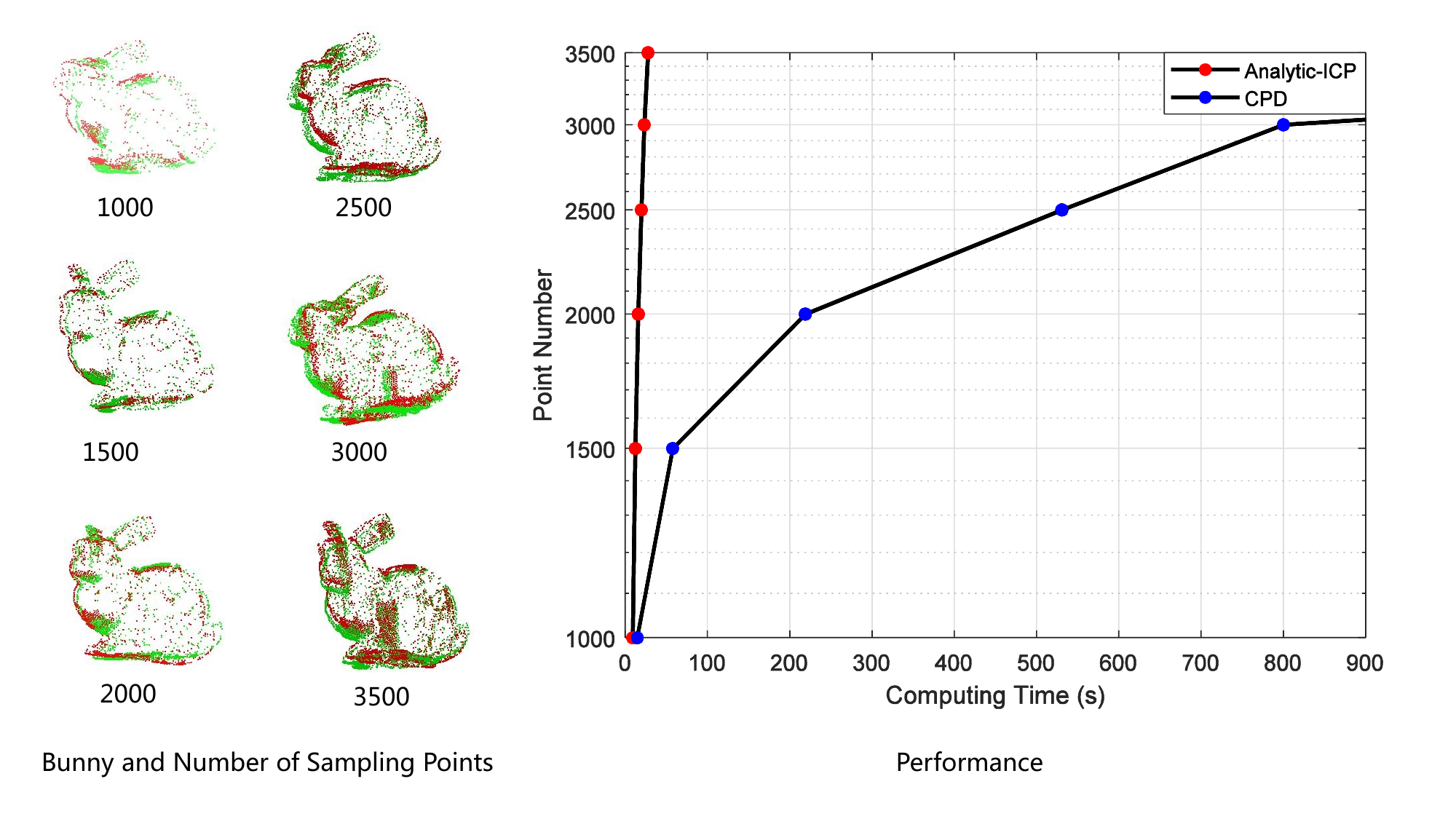}
\caption{\textbf{Runtime vs.\ point count (3D).}
Analytic-ICP scales favorably as the point count increases, whereas CPD exhibits a much steeper growth, highlighting Analytic-ICP's suitability for larger 3D sets.}
\label{fig:numberandtime}
\end{figure}

We further study scaling on the Stanford Bunny by uniform down/upsampling to $1{,}000$, $1{,}500$, $2{,}000$, $2{,}500$, $3{,}000$, and $3{,}500$ points. For each size, we synthesize deformations (truncated Taylor) while keeping the post-rigid RMSD at $\approx 0.05$ to normalize difficulty. On identical inputs, we compare Analytic-ICP and CPD and report runtimes when Analytic-ICP reaches residuals no worse than CPD, ensuring a fair timing comparison at matched (or better) accuracy. Results in Fig.~\ref{fig:numberandtime} show clearly improved scaling for Analytic-ICP, whereas CPD slows markedly with size, supporting the practical advantage of Analytic-ICP for larger 3D registration tasks.

\subsubsection{Registration Under Heterogeneous \texorpdfstring{$C^\infty$}{C-infinity} Non-Analytic Deformations}

To stress-test robustness beyond analytic maps, we construct a \textbf{globally smooth but non-analytic}
deformation using compactly supported $C^\infty$ bump weights—i.e., cases not representable by a single
Taylor polynomial. Experiments in this subsection were run on a different workstation
(\emph{AMD Ryzen 7 5800H @ 3.20\,GHz, 32\,GB RAM}) with a 64-bit single-threaded C++ build.
Unless otherwise noted, comparison settings (Analytic-ICP vs.\ CPD) follow the earlier sections; wall-clock
times are therefore comparable \emph{within this subsection} but not across hardware.

For Analytic-ICP, we employ a Taylor-order lifting controller with $m_w=10$, enabling single-step
10th-order fits, which we found to be numerically stable and accurate; the outer iteration count
remains modest and is not critical for this experiment.

\paragraph{Non-analytic deformation model}
We use a two-center, smoothly blended quadratic field:
\begin{itemize}
  \item Two second-order coefficient blocks \(Q_1,Q_2\in\mathbb{R}^{3\times 6}\) and two centers \(c_1,c_2\in\mathbb{R}^3\).
  \item The baseline Taylor coefficients follow our standard initialization: for each order \(k=0,\ldots,\deg\), the entries are drawn i.i.d.\ from \(\mathcal{U}[-0.07,0.07]\); the first-order diagonal is fixed to identity (\(A=\mathrm{diag}(1,1,1)\)), and the remaining first-order off-diagonals remain in \([-0.07,0.07]\). The translation \(t\) is small (from the zeroth-order block).
  \item Let \(\bar Q\in\mathbb{R}^{3\times 6}\) denote the backed-up \emph{second-order} block (consistent with Sec.~\ref{sec:analyticFit}). We form two fixed templates
  \[
    Q_1=\bar Q+\Delta_1,\qquad Q_2=\bar Q+\Delta_2,\qquad
    (\Delta_1,\Delta_2)_{ij}\sim\mathcal{U}[-0.03,0.03].
  \]
  \item Define the compactly supported $C^\infty$ bump
  \[
    b(r;\sigma)=
      \begin{cases}
        \exp\!\bigl(-\tfrac{1}{\,1-(r/\sigma)^2\,}\bigr), & r<\sigma,\\[2pt]
        0, & r\ge \sigma,
      \end{cases}
  \]
  and normalized weights
  \[
    w_i(y)=\frac{b(\|y-c_i\|;\sigma)}{\,b(\|y-c_1\|;\sigma)+b(\|y-c_2\|;\sigma)\,},\quad i=1,2,
  \]
  with \(\sigma=\tfrac12\|c_1-c_2\|\) (centers taken as the min/max-$x$ points of the model).
  \item The deformation for each \(y\in\mathbb{R}^3\) is
  \[
    \tau(y)\;=\;A\,y\;+\;t\;+\;\tfrac12\Big(\,w_1(y)\,Q_1+w_2(y)\,Q_2\,\Big)\,\phi^{[2]}(y-\mathfrak c),
  \]
  where \(\phi^{[2]}(\cdot)\) is the structured quadratic monomial vector used throughout (in 3D:
  \([\,(y_1-\mathfrak c_1)^2,\ 2(y_1-\mathfrak c_1)(y_2-\mathfrak c_2),\ 2(y_1-\mathfrak c_1)(y_3-\mathfrak c_3),\ (y_2-\mathfrak c_2)^2,\ 2(y_2-\mathfrak c_2)(y_3-\mathfrak c_3),\ (y_3-\mathfrak c_3)^2\,]^\top\)).
\end{itemize}

Because the bump \(b\) is $C^\infty$ but \emph{non-analytic} at \(r=\sigma\) (all derivatives vanish at the boundary while no convergent power series equals \(b\) there), the blended field is globally smooth yet non-analytic whenever \(Q_1\not\equiv Q_2\). This yields a heterogeneous, non-analytic test bed. A representative registration example under this deformation model is shown in Fig.~\ref{fig:body4smoothSpace}.

\begin{figure*}[htb]
  \centering
  \scriptsize \textbf{Analytic-ICP under heterogeneous $C^\infty$ non-analytic deformation}
  \vspace{2mm}

  \begin{minipage}[c]{0.32\textwidth}
    \centering
    \includegraphics[width=0.48\linewidth]{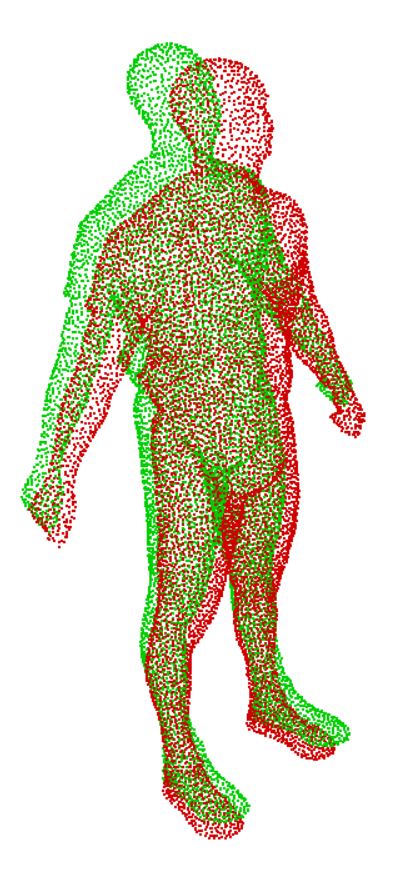}

    \textbf{Original models}\\
    Green: moving\quad Red: fixed
  \end{minipage}
  \hfill
  \begin{minipage}[c]{0.32\textwidth}
    \centering
    \includegraphics[width=0.48\linewidth]{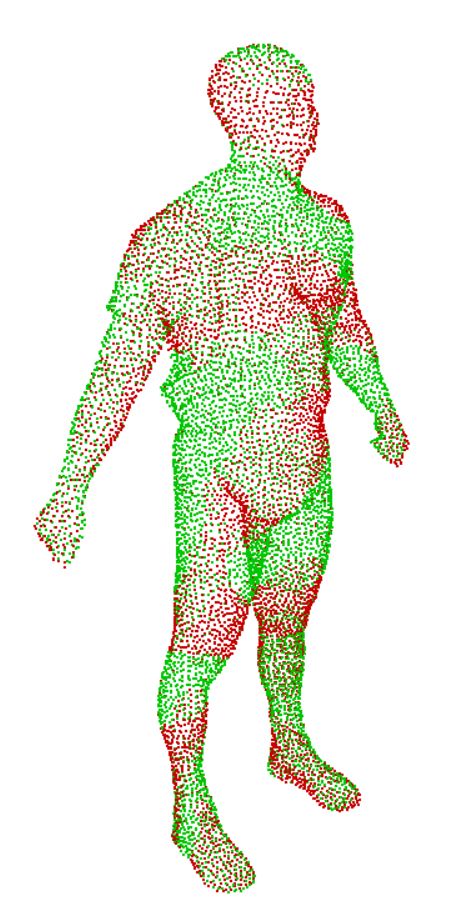}

    \textbf{CPD Result}\\
    RMSE \(= 7.07\times10^{-4}\)\\
    Time \(= 38{,}146.53\) s
  \end{minipage}
  \hfill
  \begin{minipage}[c]{0.32\textwidth}
    \centering
    \includegraphics[width=0.48\linewidth]{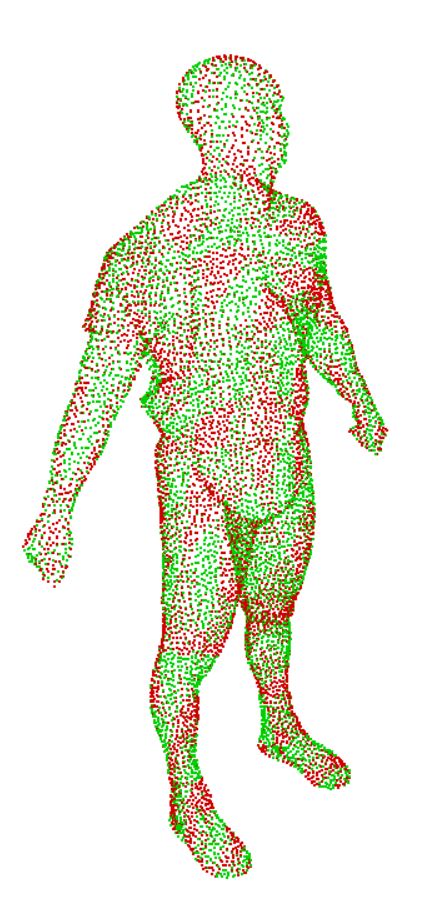}

    \textbf{Analytic-ICP Result}\\
    RMSE \(= 4.69\times10^{-4}\)\\
    Time \(= 7.621\) s
  \end{minipage}

  \caption{\textbf{Robustness to heterogeneous, smooth, non-analytic deformation.}
  The model (10{,}050 points) is perturbed by a two-center, bump-weighted quadratic field.
  On the same hardware, Analytic-ICP achieves lower error with dramatically reduced runtime compared with CPD.}
  \label{fig:body4smoothSpace}
\end{figure*}

\paragraph{Results}
We evaluated Analytic-ICP and CPD on a 3D human body point set with \(\,\mathbf{10{,}050}\,\) points.
After rigid alignment, the initial RMSD was \(\mathbf{0.03575}\).
Under identical parameter settings,
\textbf{Analytic-ICP} achieved \(\mathbf{RMSE=4.69\times10^{-4}}\) in \(\mathbf{7.621}\) s,
whereas \textbf{CPD} reached \(\mathbf{RMSE=7.07\times10^{-4}}\) in \(\mathbf{38{,}146.53}\) s.
Thus, on this heterogeneous $C^\infty$ non-analytic deformation, Analytic-ICP attains lower error with a $\sim$4-orders-of-magnitude speed advantage.

\subsubsection{On the Use of Analytic Deformation in Experiments}
We generate deformations using \emph{truncated multivariate Taylor expansions} consistent with our structured model. By Theorem~\ref{thm:main} and Lemma~\ref{lemma:structured_density_strong}, analytic/Taylor approximants are dense for smooth maps on compact domains; thus residuals under such synthetic deformations are a theoretically sound proxy for real smooth transformations. Compared with ad-hoc benchmarks, this construction affords precise control of magnitude, anisotropy, and nonlinearity by tuning $\mathcal{J}^{[k]}$ and the truncation order, enabling systematic tests of accuracy, conditioning, and convergence. The hierarchy is explicit: first order recovers affine (incl.\ rigid/similarity); higher orders introduce controlled non-affine effects. To our knowledge, this is the first systematic use of \emph{structured} Taylor expansions to synthesize deformation fields for point-set registration, providing a reproducible and principled evaluation space.

\FloatBarrier

\section{Discussion: ICP Correspondences vs.\ Model Capacity}
\label{sec:limitations}

Our strong results for small, smooth deformations reflect a limitation of the \emph{correspondence strategy} rather than the analytic mapping itself. The current instantiation uses ICP; when displacements are large or nonlocal, hard nearest-neighbor matches can be unreliable, causing suboptimal minima~\cite{TPS–RPM}. In contrast, the analytic approximation (AMVFF) is highly expressive: a single closed-form expansion unifies rigid, affine, and higher-order nonrigid terms, and composition of low-order lifts increases the effective order rapidly (Sec.~\ref{sec:analyticFit}).

It is therefore important to distinguish (i) the \textbf{geometric capacity} of AMVFF—able in principle to approximate arbitrary smooth deformations on compact sets—from (ii) the \textbf{practical capture range} of an ICP-style pipeline, which is governed by correspondence quality. Future variants can broaden this range by replacing ICP with robust or probabilistic matches (e.g., RPM/CPD/BCPD) while retaining the same AMVFF fitting core.

\subsection{Implementation and Reproducibility}
\label{sec:code}
We provide source code and scripts to reproduce all figures and tables:
\url{https://github.com/zhenglab/Analytic-ICP}.

\section{Conclusions}
\label{sec:conclusions}

This paper presents a structured approximation model for smooth vector-valued mappings, based on truncated Taylor expansions expressed in a matrix–vector form. To the best of our knowledge, this is the first application of such structured expansions in the context of computer vision and geometric registration.

The model unifies rigid, affine, and non-rigid transformations within a single analytic expression. Supported by the Structured Approximation Theorem and its associated density result, it can approximate any smooth registration mapping on compact domains with controllable complexity. The hierarchical formulation captures affine behavior at low orders and introduces nonlinear deformation through higher-order terms.

We embed this model into an ICP-style algorithm, Analytic-ICP, and demonstrate through experiments that it achieves competitive accuracy with significantly reduced computational cost—particularly as the number of points increases.

The method is well suited for fine registration tasks involving smooth deformations and provides a mathematically grounded alternative to heuristic deformation models. Future work will explore algorithmic acceleration, multi-center (atlas) structured expansions to better capture spatially heterogeneous deformations beyond a single global expansion center, and broader integration into unsupervised registration frameworks.

\appendix
\section{Proof of Theorem~\ref{thm:main}}\label{appendix:proof_structured_taylor}

We provide a complete proof of Theorem~\ref{thm:main}, which establishes the structured Taylor representation for analytic and smooth registration mappings.

\begin{proof}
Let $\tau : W \subset \mathbb{R}^n \to \mathbb{R}^n$ be a smooth (or analytic) registration mapping. Since $\mathbb{R}^n$ is isomorphic to a Cartesian product $\mathbb{R} \times \cdots \times \mathbb{R}$, we can decompose $\tau$ componentwise as:
\[
\tau(y) = (f_1(y), \dots, f_n(y)), \quad \text{with } f_i : W \to \mathbb{R}.
\]

\textbf{(i) Analytic Case.} Suppose $\tau$ is analytic on $W$. Then each $f_i$ is also analytic. By the theory of multivariate analytic functions, each $f_i$ admits a unique convergent Taylor series centered at any $\mathfrak{c} \in W$:
\[
f_i(y) = \sum_{|\alpha| = 0}^{\infty} \frac{D^\alpha f_i(\mathfrak{c})}{\alpha!} (y - \mathfrak{c})^\alpha.
\]
The structured Taylor framework reorganizes this multi-index expansion into a matrix-vector format:
\[
f_i(y) = \sum_{k=0}^{\infty} \frac{1}{k!} \mathcal{J}^{[k]}(f_i)(\mathfrak{c}) \cdot \phi^{[k]}(y - \mathfrak{c}),
\]
where:
- $\phi^{[k]} \in \mathbb{R}^{N_k}$ contains all order-$k$ generalized monomials,
- $\mathcal{J}^{[k]}(f_i) \in \mathbb{R}^{1 \times N_k}$ contains the corresponding mixed partial derivatives,
- and $N_k = \binom{n + k - 1}{n-1}$ is the number of $k$-th order multivariate monomials in $n$ variables.

Since each \( f_i \) admits a unique Taylor expansion, and the matrix-vector structure is preserved under aggregation, the full vector-valued mapping \( \tau \) has a unique representation in the structured Taylor basis:
\[
\tau(y) = \sum_{k=0}^{\infty} \frac{1}{k!} \mathcal{J}^{[k]}(\tau)(\mathfrak{c}) \cdot \phi^{[k]}(y - \mathfrak{c}),
\]
where $\mathcal{J}^{[k]}(\tau) \in \mathbb{R}^{n \times N_k}$ concatenates the rows for all $f_i$.

\textbf{(ii) Smooth Case.} Suppose now that $\tau$ is merely smooth ($C^\infty$). Let $K \subset W$ be a compact set. By the classical multivariate Weierstrass Approximation Theorem (or Whitney Approximation Theorem), for any $\epsilon > 0$, there exists a vector-valued analytic function $\tilde{\tau} : W \to \mathbb{R}^n$ such that:
\[
\sup_{y \in K} \|\tau(y) - \tilde{\tau}(y)\| < \epsilon.
\]
Each component $\tilde{f}_i$ of $\tilde{\tau} = (\tilde{f}_1, \dots, \tilde{f}_n)$ is analytic, and hence possesses a Taylor expansion in the structured basis:
\[
\tilde{f}_i(y) = \sum_{k=0}^{\infty} \frac{1}{k!} \mathcal{J}^{[k]}(\tilde{f}_i)(\mathfrak{c}) \cdot \phi^{[k]}(y - \mathfrak{c}).
\]
Truncating this expansion at order $m$ yields a structured Taylor polynomial that approximates $\tau$ within $\epsilon$ on $K$:
\[
\left\| \tau(y) - \sum_{k=0}^{m} \frac{1}{k!} \mathcal{J}^{[k]}(\tilde{\tau})(\mathfrak{c}) \cdot \phi^{[k]}(y - \mathfrak{c}) \right\| < \epsilon.
\]
Thus, any smooth $\tau$ can be approximated arbitrarily closely (in the uniform norm on $K$) by the structured expansion of an analytic surrogate $\tilde{\tau}$.

\end{proof}

\section*{Acknowledgments}
\funding{This work was supported in part by the National Natural Science Foundation of China (Grant Nos.\ 62571503 and 62171421) and in part by the TaiShan Scholar Youth Expert Program of Shandong Province (Grant No.\ tsqn202306096).}

\bibliographystyle{siamplain}
\bibliography{references}
\end{document}